\documentclass[sigconf,nonacm]{acmart}

\AtBeginDocument{%
  \providecommand\BibTeX{{%
    \normalfont B\kern-0.5em{\scshape i\kern-0.25em b}\kern-0.8em\TeX}}}

\usepackage[utf8]{inputenc}

\usepackage{amsmath,mathtools}
\usepackage{amsthm}
\usepackage{amsfonts}
\usepackage{stmaryrd}
\usepackage{xcolor}
\usepackage{dsfont}
\usepackage{subcaption}
\usepackage{environ}
\usepackage{hyperref}
\usepackage{mathrsfs}
\usepackage{stmaryrd}
\usepackage{xspace}
\usepackage{breqn}
\usepackage{tabularx}
\usepackage{textcomp}
\usepackage[final]{listings}
\usepackage[shortlabels]{enumitem}
\usepackage{graphicx}
\usepackage{algorithm}
\usepackage{algorithmic}
\usepackage{mdframed}
\usepackage{framed}
\usepackage{faktor}
\usepackage{tikz-cd}

\newcommand{\pdom}{X}
\newcommand{\salg}{\Sigma}
\newcommand{\bsalg}{\Sigma_B}

\newcommand{\prob}{P}
\newcommand{\pre}[1]{{#1}^{-1}}

\newcommand{\db}{\mathbf{D}}
\newcommand{\rel}{\mathbf{R}}

\newcommand{\schemrel}{\mathcal{R}}
\newcommand{\schemdb}{\mathcal{S}}

\newcommand{\OMIT}[1]{}

\newtheorem{theorem}{Theorem}
\newtheorem{corollary}[theorem]{Corollary}
\newtheorem{lemma}[theorem]{Lemma}

\newcommand{\arity}{\texttt{arity}}

\newcommand{\Name}{\mathsf{Attr}}

\newcommand{\Num}{\mathbb{R}}
\newcommand{\Nulls}{\mathsf{Null}}
\newcommand{\Vals}{\mathbb{R}}

\newcommand{\dom}{\mathsf{dom}}

\newcommand{\bbR}{\mathbb{R}}
\newcommand{\bbN}{\mathbb{N}}

\newcommand{\cS}{\mathcal{S}}

\newcommand{\bt}{\bar{t}}
\newcommand{\ba}{\bar{a}}

\newcommand{\nulls}{\mathsf{Null}}

\newcommand{\sem}[1]{\llbracket{#1}\rrbracket}

\newcommand{\op}{\,\mathtt{op}\,}

\newcounter{example}[section]
\newenvironment{example}[1][]{\refstepcounter{example}\par\medskip
   \noindent \textbf{Example~\theexample. #1} \rmfamily}{\medskip}

\newcommand{\df}{:=}
\newcommand{\App}{\textsc{Apply}}

\newcommand{\Group}{\textsc{SumAgg}}

\newcommand{\barN}{\bar{N}}

\newcommand{\bara}{\bar{a}}

\newcommand{\brt}{\bar{t}}

\newcommand{\true}{\normalfont\textbf{t}}
\newcommand{\false}{\normalfont\textbf{f}}

\newcommand{\isempty}{\mathtt{empty}}

\newcommand{\all}{\mathtt{all}}
\newcommand{\any}{\mathtt{any}}

\newcommand{\eq}{=}

\newcommand{\sqlra}{\ensuremath{\textsc{RA}^*}}
\usepackage{textcomp}

\newcommand{\bagleft}{ \{\hspace{-3pt} \{ }
\newcommand{\bagright}{  \} \hspace{-3pt} \}}
\newcommand{\bag}[1]{ \bagleft {#1} \bagright }
\newcommand{\set}[1]{\{ {#1} \}}

  \newcommand{\condworld}{\mathcal C}
 \newcommand{\semcondworld}[1]{{#1}({\condworld})}
 \newcommand{\adom}{\mathsf{adom}}
 
 \newcommand{\adb}{\mathbf{A}}

 \newenvironment{repeatresult}[2]
 {\vskip0.5em\par\textsc{#1 #2.}\em}
 {\vskip1em}

 \newcommand{\RAT}{\textrm{RAT}}

\newcommand{\err}{\varepsilon}

\newcommand{\idb}{\db}

\newcommand{\cA}{\mathcal{A}}
\newcommand{\cB}{\mathcal{B}}

\newcommand{\cD}{\mathcal{D}}
\newcommand{\cP}{\mathcal{P}}
\newcommand{\cG}{\mathcal{G}}
\newcommand{\cQ}{\mathcal{Q}}

\newcommand{\cR}{\mathcal{R}}

\newcommand{\cF}{\mathcal{F}}

\newcommand{\cV}{\mathcal{V}}

\newcommand{\ndis}{\mathtt{N}}

\newtheorem{claim}{Claim}

\newcommand{\bc}{\bar{c}}
\newcommand{\be}{\bar{e}}

\newcommand{\bx}{\bar{x}}
\newcommand{\by}{\bar{y}}
\newcommand{\bz}{\bar{z}}
\newcommand{\bw}{\bar{w}}

\begin{document}

\title{Querying Incomplete Numerical Data:\\ Between Certain and Possible Answers}

\author{Marco Console}
\affiliation{%
  \institution{Sapienza, U.~of Rome}
  \city{}
  \country{}
  }
\email{console@diag.uniroma1.it}

\author{Leonid Libkin}
\affiliation{%
  \institution{U.~Edinburgh / RelationalAI / ENS}
  \city{}
  \country{}
  }
\email{libkin@inf.ed.ac.uk}

\author{Liat Peterfreund}
\affiliation{%
  \institution{CNRS / LIGM, U.~Gustave Eiffel}
  \city{}
 \country{}
}
\email{liat.peterfreund@u-pem.fr}

\begin{abstract}
Queries with aggregation and arithmetic operations, as well as incomplete data, are common in real-world database, but we lack a good understanding of how they should interact. On the one hand, systems based on
SQL provide ad-hoc rules for numerical nulls, on the other, theoretical research largely concentrates on the standard notions of certain and possible answers. 
In the presence of numerical attributes and aggregates, however, 
these answers are often meaningless, returning either too little or too much. 
Our goal is to define a principled framework for databases with numerical nulls and answering queries with arithmetic and aggregations over them.

Towards this goal, we assume that missing values in numerical attributes are given by probability distributions associated with marked nulls. This yields a model of probabilistic bag databases  in which tuples are not necessarily independent since nulls can repeat.
We provide a general compositional framework for query answering and then concentrate on queries that resemble standard SQL with arithmetic and aggregation. We show that these queries are measurable,
and  their outputs have a finite representation. Moreover, since the classical forms of answers provide little information in the numerical setting, we look at the probability that numerical values in output tuples belong to specific intervals. Even though their exact computation is intractable, we show 
efficient approximation algorithms to compute such probabilities. 
\end{abstract}
\settopmatter{printfolios=true}

\maketitle
%\tableofcontents

\newcommand{\sqlnull}{\mathtt{Null}}

\section{Introduction}
 
Handling incomplete data is a subject of key importance in databases,
but the practical side of it has many well-known deficiencies
\cite{date2005,datedarwen-sql,date-sigmodr}, and the state-of-the-art
in theoretical research has a rather narrow focus in terms of its
applicability (see, e.g., \cite{pods2020} for a recent survey). In
particular, systems based on SQL follow a rather specific approach to
handling incomplete data where all types of incompleteness are
replaced with a single null value, and specific rules are applied to
those nulls depending on their usage (e.g., 3-valued logic in
selection conditions with {\em true} tuples retained; 3-valued logic
in constraint conditions with {\em non-false} constraints satisfied,
syntactic equality of nulls for grouping and set operations). With
these rules frequently leading to undesired behaviors, many practical
applications attempt to impute incomplete data in a statistically
meaningful way, i.e., resembling the distribution of other data values
in a database \cite{imputation-book}. 
While standard query evaluation
can be used on the now complete data, the knowledge that the original
data was incomplete is lost for the users, and they may take the
answers as correct without the healthy dose of scepticism they
deserve.

There is little that database research can offer to mitigate this
problem, especially for numerical attribute values where imputation is
most commonly used. For them, typical queries use arithmetic
operations, aggregate functions, and numerical constraints. In fact,
queries of this kind are very common in practice and form the absolute
majority of standard benchmarks for database systems (e.g., TPC-H or
TPC-DS). For these queries, though, we have little theoretical
knowledge to rely on.

To explain why this is so, we observe that most of the existing
theoretical research on answering queries over incomplete data
concentrates on the notions of {\em certain} and {\em possible}
answers. The former is often viewed as the holy grail of query
answering, while the latter is often a fallback position in case we do
not have enough certain answers to show.  Both possible and certain
answers, however, are often meaningless for queries with arithmetic
and aggregation. Consider, for example, the following SQL query $q$

\begin{center}
\texttt{SELECT A, SUM(B)
  FROM R GROUP BY A WHERE B}$\ \ge 2$
\end{center}
on a relation $R(A,B)$ with tuples $(1,1)$ and $(1, \sqlnull)$, where
$\sqlnull$ denotes the SQL null value. If all we know about the
attribute $B$ is that it is an integer, then no answer is certain for
$q$ (in fact, as long as there is any uncertainty about the value that
$\sqlnull$ may take, the certain answer is empty). As for possible
answers, then {\em every} tuple $(1,c)$ for $2\le c\in\mathbb{N}$ is a
possible answer, in cases where $\sqlnull$ is interpreted as
$c$ (since the tuple $(1,1)$ is excluded due to the WHERE clause). 
Hence, the usefulness of such answers is limited.  A recent
attempt to provide proper substitutes for these notions was made in
\cite{KR20b} but the shape of answers presented is quite complex for
the user to comprehend. Another approach, specifically focused on
numerical attributes, was presented in \cite{pods20CHL}, but it does
not take into account aggregation and prior knowledge on the missing data.  As
for imputation, it would choose some value for $\sqlnull$, say $2$
based on the other data in the database, and then simply return the
tuple $(1,2)$ as the definitive truth -- without any hint to the user
as to the doubts of the veracity of the answer.% \marco{Edited up to
  % this point}

But what about mixing the probabilistic approach of imputation with
the classical certain/possible answers approach? For example, suppose
that, in the above example, the possible values for $\sqlnull$
are % an instance of a random
% variable
distributed according to the normal distribution with mean $2$ and
standard deviation $0.5$. Then, with probability approximately
$0.682$, the value of the SUM aggregate will fall into the
$[2.5, 3.5]$ interval. Unlike returning the value $(1,3)$, giving no
hint that data was incomplete in first place, this new type of answer
makes it very clear that our knowledge is only partial. Also, it
provides much more information than an empty certain answer (nothing is
certain) and an infinite possible answer (absolutely everything is
possible)

The goal of this paper is to define a model of incomplete numerical
data and a framework for query answering that naturally support
answers like the one above. As a first theoretical model of this kind,
our approach will not cover the entire SQL Standard and all the
numerical attributes in it. However, we shall cover a significant
portion of the query language (essentially, the usual
SELECT-FROM-WHERE-GROUP BY-HAVING queries) and use the real numbers
$\bbR$ as an abstraction for numerical attributes, for the ease of
reasoning about probability distributions.

Our first task will be to define the data model. The example above
hints that we want to treat nulls as random variables, according to
some distribution. In our model, for generality, we allow a
real-valued random variable for each null, % to be associated with
% each null
though, in a more down-to-earth approach, we could associate them with
individual attributes (say, the height of a person that is usually
assumed to follow a normal distribution \cite{height-book}). Thus,
unlike the usual probabilistic databases \cite{probabilistic-book},
our model deals with %probabilistic databases defined by 
{\em
  uncountably infinite} probability
spaces. % , since nulls are some random variables
% over $\bbR$.
Such databases were the focus of recent studies \cite{GroLin22,
  pods21groheetal,sigreg21GroheLindner}, which mainly focused on the
crucial issue of measurability of query mappings.  Those papers,
however, left open a crucial issue: how to construct suitable
probabilistic databases from real-world incomplete numerical data. As
our first task, we fill this gap and provide such a model. Our
construction is general enough so that queries could be viewed as
measurable maps between probability spaces given by infinite
probabilistic databases.

Then, we instantiate our framework to the case of databases with
marked (or repeating) nulls \cite{IL84,AKG91}. The fact that nulls can
repeat in different tuples will take us out of the realm of
tuple-independent probabilistic databases that have dominated the
field of probabilistic data so far. In fact, the need for a model of
``tuple-dependent'' databases even transcends our immediate
goal. Consider, for example, the  query

\begin{center}
\texttt{SELECT *
  FROM R JOIN S ON R.A = S.A}
\end{center}

\noindent 
with $R(A,B)$ having one fact $(1,\sqlnull)$ and $S(A,C)$ containing
$(1,2)$ and $(1,3)$. The end result has tuples $(1,\sqlnull,2)$ and
$(1,\sqlnull,3)$. Since $\sqlnull$ is defined by the same random
variable in both tuples, the probability that a tuple, say
$(1, 1, 2)$, is in the answer depends on the prior knowledge that
another tuple, say $(1, 1, 3)$, is in the answer. %probability $1$, and
%$(1, 2, 3)$, probability $0$. 
Thus, if we want to have
a compositional  query language, i.e., the ability to query the
results of other queries (taken for granted in languages like SQL), then we
cannot rely on the tuple-independent model, as query results will be
not be tuple-independent.

We next instantiate a query language. Following the long line of work
on algebras with aggregates \cite{jacm82klug,
  HLNW01,tcs03libkin,jcss97libkinwang,jcss96grumbachmilo} we provide a
simple yet powerful extension of relational algebra that captures key
features of SQL queries such as arithmetic and aggregation. We then
show that answers for these queries define measurable sets in our
probabilistic model. Thus, a meaningful form of answers for queries in
our extended algebra, and SQL, can be defined in our framework.

With that, we move to query answering. As hinted above, we first look
at computing the probability that values in the output tuples belong
to given intervals. We formulate this intuition as a computation
problem as well as the corresponding decision problem, and study their
complexity.  For the computation problem, we prove that results are
often transcendental, and even for simple distributions, that are
guaranteed to yield rational numbers, the problem is
$\sharp$P-hard \cite{complexity-book}. For the computation problem, we show that it is NP-hard.

This means that one needs to resort to approximations. We first show that the
standard FPRAS (Fully Polynomial Randomized Approximation Scheme
\cite{complexity-book}) approximation does not exist modulo some
complexity-theoretic assumptions. We do show however that, if we allow
additive error guarantees, approximations exist and are computable in
polynomial time. This is a reasonable relaxation since numbers we
compute are probabilities, hence in the interval $[0,1]$.

This development leaves two loose ends with respect to the feasibility
of our approach. First, when a user asks a query, that user expects an
output, not necessarily going tuple by tuple and asking for
probabilities. Can there be a compact representation of the infinite
answer space? We give the positive answer and explain how to construct
finite encodings of infinite probabilistic databases that capture
information in them up to a difference that has probability zero.

Second, our approximation algorithm creates several samples of the
missing data by instantiating nulls in the input database. Even
creating one such sample would be infeasible in real life, as it
amounts to, essentially, creating a copy of the entire input data just
to ask a query over it. To mitigate this issue, we show that sampling
can be done smartly, i.e., encoded in a query that is expressible in
our relational algebra. Thus, the sampling approximation algorithm can
be expressed by a single query ran over the input database.

The remainder of the paper is organized as follows.  We start by
recalling relevant notions in databases and probability theory in
Section~\ref{sec:prelims}. Then, in Section~\ref{sec:framework}, we
present our framework of probabilistic databases and its instantiation
by instances with marked nulls attached with probability distributions. In
Section~\ref{sec:sqlra}, we introduce the query language $\sqlra$ for
databases with numerical entries. Answers for $\sqlra$ in our
framework are then introduced in Section~\ref{sec:qa} where we also
study their measurability and complexity of exact and approximate
computation. We expand on approximation in Section~\ref{sec:apx}, 
where we present an efficient approximation algorithm for $\sqlra$
queries over probabilistic databases and show how it can be encoded within $\sqlra$ queries. 
Finally,
Section~\ref{sec:finiterepoutput} shows that answers for $\sqlra$ queries
can be represented finitely. In
Section~\ref{sec:conclusion} we conclude.

\section{Background}
\label{sec:prelims}

\subsection{Incomplete Databases}

\paragraph{Schemas and instances}
We assume a countable set 
%$\Name$ of attributes \emph{names}, and an infinite disjoint set 
$\schemrel$ of \emph{relation schemas}. 
Each {relation schema} $R\in \schemrel$ has an associated \emph{arity} denoted by $\arity(R)$.
A \emph{schema} $\schemdb$ is a set of relation schemas.

We work with the model of marked nulls, i.e., we assume a countable set $\Nulls\df\{\bot_1,\bot_2,\ldots \}$ of \emph{nulls}. The entries in databases come 
from the set  $\Vals\cup \Nulls$. That is, we assume that numerical values come from 
$\Vals$; since non-numerical values can be enumerated in an arbitrary way, e.g., as $1,2,\ldots$, we assume without loss of  generality that all non-null entries come from 
$\mathbb{R}$.

A \emph{tuple} of arity $k$ is an element in $(\Vals \cup \Nulls)^k$. 
We follow SQL's bag semantics, that is, 
a \emph{relation} $\rel$ over a relation schema $\schemrel$ is a finite multi-set (bag) of tuples with arity $\arity(R)$.
%A relation instance is said to be \emph{incomplete} if it has at least entry in $\Nulls$; Otherwise, it is said to be \emph{complete}.
An \emph{incomplete database} $\db$, or just \emph{database} for short, over a schema $\schemdb$ consists of relations $\mathbf R^{\db}$ over each relation schema $\schemrel$ in $\schemdb$. 
We define $\adom(\db)$ as the set of all entries of $\db$, and $\nulls(\db)$ as all entries of $\db$ that are also in $\Nulls$. 
A database $\db$ is \emph{complete} if all its relations instances are \emph{complete}, that is, do not include nulls; Or, alternatively, if $\nulls(\db) = \emptyset$. 
%To emphasize, we denote complete databases
%by $\db$, and incomplete ones by $\idb$.

A \emph{database query} $q$ over database schema $\cS$ is a
function from databases over $\cS$ into databases over
$\cS'$, for some schema $\cS'$ that consists of a single relation symbol. We write
$q : \cS \to \cS'$ when schemas need to be specified explicitly.

\paragraph{Valuations}
%To transform incomplete instances into complete ones, 
Valuations assign values to nulls. Formally, 
a \emph{valuation} $v:\Nulls \rightarrow \Vals$ is a partial mapping, whose \emph{domain}, i.e., the subset of $\Nulls$ on which it is defined, is denoted by $\dom(v)$. 
%We say that $v$ is \emph{a valuation for $\dom(v)$}.
%whose domain is denoted by $\dom(v) \subseteq \Nulls$.
We lift valuations to databases by
defining $v(\idb)$ as the database that is obtained from $\idb$ by replacing every null entry $\bot$ with $v(\bot)$.

\subsection{Probability Theory}
We present the basic notions of probability we use.

\paragraph{Probability spaces}
A \emph{$\sigma$-algebra} on a set $\pdom$ is a family $\salg$ of subsets of $\pdom$ such that $\pdom\in \salg$ and $\salg$
is closed under complement and countable unions. 
If $\salg'$ is a family of subsets of $\pdom$, then
the $\sigma$-algebra \emph{generated by} $\salg'$, denoted by $\sigma(\salg')$, is the smallest $\sigma$-algebra on $\pdom$ containing $\salg'$. 
A \emph{measurable space} is a pair $(\pdom,\salg)$ with $\pdom$ an arbitrary set and $\salg$ a $\sigma$-algebra on $\pdom$. 
Subsets of $\pdom$ are called \emph{$\salg$-measurable} if they belong to $\salg$. 
A \emph{probability measure} on $(\pdom,\salg)$ is a function $P:\salg \rightarrow[0,1]$ with $P(\emptyset) = 0$, $P(X) = 1$ and $P(\cup_{i=1}^{n} {\sigma}_i ) = \sum_{i=1}^n P({\sigma}_i)$ for every sequence 
${\sigma}_1,\ldots, {\sigma}_n$ of disjoint measurable sets.
A measurable space equipped with a probability measure is called a \emph{probability space}.
When $\pdom \df \mathbb{R}$, we implicitly assume that $\Sigma$ is the Borel $\sigma$-algebra $\bsalg
$ (i.e., the one generated by open intervals $(a,b)$ for $a<b \in \mathbb{R}$), and hence, we only specify the density function that suffices to determine the probability space.

\paragraph{Measurable mappings}
If $(\pdom,\salg)$ and $(\pdom',\salg')$ are measurable spaces, and $\phi: \pdom \rightarrow \pdom'$ is a mapping then we say that $\phi$ is $(\salg,\salg')$-\emph{measurable}
(or simply measurable if $\salg,\salg'$ are clear from context) 
if the preimage $\pre{\phi}$ under $\phi$
of every $\salg'$-measurable set is $\salg$-measurable.

For additional details, we refer the
reader to \cite[Section 5]{probability-book}.

\section{Querying and Instantiating Probabilistic Databases}
\label{sec:framework}
Our model of probabilistic databases describes all possible data
configurations together with a probability measure defining how likely
they are to occur.  This intuition is captured by the following notion.

\begin{definition}
  \label{def:pdb}
  A \emph{probabilistic database (PDB)} $\cD$ over a database schema
  $\schemdb$ is a probability space
whose domain is a set of complete databases over $\schemdb$.
\end{definition}

Notice that, similarly to \cite{GroLin22}, the domain  of a
PDB $\cD$ may be uncountable due to the fact that $\cD$
uses reals as entries. However, in most
models of probabilistic databases known in the literature, one assumes \emph{tuple independence}, i.e., the presence of
each tuple is an independent event.
% from the others. 
This is not the case for the PDBs presented in
Definition~\ref{def:pdb} where no restriction is imposed on the
probability distribution.

\subsection{Querying PDBs}
To define queries for PDBs, we start from the standard notion of
queries for complete databases.  Let $q$ be a query over a schema
$\schemdb$, and let $\cD \df (\pdom_\cD, \salg_\cD, \prob_\cD)$ be a
probabilistic database over $\schemdb$. % where $\pdom_\cD$ is the set
% of all complete databases over $\schemdb$, $\salg_\cD$ is a
% $\sigma$-algebra, and $\prob_\cD$ a probability measure.
For a query
$q$, we write $\pdom_{q,\cD}$ to denote the set
$\{q(\db)~\mid~ \db \in \pdom_\cD\}$ of all possible answers for $q$
over $\pdom_\cD$.  We equip $\pdom_{q,\cD}$ with the $\sigma$-algebra
$\salg_{q,\cD}$ generated by the family
$\{ A \subseteq \pdom_{q,\cD}~\mid~ \pre{q}(A) \in \salg_\cD\} $.
% ,
%i.e.,
%$\salg_{q,\cD}\df \sigma \{ A \subseteq \pdom_{q,\cD}~\mid~ \pre{q}(A)
%\in \pdom_\cD\}$. 
In other words, $\salg_{q,\cD}$ is the $\sigma$-algebra generated by
those images of $q$ that are measurable in $\cD$. For each element
$\sigma$ of $\salg_{q, \cD}$, we set
$\prob_{q,\cD}(\sigma) \df \prob_\cD(\pre{q}(\sigma))$, provided the latter is
defined.

\sloppy
\begin{proposition}
  \label{pr:query-measurability}
  Given a PDB $\cD$ over a schema $\schemdb$, and a query $q$ over
  $\schemdb$, 
  \begin{itemize}
      \item $q$ is $(\salg_{\cD},\salg_{q,\cD})$-measurable; and 
      \item $\prob_{q,\cD}$ is defined over every element of $\salg_{q, \cD}$ and is a probability measure for
        $(\pdom_{q,\cD}, \salg_{q,\cD})$.
  \end{itemize}
  \end{proposition}

In view of Proposition \ref{pr:query-measurability}, we can define the
probability space for the answers to $q$ over $\cD$ as
follows.

\begin{definition}\label{def:ansspace}
  The \emph{answer space of $q$ over $\cD$} is defined as the probabilistic database
  $q(\cD) \df (\pdom_{q,\cD}, \salg_{q,\cD}, \prob_{q,\cD})$.
\end{definition}

We point out next  
that, contrary to tuple-independent probabilistic databases, our model
preserves composability of query languages.\footnote{Note that even the simplest operations, e.g., computing union of a relation with itself, results in a violation of tuple-independence.}  
We say that a pair $(q,q')$ of
queries is \emph{composable} if
$q : \schemdb \to \schemdb'$ and $q' : \schemdb' \to
\schemdb''$. In this case, we define their \emph{composition} $(q' \circ q)
%\liat{I changed the order of the notation to make it less confusing, you can revert if you object}
(\db) \df q'\left( q (\db) \right)$.  
Can we lift composability to PDBs?
%Is composability of
%queries preserved in our probabilistic model?  The answer % to this
% question
%is positive as the following proposition shows.
The following proposition shows that the answer to this question is positive.
\begin{proposition}
  \label{pr:composability}
  Let %$q_1 : \schemdb_1 \to \schemdb_2$ and
  %$q_2 : \schemdb_2 \to \schemdb_3$ 
  $(q,q')$ be composable queries. Then  $(q' \circ q)(\cD) = q'(q(\cD))$ for every PDB $\cD$ and query $q$ for $\cD$.
\end{proposition}

Note that composability is a consequence of the  answer space definition, and is independent of the query language.
\OMIT{
A family $\cQ$ of queries is closed under composition if, for every
pair of composable queries $(q,q')$ from $\cQ$, also $q' \circ q$ is
in $\cQ$. Note that \sqlra, when restricted to complete instances, is
indeed closed under composition.
A further consequence of Proposition~\ref{pr:composability} is that
closure under composition is preserved over PDBs. In fact, if $\cQ$ is
closed under composition, for every composable $(q, q')$ in $\cQ$,
$q' \circ q \in \cQ$ defines the composition of $(q, q')$ also over
PDBs.
}

\subsection{PDBs from Incomplete Databases}
% Our definition of PDBs is
% very general. 
% So general, in fact, that we may not be able to
% represent probabilistic databases in the model in a finite and compact
% form. 
% \liat{add a ref to grohe's saying}
% In order to define databases that can actually arise from
% practical use cases, we specialize our model to databases defined via
% valuations of nulls.

% To this end, we first enrich our model of incomplete instances with
% marked nulls with probabilities. Intuitively, each null symbol in the
% model comes with a real random variable that defines how likely it is
% for the null to take specific values.

We now show how the PDBs we introduced in Definition~\ref{def:pdb}
can be instantiated naturally via incomplete databases.

To associate an incomplete database $\idb$ with a PDB, we  equip each
$\bot_i \in \nulls$
with a probability density function $P_i$. 

We assume without loss of generality that $\nulls(\idb) \df \{\bot_1,\ldots, \bot_n \}$, and
we then define the
%For an incomplete database $\idb$ we define 
\emph{measurable space of valuations} $(\pdom_V, \salg_V)$ of
$\nulls(\idb)$ such that $\pdom_V$ is the set of all valuations $v$
with $\dom(v) = \nulls(\idb)$,
%that are compatible with $\nulls(\idb)$
 and $\salg_V$ is the $\sigma$-algebra generated
by the family
 % \[\{V(\sigma_1, \ldots, \sigma_n)~\mid~\sigma_i \in \bsalg, \textit{
 %     each } i = 1 \ldots, n\}\]
consisting of the sets
 \[
  V(\sigma_1, \ldots, \sigma_n) \df \{v~\mid~ v(\bot_i)\in \sigma_i  \}
 \]
 for each $\sigma_1,\ldots, \sigma_n$  in the Borel $\sigma$-algebra
 $\bsalg$.  Notice that the range of the valuations of each such set
 is a Cartesian product of sets in $\bsalg$. The events defined by
 these sets, i.e., $v(\bot_i) \in \sigma_i$, for each
 $i = 1, \ldots, n$, are assumed to be mutually independent.

\OMIT{
To associate an incomplete database $\idb$ with a PDB, we  equip each
$\bot_i \in \nulls$
with a probability density function $P_i$. 
We then define the
%For an incomplete database $\idb$ we define 
\emph{measurable space of valuations} $(\pdom_V, \salg_V)$ of
$\nulls(\idb)$ such that $\pdom_V$ is the set of all valuations $v$
with $\dom(v) = \nulls(\idb)$
%that are compatible with $\nulls(\idb)$
, and $\salg_V$ is the $\sigma$-algebra generated
by the family
 % \[\{V(\sigma_1, \ldots, \sigma_n)~\mid~\sigma_i \in \bsalg, \textit{
 %     each } i = 1 \ldots, n\}\]
consisting of the sets
 \[
  V(\sigma_1, \ldots, \sigma_n) \df \{v~\mid~ v(\bot_i)\in \sigma_i  \}
 \]
 for each $\sigma_1,\ldots, \sigma_n$ is in the Borel $\sigma$-algebra
 $\bsalg$.  Notice that the range of the valuations of each such set
 is a Cartesian product of sets in $\bsalg$. The events defined by
 these sets, i.e., $v(\bot_i) \in \sigma_i$, for each
 $i = 1, \ldots, n$, are assumed to be mutually independent.
}

To reflect this, we would like to define a probability function $\prob_V$ for which the following requirement holds:
\begin{equation}
  \label{eq:product}
  \prob_V(V(\sigma_{1}, \ldots, \sigma_{n})) = \Pi_{i=1}^n \prob_{i}(\sigma_{i})
\end{equation}
As it turns out, this requirement 
%in Equation~\ref{eq:product} is enough 
%to define a
determines
a unique probability measure for $(\pdom_V, \salg_V)$.
\begin{theorem}
  \label{th:val-space-unique}
  There exists a unique probability measure $\prob_V$ for the
  measurable space of valuations of $\nulls(\idb)$ that satisfies
  Equation~(\ref{eq:product}).
\end{theorem}
From now on, we refer to $(\pdom_V, \salg_V, \prob_V)$ as the
%\emph{probability space of valuations} (or simply 
\emph{valuation space} of
$\nulls(\idb)$.
%
 %\subsubsection{Probabilistic Databases From Valuations}
%%%%%%
%Intuitively, the valuation space of $\idb$ defines how likely specific
%valuations are according to the random variables defining
%$\nulls(\idb)$. 
The valuation space, however, does not define a PDB, since multiple
valuations may define the same complete database.\footnote{For
  example, given a Unary relation that consists of two nulls, for any
  valuation there is a symmetric one (that replaces one nulls with the
  other) that defines the same database.}
%To overcome this gap, we present a way to transform
%valuation spaces into database spaces.%  based on the notion of
% \emph{quotient measurable spaces}.

% A \emph{probabilistic database} $\cD$ for a schema $\schemdb$ is a
% probability space $(\pdom_\cD, \salg_\cD, \prob_\cD)$ where
% $\pdom _\cD$ is a set of databases over $\schemdb$.
% Intuitively, given $\idb$ as above, our goal is to define a
% probabilistic database that behaves according to the valuation space
% of $(\pdom_V, \salg_V, \prob_V)$ of $\idb$.
%
To overcome this, we define the PDB of $\idb$ as follows. Let $\pdom_\idb$ be
the set of all possible instantiations of $\idb$, i.e.,\sloppy{
  $\pdom_\idb \df \{ v(\idb) ~\mid~% v %\textit{ is a valuation for }
 \dom(v) = \nulls(\idb) \}$,}
and let 
%We define the mapping 
$\chi : \pdom_V \to \pdom_\idb$ be the mapping such that
$\chi(v) \df v(\idb)$. We set
\[\salg_\idb \df \{A ~\mid~ \pre{\chi}(A) \in \salg_V\}.\] 
%In simple words,
%$\salg_\idb$ contains all the sets of databases whose preimage under
%$\chi$ is measurable in $\salg_\cV$. % the valuation space of $\idb$.

\begin{proposition}
  \label{pr:sigma-db}
  $\salg_\idb$ is a $\sigma$-algebra on $\pdom_\idb$ for every incomplete database $\idb$.
\end{proposition}

With Proposition~\ref{pr:sigma-db} in place, 
%we define the
%\emph{measurable space of databases} for $\idb$ as the measurable
%space $(\pdom_\idb, \salg_\idb)$. To define a probability measure for
we can move to define a probability measure for  the measurable
space $(\pdom_\idb, \salg_\idb)$.
%such space, we use the probability defined for the valuation space of
%$\idb$. 
To this end, we define $\prob_\idb : \salg_\idb \to [0,1]$ %be the function 
%such
%that 
by setting, for each
$A \in \salg_\idb$, 
\[\prob_\idb(A) \df \prob_V(\pre{\chi}(A))\]  
%such that $\pre{\chi}(A) \in \salg_V$.
\begin{proposition}
  \label{pr:prob-dbs}
$\prob_\idb$ is a probability measure for  $(\pdom_\idb, \salg_\idb)$ 
%the
%  measurable space of databases for $\idb$.
for every incomplete database $\idb$.
\end{proposition}
As a consequence of the previous results, we now have a well-defined concept of the probabilistic space of $\idb$. We record this in the following definition.
\begin{definition} \label{def:pdbofidb}
  The \emph{PDB of $\idb$} is $(\pdom_\idb, \salg_\idb, \prob_\idb)$.
\end{definition}
With a slight abuse of notation, from now on, we shall sometimes refer to the PDB of $\idb$, simply as $\idb$.

\section{Query Language: Extended Relational Algebra}\label{sec:sqlra}

\newcommand{\att}[1]{\$#1}

\newcommand{\figterms}{
\[t \ \df \  \, c \, |\, \att{i} \, |\,  t + t \,\mid\, t - t \,\mid\, t \cdot t \,\mid\, t \div t, 
\,\mid\,
\, |\, \exp{(t)} \, |\, \log{(t)} 
|\,f(t_1,\cdots,t_k)
\,\,\,\,\,\,\, c \in \Num
\circ \in \{ +,-,\times, \div 
\}
\]
%where $
%c \in \Num,\, \, i\in \mathbb{N}^+, \,\,
%\circ \in \{ +,-,\times, \div \}
%,\,\, f\in \{ \}
%$
}
\newcommand{\figcond}{
%\textbf{
%Atomic 
%Conditions:}
\[\theta\, \df \,
%\true \, | \, \false \, | \,
 %\isnul(t)  \, | \,
  \att{i} \eq    \att{j} 
 \,\, \, | \ \ \,
 % \att{i} \eq   c
 %\,\, \, | \ \ \,
     \att{i} < \att{j}
 %    \,\, \, | \ \ \,
  %  \att{i} < c
   %  \,\, \, | \ \ \,
%    c < \att{i}
 %\brt\in q \, | \,  
 %\isempty(q) 
%\, | \, 
%\brt\, \omega\, \any(q)\, | \, 
%\brt\, \omega\, \all(q) 
%\text{ s.t. }
%\omega\in\{<,
%>,\le,\ge,
%\eq\}
%,\neq 
\] 
%where $\omega\in\{<,
%>,\le,\ge,
%\eq
%,\neq 
%\}$\\

%\textbf{Composite Conditions:}
%\[
%\theta \df  \theta \vee \theta % \, | \, \neg \theta \, |  \, \theta \wedge \theta .
%\]
}
\newcommand{\oldfigqueries}{\noindent
\begin{tabular}{ rll} 
		$q:=$ & $R$ & (base relation) \\ 
		 & $\pi_{\att{i_1},\ldots,\att{i_k}}(q)$ & (projection)
	\\ 
			 & 
			 %$\pi_{ \att{i_1},\ldots,\att{i_k}, t }(q)$ 
			$ \App_t(q)$
			 & (function application)
	\\ 
		 & $\sigma_{\theta}(q)$ & (selection)  \\ 
 		 & $q\times q$ & (product)  \\ 
 		 & $q\cup q$ & (union)  \\ 
 		 %& $q \cap q$ & (intersection)  \\ 
 		 & $q \setminus q$ & (difference)  \\ 
 	 %	 & $\epsilon(q)$ & (duplicate elimination)  \\ 
 	 	 & $
 	 	\Group_{\att{i_1},\ldots,\att{i_k}, 
 	 	%\langle 
 	 	\sum (\att{i_{k+1}}) %\rangle
 	 	}(q)$
 	 	 %\langle F_1(\att{j_1})
 	 	 %,\cdots, F_m(\att{j_m})
 	 	 %\rangle}(q)
 	 	 %$
 	 	 & (grouping / summation) 
 	 	 %with optional renaming
 	 	 %)
	\end{tabular}
}

\newcommand{\previousfigqueries}{
\centering
		$\begin{array}{rcl} q & := & 
		    R \ \mid\ 
		 \pi_{\att{i_1},\ldots,\att{i_k}}(q) \ \mid\ 
 \App_f(q) \ \mid\ 
 \sigma_{\theta}(q)   \ \mid\ 
 	\sum_{\att{i_1},\ldots,\att{i_k}}^{\att j} 
 	 	(q)  
 \\ && \multicolumn{1}{r}{ \ \mid\  
 q\times q \ \mid\ 
 q\cup q \ \mid\ 
 q \setminus q }
 \\
 \theta & := & \$i = \$j \ \mid\ \$i < \$j \\
 f & \in & \mathbb{R}[\$1,\$2,\ldots]
 \end{array}$
}

\newcommand{\figqueries}{
\centering
		$\begin{array}{rcl} q & := & 
		    R \ \mid\ 
		 \pi_{\att{i_1},\ldots,\att{i_k}}(q) \ \mid\ 
 \sigma_{\theta}(q)   \ \mid\ q\times q \ \mid\ 
 q\cup q \ \mid\ 
 q \setminus q 
 	 \\ && \multicolumn{1}{r}{ \ \mid\  \App_f(q) \ \mid\ 
 \sum_{\att{i_1},\ldots,\att{i_k}}^{\att j} 
 	 	(q)  
 }
 \\
 \theta & := & \$i = \$j \ \mid\ \$i < \$j \\
 f & \in & \RAT[\$1,\$2,\ldots],  \quad k\ge 0
 \end{array}$
}

\begin{figure}
\fbox{\figqueries}
    \caption{\sqlra\ Queries}
    \label{fig:queries}
\end{figure}

To analyze our framework, we need to focus on 
a concrete query language.
Our choice is a compact language that has sufficient power to express standard SQL queries with arithmetic and aggregation, similar to~\cite{tcs03libkin}.
This language, \sqlra\, is an extended form of relation algebra
and its
syntax is  presented in Figure~\ref{fig:queries}. Due to space limitations, the full formal semantics is in the Appendix; here we introduce the main constructs and explain the expressiveness of the language. 

Like SQL, the queries of 
\sqlra\ are interpreted under 
\emph{bag semantics}, i.e., relations may have multiple occurrences of the same tuple.

The first line of the Figure~\ref{fig:queries} contains the standard operations of Relational Algebra. We use the unnamed perspective \cite{AHV95} and refer to attributes by their positions in the relation, i.e., $\$1,\$2$, etc, for the first, second, etc attribute of a relation.
Projection and selection keep duplicates; Cartesian product $\times$ multiplies occurrences, $\cup$ as SQL's UNION ALL adds them, and $\setminus$ as SQL's EXCEPT ALL subtracts the number of occurrences, as long as it is non-negative. 

The operation $\App_f(q)$ takes a function $f\in\RAT[{\att 1},{\att 2},\ldots]$, which is a rational function over variables $\att{i}$, $i \in \mathbb{N}$, i.e., a ratio of two
polynomials with real coefficients. It then adds a column to a relation with the result of this function for each tuple in the result of $q$. For example, if $f(\$1,\$2)=(\$1)^2/\$2$, then $\App_f(R)$ turns a binary relation $R$ into a ternary relation that consists of tuples $(a,b,a^2/b)$ for every occurrence of $(a,b)$ in $R$. 
%$\bag{(a,b,a^2/b)\mid (a,b) \in R}$.

The operation $\sum_{\att{i_1},\ldots,\att{i_k}}^{\att j}(q)$ is grouping with summation aggregation; intuitively, this is SQL's SELECT ${\att{i_1},\ldots,\att{i_k}}$, SUM($\att j$) FROM $q$ GROUP BY ${\att{i_1},\ldots,\att{i_k}}$. When $k=0$, we group by the empty set, thus obtaining a query that returns a unique value who is the sum of all values in column $j$.

While we have chosen a rather minimalistic language for the convenience of presentation and proofs, \sqlra\ actually packs considerable expressive power. For example, it can express all of the following:
\begin{itemize}
    \item Arbitrary conditions $f \ \omega\ g$ where $\omega\in\{=,\neq,<,>,\leq,\geq\}$ and $f,g$ are polynomials in variables $\att 1, \att 2$, etc. Indeed values of $f,g$ can be attached as new attributes, and those can be compared with $=$ and $<$ comparisons; other comparisons are expressed by means of set operations. 
    \item All SQL aggregates MIN, MAX, AVG, SUM, COUNT. Indeed, MIN and MAX can be expressed in standard relational algebra. To express count, we can attach an attribute with constant value $1$ and add up those values. For example, SQL's SELECT A, COUNT(B) FROM R GROUP BY A for a two-attribute relation  is expressed as $\sum_{{\att 1}}^{{\att 3}}\big(\App_1(R)\big)$.
    To express multiple aggregates, one needs to take the product of the input relation with the result of the first aggregation to restore duplicates, and then perform another aggregation. And finally average can be expressed as the ratio of SUM and COUNT.
    \item Duplicate elimination is encoded in grouping: one needs to add an aggregate and project out grouping attributes. For example, to eliminate duplicates from a single-attribute relation $R$, one would write $\pi_{\att 1}\Big(\sum_{\att 1}^{\att 2}\big(\App_1(R)\big)\Big)$.
    \item Conditions involving subqueries, as in IN and EXISTS conditions used in SQL's WHERE. This is already doable in relational algebra \cite{sql2ra,vldb17GL}. 
\end{itemize}

Thus, we see that despite its simplicity, the language is significantly expressive, as it captures the essence of SELECT-FROM-WHERE-GROUP BY-HAVING queries of SQL.

\newcommand{\out}{\mathsf{out}}
\newcommand{\FO}{\mathtt{FO}}
\newcommand{\VAR}{\mathsf{Var}}
\newcommand{\spc}{\mathtt{SPC}}

\newcommand{\otlname}{\textsc{Likelihood}}
\newcommand{\otldef}[1]{\otlname[#1]}
\newcommand{\otlprob}{\otldef{q,\circ, \mathtt k}}

\newcommand{\otlunifdef}[1]{\otlname_{int}[#1]}
\newcommand{\otlunif}{\otlunifdef{q,\circ, \mathtt k}}

\newcommand{\totlname}{\textsc{Threshold}}
\newcommand{\totldef}[1]{\totlname[#1]}
\newcommand{\totlprob}{\totldef{q,\circ, \mathtt k, \delta}}

\newcommand{\totlunifdef}[1]{\totlname_{int}[#1]}
\newcommand{\totlunif}{\totlunifdef{q,\circ, \mathtt k, \delta}}

\section{Query Answering}
\label{sec:qa}

With the query language in place, we now turn to query answering.
To do so, we  focus on specific subsets of the answer space; essentially we ask, as was outlined in the introduction, whether values in output tuples belong to specific intervals. We want to know the probability of such events. In this section we show that such probabilities are well defined, and look at the complexity of their exact computation (which will be high, thus leading us to approximation algorithms in the following section).

\subsection{Target Sets of the Answer Space}

As we already explained, when dealing with numerical values, returning certain or possible tuples of
constants is not very informative. Instead we want to settle for providing {\em intervals} of values.
These intervals may be fixed, say $[1,2]$, or even defined by functions on nulls in a database, say $[\$1+3,\$1+4]$.

Formally, we define \emph{interval tuples} as tuples $\bar a$ whose entries are open, or closed, or half-open half-closed intervals of one of the forms $(x,y), [x,y], [x,y),(x,y]$ where $x,y\in \RAT[\bot_1,\bot_2,\ldots] \cup \set{\pm\infty}$. We denote by $\nulls(\bar a) $ the set of all nulls that appear in the entries of $\bar a $, and by $v(\bar a)$ the tuple obtained from $\bar a$ by replacing each $\bot_i$ with $v(\bot_i)$. In such a tuple, all intervals are grounded, i.e. their endpoints are real numbers.

We say that a tuple of constants $\bar t \df (t_1,\ldots, t_n)$ is \emph{consistent } with  $v (\bar a)$ if  $t_i \in v(a_i)$ for each $i$, that is, $t_i$ is a value within the interval $v(a_i)$.
% by a slight abuse of notation, 
For the rest of the section we use the symbol $\circ$ to denote one of the usual comparisons, i.e.,  $<, =$, or $>$. We write
$\#(v (\ba),A) \circ \mathtt k$ when the number of tuples in the bag $A$ that are consistent with $v(\ba)$ is 
$\circ \mathtt k$. % greater than, equal to, or
% lesser than $k$.
With this, we can move to the following key definition.
\begin{definition}
\[
 \out_{q,\idb, \circ}{(\ba, \mathtt k)} \df \left\{ q(v(\idb)) \in \pdom_{q,\idb}\,  \middle\vert \begin{array}{l}
   \dom(v) = \Nulls(\idb), \quad \\
     \#(v (\ba), q(v(\idb))) \circ \mathtt k
  \end{array}\right\}
\]
\end{definition}
Intuitively, $\out_{q,\idb, \circ}{(\ba, \mathtt k)}$
 is the set that consists of all those databases in $\pdom_{q,\idb}$ such  that each contains
$\circ \mathtt k$ tuples consistent with $\ba$. %%%%%%%
Note that $\pdom_{q,\idb}$ is the domain  of the answer space of $q$ over $\idb$ as defined in Definition~\ref{def:ansspace}.
We call $\out_{q,\idb, \circ}(\cdot\ , \cdot)$ the 
 \emph{target mapping} for $q$ over $\idb$. Target mappings capture information about the query answer space, and we next see how to compute them.

 \subsection{Measurability of Target Sets}
The target mapping can provide us with information about the answer space.
\OMIT{ To answer $q$ over $\idb$, one can use the target sets connected to
 some relevant combination of $\ba$ and $k$ and compute their
 probability.}
In particular, we can compute how likely it is that tuples consistent with $\ba$ appear in the answer (at most, at least or exactly) $\mathtt k$ times. 
This is captured by the following computational problem.

 \begin{framed}
\begin{center}
  \begin{tabular}{ll}
    \multicolumn{2}{c}{ $\otlprob$ }\\
    \textsc{Parameters}:  &  A query $q\in\sqlra$,  $\mathtt k \ge 0$, \\ &  and $\circ \in \{<, =, >\}$ \\
    \textsc{Input}: & An incomplete database $\idb$, \\ & and an interval tuple $\ba$  \\
    \textsc{Problem}: &  Compute $P_{q,\idb}( 
                        \out_{q,\idb, \circ}{(\ba, \mathtt k)}
                        )$\\
  \end{tabular}
  \end{center}
\end{framed}

Is this problem even well-defined? In other words, is it always the case that \sloppy{
$P_{q,\idb}( \out_{q,\idb, \circ}{(\ba, \mathtt k)} )$ exists? }
The answer is
not immediate since not every subset of $\pdom_{q, \idb}$ is
measurable. We settle this issue with the following claim.
    
\begin{theorem}
\label{th:measurability}
For every $q \in \sqlra$, incomplete database $\idb$, interval
tuple $\ba$, $\mathtt k \ge 0$, and $\circ \in \{<, =, >\}$, it holds that
\[\out_{q,\idb, \circ}{(\ba, \mathtt k)}\in \salg_{q,\idb}\]
%
%This remains true when replacing $\ba$ with an interval tuple $\ba^*$.
\end{theorem}

To obtain Theorem~\ref{th:measurability}, we first observe that
$\out_{q,\idb, \circ}{(\ba,\mathtt k)}$ is the union of some countable family
of sets $\out_{q,\idb, =}{(\ba, \mathtt k)}$.
%
% study the sets $O_p = \{ v ~\mid~ \#(v (\ba), q(v\idb)) = p \}$,
% with $p \in \bbN$, that consists of those valuations $v$ for $\idb$
% for which $q(v(\idb))$ contains exactly $k$ tuples consistent with
% $v(\ba)$.
Since each such set is definable by a first-order formula in the
theory of the reals with $|\nulls(\idb)| = n$ free variables, we can
conclude that each set is the finite union of open sets in $\bbR^n$,
which are Borel-measurable $\bbR^n$, and sets of smaller dimensions
which have measure zero in $\bbR^n$ (\cite{ominimal-book}).  The
claim follows from the fact that
% $\out_{q,\idb}{(\ba, k)}$ is equal to the contuable union of the
% family $\{O_p~\mid~ p>k \in \bbN\}$, and the fact that
the PDB defined by $\idb$ is Borel-isomorphic to $\bbR^n$, that is,
there exists a bijection between the two spaces (see,
e.g.,~\cite[Section 1]{durrett2019probability} for more details).

\subsection{Computing $\otlprob$ Exactly}

% Now that we know that the $\otlprob$ is well
% defined, we study its complexity.
Solving the computational problem $\otlprob$ exactly is a very challenging task that is
unfeasible in practice even for simple queries. Let $\spc$ be the
fragment of $\sqlra$ that allows only basic relations, selection,
projection, and Cartesian product. We can prove that the $\otlprob$ is
not rational even for queries in $\spc$ and nulls defined by basic
distributions. An \emph{exponentially distributed} database is an
incomplete database whose null values are annotated with exponential probability 
distributions $\lambda e^{-\lambda x}$, where
$\lambda$ is rational.

\begin{proposition}
  \label{pr:transc}
  For each $\circ \in \{<, =, >\}$, there exists a Boolean query
  $q \in \spc$ and a family of exponentially distributed databases
  $\{\idb_{\mathtt k}~\mid~\mathtt{k}>0 \in \bbN\}$ such that $\otlprob(\idb_{\mathtt k}, ())$ is a
  transcendental number.
\end{proposition}

In other words, Proposition~\ref{pr:transc} implies that
it may not be possible to even write down the exact output of $\otlprob(\idb_{\mathtt k}, ())$.
 Is it due to the density functions we chose?
It turns out, that even
when we consider very simple and well-behaved distributions the
problem remains computationally challenging. %In particular, we can
%show that $\otlprob$ is hard in the counting class $\# P$ (see, e.g.,
%\cite{complexity-book}) that is at least as hard as the whole
%polynomial hierarchy (\cite{toda-pp}). The latter remain true even if
%we only allow very simple distributions for the nulls.  
Formally, a
\emph{uniform interval database} is an incomplete database whose nulls are
annotated  with uniform interval distributions with rational
parameters.  That is, their density function is a rational constant $1/(u-l)$ over an interval $[l, u]$, with
$l < u$ both rational, and $0$ everywhere else. We denote by $\otlunif$ the
restriction of $\otlprob$ to input databases from the class of uniform interval databases.

 \begin{theorem}
\label{th:shp-hardness}
For each $\circ \in \{<, =, >\}$, there exists a Boolean query
$q \in \spc$ such that $\otlunif$ is $\# P$-hard, for each
$\mathtt k\ge 0$.
\end{theorem}
We recall that hard problems in the class $\# P$
(\cite{complexity-book}) are at least as hard as the whole polynomial
hierarchy \cite{toda-pp}.
 From a practical perspective,
it is often enough to check whether $\otlprob$ exceeds a given
threshold. We thus define the decision
version of $\otlprob$:
 \begin{framed}
\begin{center}
  \begin{tabular}{ll}
    \multicolumn{2}{c}{ $\totlprob$}\\
    \textsc{Parameters}:  &  A query $q\in\sqlra$,  $\mathtt k\ge0$,\\
                          &  a \emph{threshold $0 \le \delta \le 1$}, and $\circ \in \{<, =, >\}$ \\
    \textsc{Input}: & An incomplete database $\idb$,  \\
                          &and a interval tuple $\ba$\\
    \textsc{Problem}: &  Decide whether $P_{q,\idb}( 
                        \out_{q,\idb,\circ}{(\ba, \mathtt k)}
                        )> \delta $\\
  \end{tabular}
  \end{center}
\end{framed}
The case $\totldef{q, <, 0, \delta}$ is trivial, and it will not be considered in what follows.
We call $\totlunif$ the restriction of $\totlprob$ to uniform interval
databases only. Even this problem is not simple to solve, as the
following theorem suggests.

\begin{theorem}
\label{th:thresholdnphardnew}
For each $\circ \in \{<, =, >\}$, there exists a Boolean query
$q\in \sqlra$ such that $\totlunifdef{q, \circ, \mathtt{k}, 0}$ is NP-hard.
\end{theorem}

This has
consequences also on 
approximation algorithms for $\otlprob$. %Specifically, assume a function problem
%$P$ that asks for the value of a function $f : X \to Y$. 
A
\emph{Fully-Polynomial Randomised Approximation Scheme} (FPRAS) for $f : X \to Y$~\cite{complexity-book} is an algorithm $A$ with the following guarantees, for
input $x \in X$ and $\err \in (0,1]$, :
\begin{itemize}
\item $A$ runs in time polynomial in the size of $x$ and
  $\frac{1}{\err}$; and
\item $A$ returns a random variable $A(x, \err)$ such that
  \[P(|A(x, \err) - f(x)| \le \err \cdot \left|f(x)\right|) \ge 0.75.\]
\end{itemize}
Of course 0.75 can be replaced by any constant > 0.5; here we follow the traditional definition of 0.75 as the chosen constant.

\begin{corollary}
  \label{cr:no-fpras}
  Unless $RP = NP$, there exists $q\in\sqlra$ such $\otlunif$ admits no FPRAS.
\end{corollary}
The class $RP$ consists of decision problems that can be solved
efficiently by a randomized algorithm with a bounded one-sided error
see, e.g., \cite{complexity-book}. It is commonly believed that
$RP \neq NP$, since, otherwise, problems in $NP$ would admit efficient
practical solutions.

Are there other approximation schemes that can lead to better results? We answer this in the next section.

\newcommand{\alg}[1]{\mathtt{#1}}
\newcommand{\vsalg}{\alg{ValSampler[\idb]}}
\newcommand{\dsalg}{\alg{Direct\textit{-}Sample[q]}}
\newcommand{\daalg}{\alg{LikeApx[\mathit{q},\circ, k]}}

\newcommand{\nsampalg}[1]{\alg{Sample_{#1}}}

\section{Approximation}
\label{sec:apx}

In this section, we present an efficient approximation scheme with
additive error for $\otlprob$, hence answering positively the question
posed at the end of the previous section.

To deal with the numerical nature of our problem, in what follows we
assume an extension of the RAM model of computation with the following
assumptions: values in $\bbR$ are represented with arbitrary precision
in one single register, and basic arithmetic operations
($+, -, \cdot, \frac{\cdot}{\cdot}$) are performed in constant
time. Note that this model 
%of computation 
is commonly used for numerical
problems (e.g.,~\cite{blum1998complexity}).

An \emph{Additive Error Fully-Polynomial Randomised
  Approximation Scheme} (AFPRAS) for $f : X \to Y$ is an algorithm $A$ with the following guarantees, for
input $x \in X$ and $\err \in (0,1]$, :
\begin{itemize}
\item $A$ runs in time polynomial in the size of $x$ and
  $\frac{1}{\err}$; and
\item $A$ returns a random variable $A(x, \err)$ such that
  \[P(|A(x, \err) - f(x)| \le \err) \ge 0.75.\]
\end{itemize}
The difference between FPRASs discussed previously and AFPRASs is in
that the latter is allowed an additive rather than multiplicative
error, i.e., the difference between the value an AFPRAS returns and
the actual value it approximates is fixed rather than proportional.

\subsection{Sampling the Valuation Space}
\label{ssec:dsamp}

Our approximation scheme is obtained by sampling valuations of the
input database using \emph{samplers}, i.e., algorithms that return
%an independent 
random variables distributed according to some predefined distribution
(see, e.g.,~\cite{blum2020foundations}). For commonly used probability
distributions, samplers can be obtained using different techniques
(see, e.g.,
\cite{box-muller, non-unif}).

Formally, a \emph{sampler} for a probability space
$\cS = (\pdom_S, \salg_S, \prob_S)$ is a randomized algorithm $A$ that
returns a random variable $X$ such that
$P(X \in \sigma) = \prob_S(\sigma)$, for each
$\sigma \in \salg_{\cS}$. % We say that $A$ is \emph{efficient} if it
% runs in time polynomial in $\err^{-1}$.
%, for each input $\err$.
%
Given an incomplete database $\idb$, we say that $\idb$ is
\emph{efficiently sampled} (E.S.) if, for every
$\bot_i\in \nulls(\idb)$ defined by the probability space $\ndis_i$,
there exists a efficient sampler for $\ndis_i$ that runs in time
polynomial in the size of $\idb$. In what follows, we will use
$\nsampalg{i}$ for the result of an efficient sampler for $\ndis_i$.

Using the notions introduced so far, Algorithm~\ref{alg:direct-samp}
defines an approximated sampler for the space of valuation of
incomplete databases that is E.S.

\begin{algorithm}
\begin{algorithmic}
  \STATE {\bf  Parameter:} An E.S. database $\idb$ with $|\nulls(\idb)| = n$ \\
%  \STATE {\bf  Input:} A value $\err \in (0, 1]$\\
  \STATE  {\bf Output:} A valuation $v$ for $\nulls(\idb)$\\
  % \STATE {Let $\gamma = \frac{\err}{(en) \cdot 2}$}
  \FORALL {
    $\bot_i\in\nulls(\idb)$}{
%    \STATE Let $(\bbR, \bsalg, P_i)$ be the probability space that defines $\bot_i$;\\
  %  \STATE  { Let  $S_i$ be a R.V.  s.t. $|P(S_i \in \sigma) - P_i(\sigma)| \le \gamma$, for
    %each $\sigma \in \bsalg$;}\\
    \STATE {$s_i := \nsampalg{i}$ }}\\ \ENDFOR
  \STATE Let $v : \nulls(\idb) \to \bbR$ s.t. $v(\bot_i) = s_i$, for each $\bot_i \in \nulls(\idb)$;\\
  \STATE {\bf return} $v$;\\
   \end{algorithmic}
   \caption{$\vsalg$}
           \label{alg:direct-samp}
\end{algorithm}
The following claim formalizes the relevant properties of $\vsalg$.

\begin{lemma}
  \label{lm:db-sampler}
  For every incomplete database $\idb$, $\vsalg$ is a sampler for the
  valuation space of $\nulls(\idb)$. Moreover, $\vsalg$ runs in time
  polynomial in the size of $\idb$.
\end{lemma}

\subsection{Approximation via Direct Sampling}
\label{ssec:apx-ds}

For input E.S.~databases, we can obtain an AFPRAS for $\otlprob$
using $\vsalg$. The idea behind our technique is the following. Assume
$q \in \sqlra$, an incomplete database $\idb$, and an interval tuple
$\ba$. The algorithm uses $\vsalg$ to sample a number of valuations
$v$ from the valuation space of $\nulls(\idb)$, and computes the
average number of such $v$ for which
$\#(v (\ba), q(v(\idb))) \circ \mathtt k$ holds. Using well-known
bounds on the expected values of a sum of independent random
variables, we can prove that, with high probability, the value of this
average is close enough to $\otlprob(\idb, \ba)$.  We formalize this
intuition in Algorithm~\ref{alg:direct-apx}.

\begin{algorithm}
  \begin{algorithmic}
    \STATE {\bf Parameters:} $q \in \sqlra$, $\circ \in \{<, =, >\}$,
    an integer $\mathtt{k}\ge0$.

    \STATE {\bf Input:} An E.S. database $\idb$, an interval tuple
    $\ba$, $\err \in (0, 1]$.
    \STATE  {\bf Output:} A number between  $0$ and $1$\\
    \STATE   {  Let $\gamma := \lceil \err^{-2} \rceil$;}\\
    \STATE     {$\mathtt{count} := 0$}\\
    \FORALL { $\mathtt{i} = 1, \ldots, \gamma$} {
      \STATE{$v := \vsalg$}\\
      \IF     {$\#(v (\ba), q(v(\idb))) \circ \mathtt{k}$  } \STATE     {$ \mathtt{count}:= \mathtt{count} + 1$} \\
      \ENDIF
      
    }\ENDFOR \STATE {{\bf return }
      $\frac{\mathtt{count}}{\gamma}$}
   \end{algorithmic}
   \caption{$\daalg$}
           \label{alg:direct-apx}
\end{algorithm}

\begin{lemma}
  \label{lm:direct-apx-err}
  For every incomplete E.S. database $\idb$, interval tuple $\ba$ and
  $\err \in (0, 1]$ it holds that
  \[
    P(|\mathbf{\daalg}(\idb, \ba, \err) - \otlprob(\idb, \ba)| \le \err) \ge 0.75
  \]
\end{lemma}
We are finally ready to formally conclude the existence of an AFPRAS for
\textsc{Output Tuple Likelihood}:
%We do this in the next theorem.

\begin{theorem}
  \label{th:afpras}
  For every $q \in \sqlra$, $\circ \in \{<, =, >\}$, and integer
  $\mathtt{k} \ge 0$, $\daalg$ is an AFPRAS for $\otlprob$ over input
  databases that are E.S.
\end{theorem}

\subsection{Encoding Sampling into Queries}
\label{ssec:apx-query}

Our approximation algorithm $\daalg$ requires a separate instantiation
of the input database $\idb$ for each sample.  Practical
implementations of the algorithm would need to generate these instantiations directly inside the RDBMS hosting $\idb$, thus leading to an
overhead that is unacceptable in practical scenarios. Instead of
explicitly generating this data, however, it is possible to encode a
full run of $\daalg$ into a query to be executed directly over
$\idb$. In this section we present such construction.

First, we introduce some notation. To keep notation compact, in all following statements of our results, we let $q$ be a generic query in $\sqlra$, $\idb$ be an
incomplete database, $\ba$ be an interval tuple, $k$ be a positive
integer, and $\circ$ be an element of $\{<,=,>\}$. Moreover, to
evaluate $\sqlra$ queries directly over incomplete databases, we will
implicitly assume \emph{naive evaluation} \cite{naive-eval}, i.e., we
will treat nulls as standard constants.

\begin{theorem}
  \label{th:tuple-query-apx}
  For every $\epsilon\in(0,1]$, there exists
  $\mathtt{apx}_{\err} \in \sqlra$ such that
  \[P(|\mathtt{apx}_{\err}(\idb) - \otlprob(\idb, \ba)| \le \err ) \ge
    0.75.\]
\end{theorem}

The construction of $\mathtt{apx}_{\err}$ is done in several steps
that we proceed to describe.  We first show that we can compile a
specific valuation $v$ into $q$, that is, we can rewrite $q$ to obtain
$\tilde q_v$ such that for every input $\idb$ we have
\[
\tilde q_v(\idb) = q(v(\idb))
\]
Using the \text{COUNT} aggregation, in particular the one corresponding with SQL's $\text{COUNT}(\textasteriskcentered)$, we can obtain $q_{v}$ such that,
%encodes % $q(v(\idb))$ a bit differently, in particular, if
if $\sharp(\brt,\tilde q_v(\idb))= k'$, then $(\brt, k')$ appears once
in $q_{v}(\idb)$. We do so by setting
\[
q_{v} \df 
\text{COUNT}_{\$1,\$2, \ldots,
	\$n}(\tilde q_v)
\]

Since we are interested in the total number of tuples of $q(v(\idb))$
that are consistent with $\bar a \df (a_1,\ldots,a_n)$, we set
\[
q_{v,\bar a} \df
\pi_{\att{n+1}} \left( \sigma_{\att{1} \in v (a_1)\wedge \cdots\wedge \att{n} \in v(a_n)} \left( q_{v}  \right)\right)
\]
where $\att{j} \in v(a_j)$ is the condition requiring that $\att{j}$
is within the interval $v(a_j)$.  Then, to check whether this multiplicity
is $\circ k$, we sum the multiplicities that $q_{v, \ba}$ returns and
select upon the condition $\att{1} \circ k$.  Finally, we project the
result on the empty set, thus obtaining the Boolean query
\[
q_{v,\bar a, \circ k}\df 
\pi_{\emptyset} \left(\sigma_{\att{1} \circ k} \left(\sum_{\emptyset}^{\att{1}} q_{v,\bar a} \right)\right) 
\]
This query returns true, i.e., one occurrence of the empty tuple if $q(v(\idb))$
contains $\circ k$ tuples consistent with $\bar a$, and false, i.e., the empty bag,
otherwise.

Finally, let $V = \{v_1, \ldots, v_\gamma\}$ be the results of running
the algorithm $\vsalg$ for $\gamma = \lceil \err^{-2} \rceil$ times.
We define % obtain a set $V$ of
% valuations $v_1,\ldots v_{\gamma}$ that enable us to define
\[
  \mathtt{apx}_{\err} \df \text{AVG}_{\emptyset} \left(
    \bigcup_{j=1}^\gamma q_{v_j, \ba,\circ k}\right).\]
and show that 
%the query $\mathtt{apx}_\err$ 
it indeed has the desired
properties. Moreover, we can show that %$\mathtt{apx}_\err$ 
it can be
efficiently computed.
\begin{proposition}
	\label{pr:q-size}

        $\mathtt{apx}_{\err}$ can be computed in time polynomial in
        the size of $q$, $\idb$, and $\ba$, and the value $\err^{-1}$.
\end{proposition}

The construction presented so far can be adapted to return the
possible answers of $q(\idb)$ along with their approximated
probability.  To this end, we define
\[
\mathtt{compute}_{\err} \df
\App_{\frac{\att{1}}{\gamma}} \left(\text{COUNT}_{\att{1},\ldots,\att{n+1}}\left(\bigcup_{j=1}^{\gamma}q_{v_j} \right)\right)
\]
Recall that $\bigcup_{j=1}^{\gamma}q_{v_j}$ returns tuples of the form
$(\brt, k')$ such that $\brt$ occurs $k'$ times in $q(v_j(\idb))$.
Thus,
$\text{COUNT}_{\att{1},\ldots,\att{n+1}}\left(\bigcup_{j=1}^{\gamma}q_{v_j}
\right)$ results in tuples of the form $(\brt, k', k'')$ with the same
guarantees as before and, in addition, $k''$ is the number of
valuations $v$ amongst $v_1,\ldots, v_{\gamma}$ for which $\brt$
occurs $k'$ times in $q(v(\idb))$. We can show that each such tuple
represents a possible answer for $q$ over $\idb$.
In the next theorem, by a slight abuse of notation, we refer to a tuple $(c_1,\ldots, c_n)$ as the interval  tuple
$([c_1, c_1], \ldots, [c_n, c_n])$.
\begin{theorem}
	\label{th:list-query-apx}
	Let $\epsilon\in(0,1]$. For every tuple $(\bc, b, p)$ in
        $\mathtt{compute}_{\err}(\idb)$,
	\[P(|p - \otldef{q, =, b}(\idb, \bc)| \le \err ) \ge 0.75\]
\end{theorem}

\section{Finite Representation of the Answer Space}
\label{sec:finiterepoutput}

Additionally to the encoding technique presented previously, a desirable feature that may hint that our model can be deployed in practical use-cases
is the ability to represent the answer space of $\sqlra$
queries in a finite manner. This technique allows us to analyze fully the answer space of a given query which is fundamental in
decision-support scenarios.

 \newcommand{\orc}[1]{orc\left(#1\right)}
 \newcommand{\vext}{v^*}
 \newcommand{\minuss}{\mathttt{minus}}

Do incomplete databases suffice for instantiating the answer space? 
Let us investigate the answer space of  $\sigma_{\att{i}<\att{j}}$ over $\cD$ -- the PDB of $\mathbf R\df \bag{(\bot_1,\mu)}$, where $\bag{\cdot}$ brackets are used for bags. If $\bot_1$ is distributed according to, e.g., the normal distribution with mean $\mu$, then the domain of the answer space contains the empty relation along with instantiations of $\mathbf R$ for which $\bot_1>\mu$. Therefore, there is no incomplete database that instantiates precisely the answer space. Similarly to \cite{IL84} which dealt with the same issue for non-probabilistic databases, these considerations lead to a notion of {\em  conditional} worlds.

\subsection{Conditional Worlds}

The idea behind conditional worlds is to separate the answer space to sub-spaces, each instantiatable by an incomplete database, along with a condition that indicates when it is valid. 
For the previous example,  the answer space can be represented by the empty relation  with the condition $\bot_1<\mu$, and by the relation that consists of $(\bot_1,\mu)$ attached with $\bot_1>\mu$. Note that here and throughout this section, we assume that the probability of each null to be equal to a constant is zero. 
\begin{definition}
A probability distribution $(\dom,\salg, P)$ is said to be \emph{singleton unlikely} if $P(\{c\})=0$ for every $c\in \dom$. 
\end{definition}
This is a very mild restriction, however, since all continuous random variables enjoy this property \cite{probability-book}.
 
To formally define conditional worlds we need to define conditional databases.
Intuitively, conditional databases are pairs of (database, condition) such that database is valid whenever the condition holds.
%such that database is valid that indicates when it is valid.
%As mentioned before, the intuition is that the valuation of $\adb$ represents the output for those valuations for which $c_{\adb}$ holds. 
%which, in turn,  consist of two components -- the actual database which is referred to as an arithmetic database (as we shortly define), and a condition that specifies when it is valid.
To carry full information on the computation and to express the conditions, we extend the domain.
%of our databases. 
 We define an \emph{arithmetic database} as a database whose active domain is  $\RAT[\bot_1,\bot_2,\ldots]$, that is, ratios of polynomials with real coefficients over the nulls. 
\begin{definition}
A \emph{conditional database} is a pair 
$(\adb, c_{\adb})$ where 
\begin{itemize}
    \item $\adb$ is an {arithmetic 
    %incomplete 
    database},
    %, that is, a database $\adb$ with $\adom(\adb) \subseteq \Exp$, 
    and
    \item 
     $c_{\adb}$ is a set of \emph{arithmetic conditions} which are elements of the form 
     $e<0$ where $e\in \RAT[\bot_1,\bot_2,\ldots]$;
     For simplicity, we replace $e<0$ with $e$.
     %$e<e'$ where $e,e'\in \Exp$.
     %$v < v'$ where either $(v,v'\in \Nulls)$ or $(v\in \Nulls, v'\in \Num)$ or   $(v'\in \Nulls, v\in \Num)$.
\end{itemize}
\end{definition}
We denote the set of all nulls that appear in $\adb$ and in $c_{\adb}$ by $\nulls(\adb, c_{\adb})$.

To specify when an arithmetic condition holds, for an arithmetic condition $c$, and a valuation $v$,
we set  $v(c) \df \true$ if  $v(e) < 0$ for every $e \in c$; otherwise, we set $v(c) \df \false$.
%We set $\Nulls(\condworld)$ to consist of the nulls that appear in all elements of $\condworld$.
\newcommand{\defcondworld}{
\sloppy{
A \emph{conditional world} $\condworld$ 
is a set
$
\{ (\adb_1,c_1), \ldots, (\adb_m,c_m)\}
$}
 of conditional databases
such that 
\begin{itemize}
    \item 
    $P_V(\{ v \,\mid\, v(c_1) = \true \vee \cdots \vee v(c_m) = \true \}) = 1$, and
    \item
    $\{ v\,\mid \, v(c_i) = \true \wedge \ v( c_j) = \true\} = \emptyset$ for every $i\ne j$
\end{itemize}
where $(\dom_V,\Sigma_V,P_V)$  is the valuation space of
$\Nulls(\condworld)$ defined as the union $\bigcup_{i=1}^{m} \nulls(\adb_i,c_i)$.}
\begin{definition}\label{def:conddisj}
\defcondworld
\end{definition}
In the previous definition it may be the case that 
$U\df \{ v\mid v(c_1)=\true \vee \cdots\vee v(c_m) = \true \} \subsetneq \dom_V$, which implies that  $P_V(\dom_V \setminus U) = 0$.

 \subsection{Interpreting Conditional Worlds as PDBs}
 
In the next section we will explain how conditional worlds can instantiate the answer space, but before we show how to interpret these worlds as PDBs.
%, but before we make the following important remark.

\begin{definition}
The \emph{probabilistic interpretation $(\dom_{\mathcal C},\salg_{\mathcal C}, P_{\mathcal C})$ of a conditional world $\mathcal C \df \{ (\adb_1,c_1), \ldots, (\adb_m,c_m)\}$} is defined by:
\begin{align*}
    \dom_{\mathcal{C}}
    &\df \bigcup_{i=1}^m \{ v(\adb_i)\,\mid\, %v\text{ is compatible with }
    \dom(v) = \Nulls(\condworld) \}
    \\
     \salg_{\mathcal C} &\df \{ A\,\mid\,\chi^{-1}(A)\in \salg_V\}\\
    %&\df \sigma \{A \, \mid \, A\cap D_i \in \Sigma_i \text{ for every }1\le i\le m \}\\
    P_{\mathcal C}(A) &\df  
    %P_V(\{v\,\mid \, v(\adb_1)\cap A\ne \emptyset , c_1 = \true\}) 
P_V(\chi^{-1}(A))
\end{align*}
where $(\dom_V, \salg_V, P_V)$ is the valuation space of %$\bigcup_{i=1}^k \Nulls(\adb_i)\cup \bigcup_{i=1}^k \Nulls(c_i)$.
$\Nulls(\condworld)$, 
and $\chi : \dom_V \rightarrow \dom_{\condworld}$ is defined by
\[\chi(v) \df \left\{\begin{matrix}
 v(\adb_1)& \text{if } v(c_1) = \true \\ 
 \vdots& \vdots \\ 
 v(\adb_m)& %\text{if } v(c_m) = \true \\
 \text{ otherwise}
\end{matrix}\right.\]
\end{definition}
Note that $\chi$ is well-defined due to Definition~\ref{def:conddisj}, and hence the probabilistic interpretation is also well-defined.

 Before proceeding, we show that as its name hints, the probabilistic interpretation is indeed a PDB.
 \newcommand{\lemcondworldasprob}{ 
  For every conditional world $\condworld$, {the probabilistic interpretation of $\condworld$} is  a PDB.
 }
 \begin{lemma} \label{lem:condworldasprob}
\lemcondworldasprob
 \end{lemma}

\subsection{Lifting Queries to Conditional Worlds}

 To show how conditional worlds instantiate the answer space, we  need to show how queries can be lifted from complete databases to conditional worlds. 
For all queries, except that of inequality selections, the definitions are rather straightforward. 

The intuition behind $\sigma_{\att{i}<\att{j}}(\condworld)$ is that the result of evaluating an inequality selection on $\condworld \df \{ (\adb_1,c_1),\ldots,(\adb_m,c_m) \}$ depends on which
tuples are not filtered out. A tuple $\brt\df (t_1,\ldots, t_n)$ is filtered out if the condition $t_i-t_j<0$ is not satisfied. 
Thus,  conditions that may affect the result are  elements in  \[\condworld(i,j) \df \{ t_i -t_j \mid (t_1,\ldots, t_n) \in \cup_{i=1}^m {\adb_m}
\}\]
 and each subset of this set corresponds with a subset of the results.
%In fact, we can disregard constants $c$ in $\tilde\condworld(i,j)$ since whether $c<0$ is satisfied is independent of the valuation, that is, it is either always  satisfied or never. We can thus proceed with $\condworld(i,j) \df \tilde{\condworld}(i,j) \setminus \mathbb{R}$.
We define the condition $c_B$ of a subset  $B\subseteq \condworld(i,j)$ in a way it ensures that (1) all tuples that are consistent with $B$ remains, and (2) all those that are not are filtered out. Thus, we set $ c_B \df B \cup \{-b \mid b\in  \condworld(i,j) \setminus B\}$. The first element in the union ensures (1) and the second (2). 
With this notation we can present the full definition.

\begin{definition}\label{def:semcondworld}
We define the semantics $q(\condworld)$ of queries $q$ on conditional worlds $\condworld$ inductively as follows:
\begin{align*}
     \condworld \ \circ \  \condworld' &\df \set{
	(\adb \  {\circ}\  \adb' ,c \cup c')\,\mid\,  (\adb,c)\in \condworld, (\adb',c')\in\condworld' } 
	\\
      \semcondworld{q}  &\df\{ (q(\adb),c),\ \mid \, (\adb,c)\in \condworld \}  
       \\
      \semcondworld{\sigma_{\att{i}<\att{j}}} &\df 
\{(\sigma_{B}(\adb), c \cup c_B \,\mid\,
(\adb,c)\in \condworld,\, 
B\subseteq \condworld(i,j) \} 
\end{align*}
with  $\circ 
	\in \{ \times, \cup, \setminus \}$, 
	 $q \df R \mid \pi_{\att{i_1},\ldots, \att{i_k}}(q) \mid  \sigma_{\theta}(q) \mid  \App_f(q) \mid  \sum_{\att{i_1},\ldots, \att{i_k}}^{\att j}(q)$, 
$\sigma_{\emptyset}(\adb) \df \emptyset$,
   and all operators defined on  complete databases whose active domain is extended to $\RAT[\bot_1,\bot_2,\ldots]$.
\end{definition}

To show that this inductive definition is indeed well defined, it suffices to show  compositionality of  conditional worlds:
 \newcommand{\lemconddisj}{
 $\semcondworld{q}$ is a conditional world for every conditional world $\condworld$ and query $q$.
 }
 \begin{lemma} \label{lem:conddisj}
 \lemconddisj
 \end{lemma}
\begin{example}
\OMIT{
Assume that $\condworld$ consists of $(\adb,c)$ where $\adb = \bag{ (\bot_1,2),(\bot_1,\bot_1+2),(\bot_3,\bot_1+\bot_3) }$. 
To compute $\pi_{\att{1} < \att{2}}(\condworld)$
we first notice that  $\condworld(1,2) = \{\bot_1-2, -\bot_1 \} $. 
There are, therefore, four subsets of $\condworld(1,2)$, and  we obtain
\begin{multline*}
    \pi_{\att{1} < \att{2}}(\condworld) = \{ (\emptyset,c \cup \{ -\bot_1+2,\bot_1 \} ),\\
(\sigma_{\bot_1-2}(\adb),c \cup \{ \bot_1-2, \bot_1 \}),\\ \hspace{1cm}\ \ \ \ \ \ \ \ 
(\sigma_{-\bot_1}(\adb),c \cup \{-\bot_1, -\bot_1+2 \}),\\
(\sigma_{\bot_1-2, -\bot_1}(\adb),c \cup\{\bot_1-2, -\bot_1 \})\}
\end{multline*}
} 
Assume that $\condworld$ consists of $(\adb,c)$ where $\adb = \bag{ (\bot_1,0),(\bot_1,\bot_1),(\bot_3,\bot_1+\bot_3) }$. 
To compute $\sigma_{\att{1} < \att{2}}(\condworld)$, 
we first notice that  $\condworld(1,2) = \{\bot_1, -\bot_1 \} $. 
There are, therefore, four subsets of $\condworld(1,2)$, and  we obtain
\begin{multline*}
    \sigma_{\att{1} < \att{2}}(\condworld) = \{ (\emptyset,c \cup \{ -\bot_1,\bot_1 \} ),
(\sigma_{\bot_1}(\adb),c \cup \{  \bot_1 \}),\\ \hspace{1cm}\ \ \ \ \ \ \ \ 
(\sigma_{-\bot_1}(\adb),c \cup \{ -\bot_1 \}),
(\sigma_{\bot_1, -\bot_1}(\adb),c \cup\{\bot_1, -\bot_1 \})\}
\end{multline*}
Note that each condition that contains $\{ -\bot_1, \bot_1 \}$ never holds. And, indeed, their corresponding databases cannot be instantiated to any of the possible answers.
\end{example}

\subsection{Answer Spaces as Conditional Worlds}
\label{sec:OutasCond}
 
 We now show that we can instantiate a probability space that is almost similar to the answer space using the lifting we presented in the previous section. 
 But before, what does it mean to be almost similar?
 A probability space $(\dom, \salg, P)$ is said to be a \emph{trivial extension} of a probability space $(\dom', \salg', P')$ if $\dom \supseteq \dom'$, $\salg \supseteq \salg'$, and $P'(\sigma) = P(\sigma)$ for every $\sigma \in \salg'$.
 
 \newcommand{\thmcondandoutput}
 {
 For every incomplete database $\idb$ and query $q$,  
 %that if $\condworld' \df \semcondworld{q}$
%then 
if each null in  $\Nulls(\idb)$ has a singleton-unlikely distribution then
the answer space of $q$ over $\idb$ is a trivial extension of
 $\semcondworld{q}$ where $\condworld \df \{ (\idb, -1)\}$.
 }
 \begin{theorem}
\label{thm:condandoutput}
\thmcondandoutput
 \end{theorem}
 The choice of $\condworld$ is due to the fact that its probabilistic interpretation is similar to the PDB of $\idb$.

\section{Conclusion}
\label{sec:conclusion}

We presented a novel model for incomplete numerical data based on tuple-dependent infinite probabilistic databases.
We studied query answering in a language
that captures key features of SQL, looked at exact and approximate answers, and showed how to represent query outputs.
As for future work,
%from a theoretical standpoint,
we would like to extend our model to other SQL data-types, especially
discrete domains.
We would like to study other forms of approximations,
%and more efficient
%algorithms.
in particular direct sampling  from the conditional world that represents the output. Such representations may be optimized deploying ideas from cylindrical algebraic decomposition~\cite{DBLP:conf/dimacs/2001}. 
From a practical standpoint, we would like to implement
our theoretical algorithms using real-world RBDMSs, and test their
performances with real data.

\bibliographystyle{acm}
\bibliography{biblio}

\onecolumn
\appendix

\section{Additional Notions }

A family of sets $\cS$ is a \emph{$\pi$-system} if it is closed under
(finite) intersection, i.e., $S,S' \in \cS$ implies
$S \cap S' \in \cS$. % Moreover, $\cS$ is an algebra of sets (simply,
% algebra) on a domain $X$ if $S \subseteq X$, for each $S \in \cS$,
% $X \in \cS$, and $\cS$ is closed under finite intersection and taking
% complements w.r.t. $X$, i.e., $S \in \cS$ implies
% $X \setminus S \in \cS$.
Assume a measure space $(\pdom, \salg, \mu)$. We say that $\mu$ is
\emph{$\sigma$-finite} on a family $\cB \subseteq \salg$ if there
exists a countable family of disjoint sets $\{B_i~\mid i > 0\}$ in
$\cB$ such that $\bigcup_i \{B_i~\mid~i>0\} = \pdom$ and
$\mu(B_i) < \infty$, for each $i >0$.

\begin{proposition}
  \label{pr:unique-pie}
  Suppose $\salg$ is the $\sigma$-algebra generated by the
  $\pi$-system $\cP$ and that $\mu$ is $\sigma$-finite over
  $\cP$. Then there exists a unique measure $\mu'$ for
  $(\pdom, \salg)$ that coincides with $\mu$ over $\cA$.
\end{proposition}
The proof of this classical result can be found, e.g., in
\cite{durrett2019probability}.

To avoid confusion, in what follows we use $\bigtimes_{i=1}^n F_i$,
where $\{F_i~\mid~ i=1, \ldots,n\} $ is a family of sets, to denote
the \emph{Cartesian product} $F_1 \times F_2 \times \ldots \times F_n$
and $\Pi_{i=1}^n c_i $, with $c_i \in \bbR$ for $i = 1, \ldots, n$, to
denote the \emph{product} $c_1 \cdot c_2 \cdot \ldots \cdot c_n$.

To define probability measures for our spaces, in the proofs, we will
often rely on the \emph{push forward} technique. Let
$\cA_1= (\pdom_1, \salg_1)$ and $\cA_2 = (\pdom_2, \salg_2)$ be
measurable spaces and assume an $\cA_1/\cA_2$-measurable mapping
$f : \pdom_1 \to \pdom_2$ and a probability measure $\prob_1$ for
$\cA_1$. The \emph{push forward} of $\prob_1$ under $f$ is the
function $\mu \circ f^{-1} : \cF_B \to \bbR$. The following
proposition is a well-known result of probability theory (see, e.g.,
\cite{schilling2017measures}).
\begin{proposition}
  \label{pr:pushforward}
  The push forward of $\prob_1$ under $f$ is a measure for $\cA_2$.
\end{proposition}

Before proceeding to present the technical proofs of the claims in the
paper, we prove a general result that will be reused throughout the
appendix. Let $\idb$ be an incomplete database, and let $\cV$ be the
set of generators of the $\sigma$-algebra $\salg_V$ of the measurable
space of valuations for $\idb$. We prove

\begin{claim}
  \label{cl:vals-algebra}
  The set $\cV$ is $\pi$-system.
\end{claim}
\begin{proof}

  We can just prove that, for each $V, V' \in\cV$, also
  $V \cap V' \in \cV$. The general case follows by induction. Suppose
  $V = V(\sigma_1, \ldots, \sigma_n)$ and
  $V' = V(\sigma'_1, \ldots, \sigma'_n)$ and let
  $V(\sigma_1 \cap \sigma'_1, \ldots, \sigma_n\cap \sigma'_n) =
  \{v~\mid~ v(\bot_i) \in \sigma_i \cap \sigma'_i \textit{, for each }
  i = 1, \ldots n\}$.  First, we show
  $V \cap V' = V(\sigma_1 \cap \sigma'_1, \ldots, \sigma_n\cap
  \sigma'_n)$. Let $v \in V \cap V'$. Then, $v \in V$ and $v \in V'$,
  therefore, $v(\bot_i) \in \sigma_i$ and $v(\bot_i) \in \sigma'_i$,
  for each $i = 1, \ldots, n$. We can conclude that
  $v \in V(\sigma_1 \cap \sigma'_1, \ldots, \sigma_n\cap
  \sigma'_n)$. Let now
  $v \in V(\sigma_1 \cap \sigma'_1, \ldots, \sigma_n\cap
  \sigma'_n)$. Since $v(\bot_i) \in \sigma_i$, for each
  $i = 1, \ldots, n$, we can conclude $v \in V$. Similarly, since
  $v(\bot_i) \in \sigma'_i$, for each $i = 1, \ldots, n$, we can
  conclude $v \in V'$, in turn proving $v \in V \cap V'$. We observe
  now that, since $\bsalg$ is a $\sigma$-algebra,
  $\sigma_i \cap \sigma'_i \in \bsalg$, for each $i = 1, \ldots n$. By
  definition, then,
  $V(\sigma_1 \cap \sigma'_1, \ldots, \sigma_n\cap \sigma'_n) \in \cV$
  and the claim follows.

\end{proof}

\section{Proofs of Section~\ref{sec:framework}}

\subsection{Proof of Proposition~\ref{pr:query-measurability}}
First, we prove that $q$ is $(\salg_{\cD},\salg_{q,\cD})$-measurable.
Let $\cG_q = \{ A~\mid~ \pre{q}(A) \in \pdom_\idb\}$ be the set of
generators of $\salg_q$. For every set $A \in \cG_q$,
$\pre{q}(A) \in \salg_\cD$ by definition, and, thus, $q$ is measurable
over the generators of $\salg_{q, \cD}$. We observe now that if a
mapping is measurable over the generators of a $\sigma$-algebra, then
it is measurable over the whole $\sigma$-algebra (see, e.g.,
\cite{schilling2017measures}). We can conclude that $q$ is
$(\salg_{\cD},\salg_{q,\cD})$-measurable.

We now show that $\prob_{q,\cD}$ is a probability measure for
$(\pdom_{q,\cD}, \salg_{q,\cD})$. To this end, simply observe that
$\prob_{q,\cD}$ is the push forward of $\prob_\cD$ under $q$. Since
$q$ is $(\salg_{\cD},\salg_{q,\cD})$-measurable, $\prob_{q,\cD}$ is
the push forward of a measure under a measurable mapping. The claim
follows from the fact that the push forward of a measure under a
measurable mapping defines a measure
(Proposition~\ref{pr:pushforward}).

\subsection{Proof of Proposition~\ref{pr:composability}}
Assume $\cD = (\pdom_\cD, \salg_\cD, \prob_\cD)$, and let $(f,g)$ be a
composable pair of queries such that $q = f \circ g$. Moreover, let
$\cQ = (\pdom_q, \salg_q, \prob_q)$,
$\cF = (\pdom_f, \salg_f, \prob_f)$, and
$\cG = (\pdom_g, \salg_g, \prob_g)$ be the answer spaces defined by
$q(\cD)$, $f(\cD)$ and $g(f(\cD))$, respectively. We proceed to prove
that $\cQ = \cG$, thus proving that $g \circ f(\cD) =
g(f(\cD))$. Since $q(D) = g(f(D))$, for every complete instance $d$
for $\schemdb$, we have $\pdom_q = \pdom_g$. What is left to show is
that $\salg_q = \salg_g$ and $\prob_q = \prob_g$. We divide the proof
in three parts.

$\salg_q \subseteq \salg_g$. Let $A \in \salg_q$, we prove that
$A \in \salg_g$. First, we recall that $q$ is a
$(\cD, \cQ)$-measurable
(Proposition~\ref{pr:query-measurability}). Therefore, there exists
$B \in \salg_\cD$ such that $\pre{q}(A) = B$. From the definition of
answer spaces then, we can conclude that there exists $C \in \salg_f$
such that $\pre{f}(C) = B$ and, similarly, there exist $H$ in
$\salg_g$ such that $\pre{g}(H) = C$. Observe now that $q$, $g$, and
$f$ are surjective mappings. Therefore, $H = g(f(B)) = q(B) = A$, in
turn proving that $A \in \salg_g$.

$\salg_q \supseteq \salg_g$. Let now $A \in \salg_g$, we prove that
$A \in \salg_q$.  First, observe that $g$ is $(\cF, \cG)$-measurable,
is $f$ $(\cD, \cF)$-measurable
(Proposition~\ref{pr:query-measurability}) Then, there exists $B$ in
$\salg_f$ such that $\pre{g}(A) = B$ and, similarly, there exists $C$
in $\salg_\db$ such that $\pre{f}(B) = C$. From the definition of
answer spaces, we can conclude that there exists $H \in \salg_q$ such
that $\pre{q}(H) = C$. Observe now that $q$, $g$, and $h$ are
surjective mappings. Therefore, $H = g(f(C)) = q(C) = A$, in turn
proving that $A \in \salg_q$.

$\prob_q(A) = \prob_g(A)$, for each $A \in \salg_q$. Let
$A \in \salg_q$. Then, $\prob_q(A) = \prob_\cD(B)$, where
$B = \pre{q}(A)$. Similarly, $\prob_g(A) = \prob_f(C)$, where
$C = \pre{g}(A)$ and $\prob_f(C) = \prob_\cD(H)$, where
$H = \pre{f}(C)$. This is due to the definition of answer space. We
now observe that $q(d) = g(f(d))$, for each complete instance for
$\schemdb$, and then, $\pre{q}(\sigma) = \pre{g}(\pre{f}(\sigma))$,
for each $\sigma \in \salg_q$. We can conclude that $A = H$, which in
turn proves that $\prob_q(A) = \prob_g(A)$.

\subsection{Proof of Theorem~\ref{th:val-space-unique}}

Assume that $\nulls(\idb) = \{\bot_1, \ldots, \bot_n\}$ and that
$\ndis_i = (\bbR, \bsalg, \prob_i)$ defines $\bot_i$, for
$i = 1, \ldots, n$. Let $\cR = (\bbR^n, \bsalg^n)$ be the measurable
space defined by the algebra $\bsalg^n$ generated by the family
$\{\bigtimes_{i=1}^n \sigma_i~\mid~ \sigma_i \in \bsalg \textit{ for }
i=1, \ldots n\}$. In other words, $\cR$ is a product measurable space
generated by the probability spaces that define the nulls of $\idb$.
\begin{proposition}
  \label{pr:uniq-prod}
  There exists a unique measure $P$ for $\cR$ such that, for each
  $\sigma_1, \ldots, \sigma_n \in\bsalg$ we have
  \[P({\sigma_1 \times \ldots \times \sigma_n}) = \Pi_{i=1}^n P_i(F_i)\]
\end{proposition}
The proof of this classic result can be found, e.g., in
\cite{durrett2019probability}.  Using Proposition~\ref{pr:uniq-prod},
we will show that there exists a unique probability measure $P$ for
the measurable space of valuations $\cV = (\pdom_V, \salg_V)$ of
$\idb$ that satisfies Equation~\ref{eq:product}. To this end, we
define the mapping $g : \bbR^n \to \pdom_V$ s.t., for each
$(r_1, \ldots, r_n) \in \bbR^n$, $g((r_1, \ldots, r_n))=v$ with
$v(\bot_i) = r_i$, for each $\bot_i \in \nulls(\idb)$. It is easy to
see that $g$ is a mapping, we proceed to show that $g$ is also
measurable.

\begin{lemma}
\label{lm:mes-map}
The mapping $g$ is a $\cR/\cV$-measurable.
\end{lemma}
\begin{proof}
  To prove that $g$ is measurable, we just need to prove that it is
  measurable for the generators of $\salg_V$ (see, e.g.,
  \cite{schilling2017measures}). More specifically, with
  $\cG = \{ V(\sigma_1, \ldots, \sigma_n) ~\mid~ \sigma_i \in \bsalg
  \textit{ for } i=1, \ldots n\}$, we need to prove the following
  statement:
  \[
    \pre{g}(X) \in \bsalg^n, \textit{ for each } X \in \cG
  \]
  To prove the claim, simply observe that, for each
  $X = V(\sigma_1, \ldots, \sigma_n) \in \cG$, we have
  $\pre{g}(X) = \{(v(\bot_1), \ldots, v(\bot_n)) \in \bbR^n~\mid~
  v(\bot_i) \in \sigma_i\}$. From the definition of $\cR$, then, we
  can conclude that $\pre{g}(X) \in \bsalg^n$, and the claim follows.
\end{proof}
Using Lemma~\ref{lm:mes-map}, we can prove the
Theorem~\ref{th:val-space-unique}.

\begin{proof}[Proof of Theorem~\ref{th:val-space-unique}]
  (Existence) The existence of $P$ is consequence of the fact that $g$
  is a measurable mapping ( Lemma~\ref{lm:mes-map}). Let $\rho$ be the
  product measure for $\cR$ such that
  $\rho(\sigma_1 \times \ldots \times \sigma_n) = \Pi_{i=1}^n
  P_i(\sigma_i)$, and let $P$ be the push forward of $\rho$ under
  $g$. For each $V(\sigma_1, \ldots, \sigma_n)$, with
  $\sigma_i \in \bsalg$, we have
  $\pre{g}(V(\sigma_1, \ldots, \sigma_n)) = P_1(\sigma_1) \times
  \ldots \times P_n(\sigma_n)$, and therefore
  $P(V(\sigma_1, \ldots, \sigma_n)) = \Pi_{i=1}^n P_i(\sigma_i)$. The
  claim follows from the fact that $P$ is a measure.

  (Uniqueness) Assume the measure $P$ defined in the previous step of
  this proof. To prove the claim we use the result in
  Proposition~\ref{pr:unique-pie}. Let $\cV$ be the set of generators
  for $(\pdom_V, \salg_V)$. Clearly, since the set of all valuations
  $\pdom_V$ is in $G$ by definition, and $P(\pdom_V) = 1$, we can
  conclude that $P$ is $\sigma$-finite over $\cV$. Moreover, we know
  that $\cV$ is a $\pi$-system (Claim~\ref{cl:vals-algebra}).
%
  % We proceed to show
  % that $\cV$ is closed under intersections, in turn proving that $\cV$
  % is a $\pi$-system. Let
  % $V(\sigma_1, \ldots, \sigma_n), V(\sigma'_1, \ldots, \sigma'_n) \in
  % \cV$. By definition,
  % $V(\sigma_1, \ldots, \sigma_n) \cap V(\sigma'_1, \ldots, \sigma'_n)
  % = V(\sigma_1 \cap \sigma'_1, \ldots, \sigma_n \cap
  % \sigma'_n)$. Observe now that , since $\bsalg$ is a
  % $\sigma$-algebra, we have that
  % $(\sigma_i \cap sigma'_i) \in \bsalg$, for each $i = 1, \ldots,
  % n$. By definition then
  % $V(\sigma_1 \cap \sigma'_1, \ldots, \sigma_n \cap \sigma'_n) \in
  % \cV$, proving that $\cV$ is closed under intersections.
  To conclude the proof, we simply observe that, since $\cV$ is a
  $\pi$-system and $P$ is $\sigma$-finite over $\cV$,
  Proposition~\ref{pr:unique-pie} implies that every measure for
  $(\pdom_V, \salg_V)$ coincides with $\mu$.
\end{proof}

\subsection{Proof of Proposition~\ref{pr:sigma-db}}
  We show that $\salg_\idb$ is closed under complement and countable unions.

  (Complement) Let $D \in \salg_\idb$ with
  $\pre{\chi}(D) = A \in \salg_V$. We show that
  $\pre{\chi}(\bar{D}) \in \salg_V$, where $\bar{D}$ is the complement
  of $D$, in turn proving that $\bar{D} \in \salg_\idb$. To this end,
  we observe that $\pre{\chi}(\pdom_\idb) = \pdom_V$. Therefore,
  $\pre{\chi}(\pdom_\idb \setminus D) = \pre{\chi}(\pdom_\idb)
  \setminus \pre{\chi}(D) = \pdom_V \setminus A = \bar{A}$. Since
  $\salg_V$ is a $\sigma$-algebra and $A \in \salg_V$, we can conclude
  that $\bar{A} \in \salg_V$. In turn, the latter proves that
  $\pre{\chi}(\bar{D}) \in \salg_V$ and the claim follows.

  (Countable Union) Let $\{D_i~\mid~i>0\}$ be a countable family of
  sets in $\salg_\idb$, and let $\bigcup_i D_i = D$. We proceed to
  show that $\pre{\chi}(D) \in \salg_V$, in turn showing that
  $D \in \salg_\idb$. Since each $D_i$ is in $\salg_\idb$, by
  definition, we have $\pre{\chi}(D_i) \in \salg_V$. Therefore,
  $\bigcup_i \pre{\chi}(D_i) = V \in \salg_V$. By definition,
  $d \in \bigcup_i \{D_i~\mid~i>0\}$ implies $d \in D_i$, for some
  $i = 1, \ldots, n$. Therefore,
  $\pre{\chi}(D) = \{d~\mid~d\in \pre{\chi}(D_i) \textit{ for some } i
  = 1, \ldots, n \}$.  In turn, the latter implies that
  $\pre{\chi}(D) = V$ and the claim follows.

\subsection{Proof of Proposition~\ref{pr:prob-dbs}}

Let $\cV = (\pdom_V, \salg_V)$ and $\cD = (\pdom_\idb, \salg_\idb)$.
To prove the claim, we first observe that $\prob_\idb$ is the push
forward of $\prob_V$ under the mapping $\chi$. To prove that
$\prob_\idb$ is a probability, we simply need to show that $\chi$ is
$(\cV/\cD)$-measurable and then use Proposition~\ref{pr:pushforward} to
prove the claim. 

To prove that $\chi$ is $(\cV/\cD)$-measurable we observe that, by
definition of the PDB of $\idb$, $D \in \salg_\idb$ if and only if
$\pre{\chi}(D) \in \salg_V$. Then, for each $D \in \salg_\idb$, we
have $\pre{\chi}(D) \in \salg_V$ and the claim follows
straightforwardly.

\newcommand{\semvl}[1]{{\llbracket{#1}\rrbracket}}
\newcommand{\join}{\bowtie}
\section{Appendix for Section~\ref{sec:sqlra}}
\newcommand{\Card}{\mathsf{Card}}
\newcommand{\notnull}{\mathsf{NotNull}}
\newsavebox{\semt}
\sbox{\semt}{%
	\parbox{\textwidth}{%
		\begin{align*}
		\sem{t}_{\eta}& \df \left\{\begin{matrix}
		t & t \in \mathbb{R} \\ 
		\eta(t) & t \df \att{i} \\
	%	\eta(t_1) \circ \eta(t_2)& t \df t_1 \circ t_2 \\
	%	\exp(\eta(t')) & t \df \exp(t')\\
	%	\log(\eta(t')) & t \df \log(t')
		\end{matrix}\right.\\
		\sem{(t_1 \circ t_2 )}_{\eta } &\df (\sem{t_1}_{\eta} \circ \sem{t_2}_{\eta} )\\
		\sem{\exp(t )}_{\eta } &\df
	 \exp(\sem{t}_{\eta})\\	
	 	\sem{\log(t )}_{\eta } &\df
	 \log(\sem{t}_{\eta})
		\end{align*}
	}%
}%
\newcommand{\semfigt}{%
	\begin{figure*}%[h]
		\centering{
		\fbox{\usebox{\semt}}}
		\caption{Semantics of Terms}
	\end{figure*}
}%
%{\semfigt}

%%%%%%%%%%%%%%%%%%%%%%%%%%%%%%%%%%

\newsavebox{\semn}
\sbox{\semn}{%
	\parbox{\textwidth}{%
		\begin{align*}
\ell\left(\pi_{t_1 [\shortrightarrow  \! N_1],\cdots, t_m [\shortrightarrow   \!N_m] }(E)\right) & \df \tilde N_1 \cdots \tilde N_m\\
\text{where } \tilde N_i& \df \left\{\begin{matrix}
N_i & \text{ if $[\shortrightarrow \!  N_i] $} \\ 
\Name(t_i)  & \text{otherwise}
\end{matrix}\right.\\
\ell\left(\sigma_{\theta}(E)\right) & \df \ell(E)\\
\ell\left( E_1 \times E_2 \right) &\df \ell(E_1)\cdot \ell(E_2)\\
\ell\left( E_1 \op E_2 \right) &\df \ell(E_1) \text{ for } \op\in \{\cup, \cap, \setminus  \} \\
\ell\left(\epsilon(E)\right) & \df \ell(E)\\
\begin{multlined}[t]
\ell ( \Group_{\barN} \langle F_1(N_1)[\shortrightarrow \!  N'_1],\cdots,\\ F_m(N_m)[\shortrightarrow  \! N'_m] \rangle (E) ) 
\end{multlined}
&\df
\begin{multlined}[t]
 \barN\cdot\tilde N_1 \cdots \tilde N_m
\end{multlined}
 \\
\text{where } \tilde N_i& \df \left\{\begin{matrix}
N_i & \text{ if $[\shortrightarrow \!  N_i] $} \\ 
\Name(F_i(N_i))  & \text{otherwise}
\end{matrix}\right.
		\end{align*}
	}%
}%

\newcommand{\semfign}{%
	\begin{figure*}%[t]
		\centering
		\fbox{\usebox{\semn}}
		\caption{Names Assigned to Expressions}
%		\label{fig:semn}
	\end{figure*}
}%
%{\semfign}

\newsavebox{\semc}
\sbox{\semc}{%
	\parbox{\textwidth}{%
	\textbf{Basic conditions}
	        \begin{align*}
		\sem{\true}_{D,\eta}& \df \true %\\
                \ \ \ \ 
		\sem{\false}_{D,\eta} \df \false\\
%		\sem{\isnul(t)}_{D,\eta}&\df	\left\{\begin{matrix}
%		\true  &  \sem{t}_{\eta} = \NULL \\ 
%		\false &  \text{otherwise}
%		\end{matrix}\right.\\
	\sem{t \,\omega \, t'}_{D,\eta} & \df
\left\{\begin{matrix}
		\true  &  	\sem{t }_{\eta} \,\omega \, 	\sem{t'}_{\eta}\\ 
		\false &  \text{otherwise}
		\end{matrix}\right.\\
			\text{where } \omega& \in \{ <,>,\le,\ge,\eq,\neq \} \\
		\sem{\brt \eq \brt'}_{D,\eta} & \df \bigwedge_{i=1}^{n} \sem{t_i \eq t'_i}_{\eta}\\
		\sem{\brt \ne \brt'}_{D,\eta} & \df
		\neg \sem{\brt \eq \brt'}_{D,\eta} \\
		\text{for } \omega\in \{<,>\}\,\, 
			\sem{\brt \,\omega\, \brt'}_{D,\eta}&\df
	\bigvee_{1\le i\le n-1} \big(
	\bigwedge_{1\le j \le i} 
	\sem{t_i  \eq t'_i}_{\eta} 
	\wedge
	\sem{t_{i+1}  \, \omega \, t'_{i+1}}_{\eta} 
	\big)
	 \\
	 			\text{where } \brt& \df (t_1,\ldots, t_n)\\
			\brt'& \df (t'_1,\ldots, t'_n)\\
			\sem{\brt \le \brt'}_{D,\eta}&\df
			\sem{\brt \eq \brt'}_{D,\eta} \vee \sem{\brt < \brt'}_{D,\eta}
	 \\
	 			\sem{\brt \ge \brt'}_{D,\eta}&\df
			\sem{\brt \eq \brt'}_{D,\eta} \vee \sem{\brt > \brt'}_{D,\eta}
	 \\
		\sem{\brt\in q}_{D,\eta}& \df
		\bigvee_{\brt' \in \sem {q}_{D,\eta}
	}  \sem{\brt \eq \brt'}_{\eta}\\
	\sem {\brt\, \omega\, \any(q)}_{D,\eta}& \df
		\bigvee_{\brt' \in \sem {q}_{D,\eta}
	}  \sem{\brt\, \omega\, \brt'}_{\eta}\\
	\sem {\brt\, \omega\, \all(q)}_{D,\eta}& \df
		\bigwedge_{\brt' \in \sem {q}_{D,\eta}
	}  \sem{\brt\, \omega\, \brt'}_{\eta}\\
		\sem {\isempty(q)}_{D,\eta} &\df 
		\left\{\begin{matrix}
		\true  & \sem{q}_{D,\eta} = \emptyset	 \\ 
		\false &   \text{otherwise}
		\end{matrix}\right.\\
		\end{align*}
\textbf{Composite conditions}
\begin{align*}
	\sem{\theta_1 \vee \theta_2}_{D,\eta} &\df \sem{\theta_1 }_{D,\eta}\vee \sem{\theta_2}_{D,\eta}\\
	\sem{\theta_1 \wedge \theta_2}_{D,\eta}& \df \sem{\theta_1 }_{D,\eta}\wedge \sem{\theta_2}_{D,\eta}\\
	\sem{\neg \theta}_{D,\eta} &\df \neg \sem{\theta}_{D,\eta}.\\
\end{align*}
}
}%

\newcommand{\semfigc}{%
	\begin{figure}[h]
		\centering
		\fbox{\usebox{\semc}}
		\caption{Semantics of Conditions}
%		\label{fig:semc}
	\end{figure}
}%
%\semfigc

\newsavebox{\seme}
\sbox{\seme}{%
	\parbox{\textwidth}{%
	\begin{align*}
	\sharp(\bara, {R}(\db)) = k & \text{ if } \sharp(\bara , R^{\db}) = k\\
\sharp((a_{1},\ldots,a_{m}) , \pi_{\att{i_1} ,\cdots,\att{i_m} }(q)(\db)) = k
& \text{ if } 
     k\df | \bag{(t_1,\ldots,t_n) \in  {q}(\db) \mid  a_1 = t_{i_1},\ldots, a_m = t_{i_m}  } |
     \\
     	\sharp(\bara, \sigma_{ \att{i} \circ \att{j}}(q)(\db)) = k &\text{ if }
     \sharp((a_1,\ldots, a_n),  {q}(\db)) = k \ \text{and}\
  a_{i} \circ a_j,\ \circ\in\{=,< \}
     \\
	\sharp((a_1,\ldots, a_n), q_1 \times q_2 (\db)) = k &\text{ if } \exists i:\
	\sharp((a_1,\ldots,a_i),q_1(\db)) = k_1, \	\sharp((a_{i+1},\ldots,a_n),q_2(\db)) = k_2, \
	k=k_1\cdot k_2
	\\
		\sharp((a_1,\ldots, a_n), q_1 \cup q_2 (\db)) = k &\text{ if }
	\sharp((a_1,\ldots,a_n),q_1(\db)) = k_1, \	\sharp((a_{1},\ldots,a_n),q_2(\db)) = k_2, \
	k=k_1 + k_2
	\\
			\sharp((a_1,\ldots, a_n), q_1 \setminus q_2 (\db)) = k &\text{ if }
	\sharp((a_1,\ldots,a_n),q_1(\db)) = k_1, \	\sharp((a_{1},\ldots,a_n),q_2(\db)) = k_2, \
	k=\max\{0,k_1 - k_2\}
	\\
\sharp((a_1,\ldots, a_n), \App_f(q) (\db)) = k &\text{ if } \sharp((a_1,\ldots, a_{n-1}), q (\db)) = k, \ a_n = f([\att{1}\mapsto a_{1},\cdots,\att{n}\mapsto a_{n} ]) \\
\sharp((a_1,\ldots, a_{m+1}), \sum_{\att{i_1},\ldots, \att{i_m}}^{\att{j}}(q) (\db)) = 1 
&\text{ if } a_{m+1} = \sum_{(t_1,\ldots,t_n)\in I}t_{j}, \ I = \sigma_{\att{i_1} = a_1}\left(
\cdots \sigma_{\att{i_m} = a_m}(q(\db))\cdots 
\right), I\ne \emptyset \\
\sharp((0), \sum_{\att{i_1},\ldots, \att{i_m}}^{\att{j}}(q) (\db)) = 1 
&\text{ if } \sigma_{\att{i_1} = a_1}\left(
\cdots \sigma_{\att{i_m} = a_m}(q(\db))\cdots 
\right)=\emptyset
	\end{align*}
	}%
}%

\newcommand{\semfige}{%
	\begin{figure*}[t]
		\centering
		\fbox{\usebox{\seme}}
		\caption{Semantics of Queries}
		\label{fig:seme}
	\end{figure*}
}%

\OMIT{
\begin{figure*}
\begin{minipage}[c]{1\textwidth}
%\begin{subfigure}{1\textwidth}
% \centering
%  	\fbox{\usebox{\semn}}
 % \caption{Names}
 % \label{fig:semn}
%\end{subfigure}\\
\begin{subfigure}{1\textwidth}
  \centering
   	\fbox{\usebox{\semt}}
  \caption{Terms}
  \label{fig:semt}
\end{subfigure}
\\
\begin{subfigure}{1\textwidth}
  \centering
   	\fbox{\usebox{\semc}}
  \caption{Conditions}
  \label{fig:semc}
\end{subfigure}
\\
%\begin{subfigure}{1\textwidth}
%	\centering
%	\fbox{\usebox{\seme}}
%	\caption{Expressions}
%	\label{fig:seme}
%\end{subfigure}
\end{minipage}
\caption{Semantics of Terms and Conditions}
	%,Conditions, 
\label{fig:semantics}
\end{figure*}
}
\semfige
We define the formal semantics of \sqlra\ in
the spirit of \cite{vldb17GL,hottsql,benzaken}.

The semantics ${q}({\db})$ of a query $q$ on a complete database $\db$
is defined in Figure~\ref{fig:seme} in which we use the following
notations: For a bag $B$, we use $|B|$ to denote the number of tuples
in $B$.  For a function $f\in \RAT[\att{1},\att{2},\ldots]$ we denote
by $f[\att{1}\mapsto a_1,\ldots \att{n}\mapsto a_n]$ the value
obtained from $f$ by replacing each of $\att{i}$ with $a_i$.

\subsection{Additional expressiveness of \sqlra}

We elaborate on the expressiveness of \sqlra:
\begin{itemize}
	\item 
	To express selections $\sigma_{\theta}(R)$ with 
	composite conditions, i.e., of the form $\theta \vee \theta, \theta \wedge \theta, \neg \theta$ we use
	   \begin{enumerate}
		\item  \label{it:vee}$\sigma_{\theta_1 \vee \theta_2}(R) =
		(\sigma_{\theta_1 }(R)\cup 
		\sigma_{\theta_2}(R) ) \setminus  \sigma_{\theta_2}\left(\sigma_{\theta_1 }(R)\right)$
		\item \label{it:wedge}
		$\sigma_{\theta_1 \wedge \theta_2}(R) =
		\sigma_{\theta_2}\left( \sigma_{\theta_1 }(R)\right) $
		\item \label{it:neg}
		$\sigma_{\neg \theta}(R) =
		R \setminus \sigma_{\theta }(R)$
	\end{enumerate}
	\item 
	To express selections $\sigma_{\theta}(R)$ with
	arbitrary conditions $\theta \df f \ \omega\ g$ where $\omega\in\{=,<\}$ and $f,g\in\RAT[ \att 1, \att 2,\ldots]$ we write
	\[
	\sigma_{\att{n+1} \ \omega \ \att{n+2}} \left(\App_{g,f}(R) \right)
	\]
	where $\App_{g,f}(R) \df \App_f \left( \App_g (R)  \right)$ and $R$ is a relation with arity $n$.
	For $\omega \in \{ \le,\ge,> \}$ we use composite conditions.
		\item To express selections $\sigma_{\theta}(R)$ where $\theta \df \att{i} \in [f,g]$ with $f,g\in\RAT[\att{1},\ldots]$ we write  \[
	\sigma_{\att{i}\ge f \wedge \att{i}\le g}(R)
	\]
	we extend it naturally to open and half-open intervals.
	\item To express SQL duplicate elimination where $\arity(R) = n$
	\[
	\pi_{\att 1,\ldots \att n}\Big(\sum_{\att 1,\ldots,\att n}^{\att{n+1}}\big(\App_1(R)\big)\Big)
	\] 
	\item  To express SQL aggregates:
	\begin{enumerate}
	    \item Count aggregate $\text{COUNT}_{\att{i_1},\ldots,\att{i_k}}^{\att j}(R)$ where $\arity(R) = n$  can be written as
	    \[\pi_{\att{1},\ldots,\att{k},\att{k+1}}\left(
	    \sum_{{\att{i_1},\ldots,\att{i_k} }}^{{\att {n+1}}}\big(\App_1(R)\big) \right)
	    \]
	    In fact, the definition of COUNT is not affected by $\att{j}$ so it can be omitted.
	   \item To compose aggregations $\text{AGG1}_{\att{i_1},\ldots,\att{i_k}}^{\att j}(R)$ and $\text{AGG2}_{\att{i_1},\ldots,\att{i_k}}^{\att j}(R)$ we set
	    ${\text{AGG1}^{\att{j}},\text{AGG2}^{\att{j'}}} _{\att{i_1},\ldots,\att{i_k}}(R)$
	    as follows
	    \[
	    \pi_{\att{1},\ldots,\att{k+1},\att{2k+2}}\left(
	    \sigma_{\att{1} = \att{k+2} \wedge \cdots \wedge \att{k} = \att {2k+1} }\left(\text{AGG1}_{\att{i_1},\ldots,\att{i_k}}^{\att j}(R) \times \text{AGG2}_{\att{i_1},\ldots,\att{i_k}}^{\att j'}(R)\right)\right)
	    \]
	    \item
	    Average aggregate $\text{AVG}_{\att{i_1},\ldots,\att{i_k}}^{\att j}(R)$ where $\arity(R) = n$  can be written as
	    \[
	   \pi_{\att{k+3}}\left(\App_{\frac{\att{(k+1)}}{ \att{(k+2)}}} 
	   \left(\pi_{\att{1},\ldots,\att{k+1},\att{2k+2}}\left(
	   \sigma_{\att{1}=\att{k+2} \wedge \cdots \wedge \att{k}= \att{2k+1} }
	   \left( \text{SUM}_{\att{i_1},\ldots,\att{i_k}}^{\att j}(R) \times
	   %_{\att{1},\ldots,\att{k}}
	     \text{COUNT}_{\att{i_1},\ldots,\att{i_k}}^{\att j}(R)\right)
	     \right)
	     \right)
	     \right)
	    \]
\item
MIN aggregate $\text{MIN}_{\att{i_1},\ldots,\att{i_k}}^{\att j}(R)$ where $\arity(R) = n$ can be written as
\[
\epsilon \left( 
\pi_{\att{i_1},\ldots, \att{i_k}, j}(R)
\right)
\setminus
 \epsilon\left( \pi_{ \att{1},\ldots, \att{k+1} } \left(\sigma_{\att{k+1} > \att{k+2} } \left( (\pi_{\att{i_1},\ldots \att{i_k},\att{j}} R) \join_{\att{1},\ldots \att{k}} (\pi_{\att{i_1},\ldots \att{i_k},\att{j}} R)\right)
\right)\right)
\]
\item
MAX aggregate $\text{MAX}_{\att{i_1},\ldots,\att{i_k}}^{\att j}(R)$ where $\arity(R) = n$ can be written as 
$\text{MIN}_{\att{i_1},\ldots,\att{i_k}}^{\att j}(R)$ while replacing $>$ with $<$.
	\end{enumerate}
\end{itemize}

\newcommand{\fore}{\FO(+,\cdot, <)}
\newcommand{\br}{\bar{r}}
\newcommand{\bn}{\bar{n}}

\section{Appendix for Section~\ref{sec:qa}}

\subsection{Proof of Theorem~\ref{th:measurability}}
Assume an incomplete database $\idb$ such that
$\nulls(\idb) = \{\bot_1, \ldots, \bot_n\}$, a query $q \in \sqlra$ of
arity $m$, and a generalized answer tuple $\ba$. We will show that
$\out_{q,\idb, \circ}{(\ba, k)}$ is measurable in the PDB
$\cD_\idb = (\pdom_\idb, \salg_\idb, \prob_\idb)$ defined by $\idb$.

Let $\cV_\idb = (\pdom_V, \salg_V, \prob_V)$ be the space of
valuations of $\idb$. Let $\cB = (\bbR, \bsalg)$ be the measurable
space defined by $\bbR$ equipped with the Borel $\sigma$-algebra
$\bsalg$ over $\bbR$, and, for any finite integer $n>0$, let
$\cB^n = (\bbR^n, \bsalg^n)$ where $\bsalg^n$ is the $\sigma$-algebra
generated by the family of sets
$\cG_n = \{\sigma_1 \otimes \ldots \otimes \sigma_n~\mid~ \sigma_i \in
\bsalg \textit{ for each } i =1, \ldots, n \}$. It is well known that
$\bsalg^n$ is the $\sigma$-algebra of open sets on $\bbR^n$, for each
$n>0$.

First, we show that there exists a bijection
$\chi : \bbR^n \to \pdom_V$ such that $\chi(\sigma) \in \salg_V$, for
every $\sigma \in \bsalg^n$. Let, for each $\br = (r_1, \ldots, r_n)$
in $\bbR^n$, $\chi(\br) = v_\br$ with $v(\bot_i) = r_i$, for each
$i = 1, \ldots, n$. In other words, $\chi$ maps each element $\br$ of
$\bbR^n$ into a corresponding valuation $v_\br$ such that $v_\br$
assigns the $i$-th component of $\br$ to the $i$-th null of $\idb$,
for each $i = 1, \ldots, n$. Clearly, $\chi$ is a bijection and the
inverse $\chi^{inv}: \pdom_V \to \bbR^n $ of $\chi$ is
$(\pdom_V , \bbR^n)$-measurable since, by definition of $\salg_V$, the
preimage $(\chi^{inv})^{-1}(G)$ of each generator $G \in \cG_n$ of
$\bsalg^n$ under $\chi^{inv}$ is contained in $\salg_V$, i.e.,
$(\chi^{inv})^{-1}(G) \in \salg_V$, for each $G \in \cG_n$. Moreover,
since $\chi$ is a bijection, we have $(\chi^{inv})^{-1}(X) = \chi(X)$,
for each $X \in \bsalg^n$. Therefore, for every $\sigma \in \bsalg^n$,
$(\chi^{inv})^{-1}(\sigma) = \chi(\sigma) \in \salg_V$.

Let now
$O_\circ = \{\br\in\bbR^n~\mid~ \#(v_\br (\ba), q(v_\br(\idb))) \circ
k\}$ be the set of elements $\br \in \bbR^n$ that satisfy
$\#(v_\br (\ba), q(v_\br(\idb))) \circ k$, for $\circ \in \{<,=,>\}$.
To prove our claim, we will show the following
 
\begin{claim}
  \label{cl:meas-rn}
  The set $O_\circ$ is measurable in $\bsalg^n$, for each
  $\circ \in \{<,=,>\}$.
\end{claim}

In turn, Claim~\ref{cl:meas-rn} implies that $\chi(O_\circ)$ is
measurable in $\salg_V$, since $\chi(A) \in \salg_V$, for each
$A \in \bsalg^n$. To prove the claim of
Theorem~\ref{th:measurability}, then, we simply observe that
$\out_{q,\idb,\circ}{(\ba, k)} = \{v_\br (\idb)~\mid~ \#(v_\br (\ba),
q(v_\br(\idb))) \circ k\}$, since every $v \in \pdom_V$ is equal to a
corresponding $v_\br$ and, since $\chi$ is a bijection, we
have$\out_{q,\idb}{(\ba, k)} = \{v_\br (\idb)~\mid~ v_\br \in
\chi(O_\circ)\}$.  Then, from $\chi(O_\circ) \in \salg_V$, it follows
that $\out_{q,\idb}{(\ba, k)} \in \salg_\idb$ from the definition of
the PDB $\cD_\idb$ of $\idb$.

\subsubsection{Proof of Claim~\ref{cl:meas-rn}} We use $\fore$ for the
first-order theory of the reals, i.e., every finite first-order
formula that can be defined using and $+$, $\cdot$ as functions, the
binary relations $<$ and $=$, and constants from $\bbR$.  Formulae in
$\fore$ are interpreted over the structure
$\langle{\bbR, \Omega}\rangle$ where $\Omega$ defines the standard
interpretation of $+$, $\cdot$, $<$ and $=$ over
$\bbR$. Interpretation of formulae over such structure is as
customary. With this interpretation, $\fore$ can express also the
functions $a - b$, $\frac{a}{b}$, using suitable formulae.  For
formulae in $\fore$, we write $\phi(\bx)$ to denote the fact that
$\bx$ is the tuple of free variables in $\phi$. Given a formula
$\phi(\bx)$ with $\bx$ of length $n$, and a $n$-tuple $\br$ over
$\mathbb{R}$, we say that $\phi(\br)$ is true if the formula obtained
from $\phi(\br)$ by substituting $\bx$ with $\br$ is satisfied in
$\langle{\bbR, \Omega}\rangle$. We say that a set $S$ is definable in
$\fore$ if there exists $\phi(\bx) \in \fore$ such that $\br \in S$ if
and only if $\phi(\br)$ is true. The following proposition is a
consequence of the fact that sets definable in $\fore$ enjoy the cell
decomposition property (see, e.g., \cite{ominimal-book}).

\begin{proposition}
  \label{pr:opensets}
  Let $S \in \bbR^n$. If $S$ is definable in $\fore$ then $S$ is
  equivalent to a finite union of open set in $\bbR^n$, and sets of
  smaller dimensions.
\end{proposition}

In turn, Proposition~\ref{pr:opensets} implies that every set
$S \in \bbR^n$ that is definable in $\fore$ is Borel-measurable in
$\bbR^n$. In other words, if $S$ is definable in $\fore$ then
$S\in\bsalg^n$. To prove Claim~\ref{cl:meas-rn} then, it suffices to
show the following:

\begin{claim}
\label{cl:fo-definability}
  The set
  $O_p = \{\br\in\bbR^n~\mid~ \#(v_\br (\ba), q(v_\br(\idb))) = p \textit{ for each } p\in\bbR\}$
  is definable in $\fore$, for every $p \in \bbN$.
\end{claim}
Claim~\ref{cl:meas-rn} is then a straightforward consequence of
Claim~\ref{cl:fo-definability} and the fact that $O_>$ is equivalent
to the union of the countable family $\{O_p | p>k \in \bbN\}$, and
$O_<$ is equivalent to the union of the finite family
$\{O_p | p<k \in \bbN\}$
To prove Claim~\ref{cl:fo-definability}, we first prove the
following. Recall that the arity of $q$ is $m$ and
$|\nulls(\idb)| = |\{\bot_1, \ldots, \bot_n\}| = n$.

\begin{claim}
  \label{cl:formula-constants}
  There exists a formula in $FO(+, \cdot)$ such that
  $\phi(\br, k, \bn)$ is true if and only if
  $\#(\bn, q(v_\br(\idb))) = k$, for each $\br \in \bbR^n$,
  $k \in \bbR$, and $\bn \in \bbR^m$.
\end{claim}
  
To prove Claim~\ref{cl:formula-constants}, we will define an encoding
of $q$ and $\idb$ as a formula $\phi_{q,\idb}(\bx, k, \by) \in \fore$.
Before presenting our translation, however, we need to prove a
preliminary result. Given a query $q \in \sqlra$, the size $|q|$ of
$q$ is equal to the number of operators occurring in the expression
that defines $q$. Given a bag $B$ of tuples over a domain $\Delta$,
the active domain $adom(B)$ of $B$ is the set of elements of $\Delta$
that occur in the tuples of $B$ at least once, i.e.
$adom(B) = \{c~\mid~ c \in \bc \text{ and } \#(\bc, B) > 0\}$.

\begin{lemma}
\label{lm:answeradombound}
For every $q \in \sqlra$, there exists a finite $b \in \bbN$ such that
$|adom(q(v\idb))| \le b$, for each valuation $v$ for $\nulls(\idb)$.
\end{lemma}
\begin{proof}
  To prove the claim, we prove something stronger. Let $D$ be a
  complete databases, then
  $|adom(q(D))| \le u \cdot |q|^v \cdot |adom(D)|^{w \cdot |q|}$, for
  some finite x$u,v,w \ge 0$.  The claim follows from the fact that
  $|adom(v \idb) | = |adom(\idb)|$, for each valuation $v$ for
  $\nulls(\idb)$. We prove the claim by induction over the structure
  of $q$.

  $q = R$. $adom(q(D)) = |adom(R^D)| \le |adom(D)|$.

  $q = \sigma_\theta(q_1)$. The selection operator does not add extra
  elements in the active domain of answers. Therefore,
  $|adom(q(D))| = |adom(q_1(D))|$. By I.H., we have
  $|adom(q_1(D))| \le u_1 \cdot |q_1| ^{v_1} \cdot |adom(D)|^{w_1
    \cdot |q_1|} = u_1 \cdot (|q|-1)^{v_1} \cdot |adom(D)|^{w_1 \cdot
    (|q|-1)}$, for some finite $u_1,v_1,w_1 > 0$.

  $q = \pi_{\$i_1, \ldots, \$i_k }(q_1)$. Identical to the case
  $q = \sigma_\theta(q_1)$.

  $q = q_1 \cap q_1$. The intersection operator does not add elements
  to the active domain of the answers. Therefore,
  $|adom(q(D))| \le |adom(q_1(D))| + |adom(q_2(D))|$. By I.H., for
  some finite $u_1,v_1,w_1,u_2,v_2,w_2 > 0$, we have
  \begin{equation*}
    \begin{array}{l}
      |adom(q_1(D))| + |adom(q_2(D))|  \le
      (u_1 \cdot |q_1|^{v_1} \cdot |adom(D)|^{w_1
      \cdot |q_1|}) +
      (u_2 \cdot |q_2|^{v_2}  \cdot |adom(D)|^{w_2 \cdot |q_2|}) \le \\
      (u_1 \cdot |q_1|^{v_1} \cdot |adom(D)|^{(w_1+w_2) \cdot (|q_1|+|q_2|)}) + (u_2 \cdot |q_2|^{v_2}
      \cdot |adom(D)|^{(w_1+w_2) \cdot (|q_1|+|q_2|)}) = \\

      ( u_1 \cdot |q_1|^{v_1}+ u_2 \cdot |q_2| ^{v_2})
      \cdot |adom(D)|^{(w_1+w_2) \cdot (|q_1|+|q_2|)} \le 
      (u_1 \cdot u_2) (|q_1| + |q_2|)^{v_1 \cdot  v_2} \cdot |adom(D)|^{(w_1+w_2) \cdot (|q_1|+|q_2|)} =\\ 
      (u_1 \cdot u_2) ( |q|-1)^{v_1 \cdot  v_2}
      \cdot |adom(D)|^{(w_1+w_2) \cdot (|q|-1)} 
    \end{array}
\end{equation*}

The cases $q = q_1 \cup q_2$, $q = q_1 \times q_2$, and
$q = q_1 \setminus q_2$ are identical to the case $q = q_1 \cap q_2$.

$\App_{F}(q_1)$.  By definition, $\App_{F}$ may generate at most one
fresh new value $F(\bc)$, for each distinct tuple $\bc$ in
$q_1(D)$. Suppose that the arity of $q_1$ is $ar(q_1)$. Thus,
$|adom(q(D))| \le |adom(q_1(D))| + |adom(q_1(D))|^{ar(q_1)} \le 2
\cdot |adom(q_1(D))|^{ar(q_1)+1}$. By I.H., for some finite
$u_1,v_1,w_1 > 0$, we have
$|adom(q_1(D))| \le (u_1 \cdot |q_1|^{v_1} \cdot |adom(D)|^{w_1 \cdot
  |q_1|})$, and therefore, for some finite $u,v,w \ge 0$, we can
conclude
\begin{equation*}
  \begin{array}{l}
    |adom(q(D))| \le 
    (2 \cdot u_1 \cdot |q_1|^{v_1} \cdot |adom(D)|^{w_1 \cdot |q_1|})^{(ar(q_1)+1)} =
    u \cdot (|q|-1)^v \cdot |adom(D)|^{w \cdot  (|q|-1) }
  \end{array}
\end{equation*}

$\sum^\$i_{\$i_1, \ldots, \$i_j}(q_1)$.  By definition,
$\sum^\$i_{\$i_1, \ldots, \$i_j}$ may generate at most one value
$\sum(\{a_1, \ldots, a_m\})$, for each distinct tuple in the result of
$\pi_{\$i_1, \ldots, \$i_j}(D)$. Therefore, this case follows from the
argument used for the case $\App_F$.
\end{proof}

With Lemma~\ref{lm:answeradombound} in place, we proceed to define the
formula $\phi_{q,\idb}(\bx, k, \by)$ associated to $\idb$ and $q$. To
simplify the presentation, we use the following assumptions. In every
occurrence of $\pi_{A}(q')$ in $q$, the tuple $A$ contains the first
$l$ attributes of $q'$, for some $l \ge 0$. Similarly, for each
$\sum^\$i_{A}(q')$ in $q$, we assume that the tuple $A$ contains the
first $l$ attributes of $q'$, for some $l>0$. This assumptions are
w.l.o.g. since the general case can be recovered simply by considering
the positions of each attribute in the projection.

In what follows, given two n-ary tuples of variables and constants
$\bx, \by$, we use $(\bx = \by)$ to denote the formula
$\bigwedge_{i=1}^n (a_i = c_i)$. Moreover, given terms $a$, $b$, and
$c$, we use $(a \le b)$ for $(a < b) \vee (a = b)$, and $(a \ge b)$
for $(a > b) \vee (a = b)$.% , and $(\frac{a}{b}) = c$ for
% $(a = b \cdot c)$.
Finally, to every $\bot_i \in \nulls(\idb)$, we
associate a variable $x_i$ in $\bx$. The formula $\phi_{q,\idb}$ is
defined inductively as follows:

\begin{itemize}
\item $q = R$, for some relational symbol $R$.  Let $R^{\idb}_{vars}$
  be the bag obtained from $R^{\idb}$ by replacing each $\bot_i$ with
  $x_i$, and let $T$ be the set of distinct tuples occurring in
  $R^{\idb}_{vars}$
  \[
    \phi_{q,\idb}(\bx, k, \by) = \bigvee_{C \subseteq T}
    \Big(\bigwedge_{\bc \in C} (\by = \bc) \wedge (\sum_{\bc \in C}
    \#(\bc, R^{\idb}_{vars}) = k) \wedge \bigwedge_{\bc \in T
      \setminus A} \neg (\bc = \by)\Big)
\]
\item $q = \sigma_{\$i \circ \$j}(q_1)$, with $\circ \in \{<, =\}$.  $\phi_{q,\idb}(\bx, k, \by)$ is defined as follows:

\begin{equation*}
    \begin{array}{c@{}c}
      \exists k_1. \phi_{q_1, \idb}(\bx,
      k_1, \by) &\wedge ( ( (k > 0) \wedge (k_1 = k) \wedge  (y_i \circ y_j))) \vee \\
      &( (k = 0) \wedge ((k_1 = 0) \vee \neg(y_i \circ y_j)))
    \end{array}
  \end{equation*}

% \item $q = \varepsilon(q_1)$.  $\phi_{q,\idb}(\bx, k, \by)$ is defined as follows:
%   \begin{equation*}
%     \begin{array}{l}
%       \exists k'. \phi_{q_1, \idb}(\bx, k', \by) \wedge \left( \left((k' > 0) \wedge  (k = 1) \right) \vee 
%       ((k = 0) \wedge (k' = 0))\right)
%     \end{array}
% \end{equation*}

\item $q = q_1 \times q_2$. Suppose the arity of $q_1$ is $ar(q_1)$
  and the arity of $q_2$ is $ar(q_2)$.  $\phi_{q,\idb}(\bx, k, \by)$
  is defined as follows:
  \begin{equation*}
    \begin{array}{l}
      \exists k_1, k_2. (k = k_1 \cdot k_2) \wedge 
      (\phi_{q_1, \idb}(\bx, k_1, y_1,
  \ldots, y_{ar(q_1)}) \wedge \phi_{q_2, \idb}(\bx, k_2, y_{ar(q_1)+1}, \ldots, y_{ar(q_1)+ar(q_2)}) )
    \end{array}
\end{equation*}
 \item $q = q_1 \cup q_2$.  $\phi_{q,\idb}(\bx, k, \by)$ is defined as follows:
   \begin{equation*}
     \begin{array}{l}
       \exists k_1,k_2. (k = k_1+ k_2) \wedge 
       ( \phi_{q_1, \idb}(\bx, k_1, \by) \wedge
       \phi_{q_2, \idb}(\bx, k_2,\by) )
     \end{array}
  \end{equation*}
% \item $q = q_1 \cap q_2$.  $\phi_{q,\idb}(\bx, k, \by)$ is defined as follows:
%   \begin{equation*}
%     \begin{array}{l}
%       \exists k_1,k_2. (( \phi_{q_1, \idb}(\bx, k_1, \by) \wedge
%                               \phi_{q_2, \idb}(\bx, k_2,\by) \big)  \wedge
%                             ((k_1 \le k_2 \wedge k = k_1) \vee (k_2 \le k_1 \wedge k = k_2))

  %   \end{array}
  % \end{equation*}

\item $q = q_1 \setminus q_2$.  $\phi_{q,\idb}(\bx, k, \by)$ is defined as follows: 
  \begin{equation*}
    \begin{array}{ll}
      \exists k_1,k_2. (k = k_1- k_2)\wedge ( \phi_{q_1, \idb}(\bx, k_1, \by) \wedge
      \phi_{q_2, \idb}(\bx, k_2,\by)) 
    \end{array}
\end{equation*}
%\item
%\liat{I added this item:}
% $ q = \pi_{A}(q_1) $ where $\ar(q_1) = n$ and $A=\{i_1,\ldots,i_m\}$. We define $\phi_{q,\idb}(\bx, k, (y_{i_1},\ldots, y_{i_m}))$
% \[
% \exists k_1 \vee_{\text{ways to complete }(y_{i_1},\ldots, y_{i_m})}
% \left( 
% \phi_{q_1,\idb}(\bx, k_1, (y_1,\ldots,y_n)) \wedge
 %%compatible
 %%the rest is not
 %\right)
 % \]

\item $\App_f(q_1)$. Let $\phi_f(\by, z)$ be the formula $z =
  f(\by)$. For forumulae in $RAT$, such formula always exists.  We
  $\phi_{q,\idb}(\bx, k, \by, y_f)$ as follows:
  \begin{equation*}
    \begin{array}{ll}
      \exists z,k_1 . \phi_{q_1,\idb}(\bx, k_1, \by)  \wedge  (\phi_f(\by, z)) \wedge ( ( (k > 0) \wedge (k = k_1) \wedge ( z = y_f)) \vee
      ((k = 0) \wedge ((k_1 = 0) \vee \neg (y_f = z) ) ))
    \end{array}
  \end{equation*} 
  
\item $ q = \pi_{\$1, \ldots, \$l}(q_1) $. Assume
  $|adom(q_1(v (\idb)))| \le b$, for each $v$ for $\nulls(\idb)$
  (Lemma~\ref{lm:answeradombound}). Let $s = b^{ar(q_1) }$,
  $\bw_i = (w_1^i, \ldots, w_{ar(q_1)-l}^i)$, for each
  $i = 1, \ldots, s$, and $\bz = (z_1, \ldots, z_{ar(q_1)-l})$. We
  define
  $\phi^\pi_{(q_1, l, m)}(\bx, \by, k_\pi, k_1, \bw_1, \ldots, k_s,
  \bw_s) $ as follows:
  \begin{equation*}
    \begin{array}{l}
      \bigwedge_{i \neq j \in [1, s]} \neg (\bw_i = \bw_j) \wedge 
      \bigwedge_{i=1}^{s}  \phi_{q_1, \idb}(\bx, k_i,  \by, \bw_i) \wedge   (\sum_{i = 1}^{s}(k_i) = k_\pi) \wedge
      \neg (\exists k_z, \bz. (k_z > 0) \wedge \bigwedge_{i}^{s} \neg (\bz = \bw_i) \wedge \phi_{q_1, \idb}(\bx, k_z, \by,  \bz)) 

    \end{array}
  \end{equation*}
  The formula $\phi_{q,\idb}(\bx, k, \by)$ is defined as follows:
  $ \exists k_1, \bw_1,\ldots, k_s, \bw_s.  \phi_{(q_1, A,m)}(\bx,
  \by, k_\pi, k_1, \bw_1, \ldots, k_s, \bw_s) \wedge (k_\pi =
  k)$
  
\item $ q = \sum_{_{\$1, \ldots, \$l}}^{\$j}(q_1) $. Assume
  $|adom(q_1(v \idb))| \le b$, for each $v$ for $\idb$
  (Lemma~\ref{lm:answeradombound}), Let $s = b^{ar(q_1)}$, and
  $\bw_i = (w_1^i, \ldots, w_{ar(q_1)}^i)$, for each
$i = 1, \ldots, s$. The formula $\phi_{q,\idb}(\bx, k, \by, y')$ is
defined as follows:

  \begin{equation*}
    \begin{array}{l}
      \exists y_\Sigma, k_1, \bw_1,\ldots, k_s, \bw_s. 
      \phi^\pi_{(q_1, l)}(\bx, \by,  k_\pi, k_1,
      \bw_1, \ldots, k_s, \bw_s. ) \wedge (y_\Sigma = \sum_{i=1}^s w^{i}_{j-l} \cdot k_i) \wedge \\
      (((k = 1) \wedge (k_\pi > 0) \wedge (y' = y_\Sigma)) \vee 
      ((k = 0) \wedge ((k_\pi = 0) \vee (\neg (y' = y_\Sigma) ))))
    \end{array}
 \end{equation*}
\end{itemize}

We are ready to prove Claim~\ref{cl:formula-constants}.

\begin{proof}[Proof of Claim~\ref{cl:formula-constants}]
  
  We recall the reader that we are under the assumption that every
  occurrence of $\pi$ and $\sum$ project the first $l$ attributes of
  the sub-query, for some $l>0$ (not necesserily the same). As
  discussed above, this can be done w.l.o.g.
  We prove the claim by induction on the structure of $q$. The two
  sides of the claim will be discussed separately.

  (\underline{Case $R = q$}).  \underline{Only if part}. Suppose that
  $\phi_{q,\idb}(\br, k, \bn)$ holds for
  $\br \in \bbR^n, k \in \bbR, \bn \in \bbR^m$ and let $v_\br$ be the
  valuation defined by $\br$. Then, there exists a set of tuples
  $A = \{\bc_1, \ldots, \bc_k\}$ occurring in $ R^{\idb}$ such that
  $v_\br(\bc_i) = \bn$, for each $i = 1, \ldots, n$, and
  $\sum_{i=1}^n \#(\bc_i, v_\ba(R^\idb)) = k$. Moreover, there exists
  no tuple $\bc_z$ in $ R^{\idb}$ such that $\bc_z \not\in A$ and
  $v_\ba \bc_z = \bc$. We can conclude that
  $\#(\bc, q(v_\ba(\idb))) = k$.
  \underline{If part}. Suppose that $\#(\bn, q (v_\br(R^\idb))) = k$,
  for some $\br \in \bbR^n, k \in \bbR,\bn \in \bbR^m$. Then, there
  exists a subset $A = \{\bc_1, \ldots, \bc_n\}$ of the tuples
  occurring in $R^{\idb}$ such that $v_\br(\bc_i) = \bn$, for each
  $i = 1, \ldots, n$, and $\sum_{i=1}^n \#(\bc_i, v_\br(R^\idb)) =
  k$. From this observation, we can conclude that
  $\phi_{q,D}(\ba, k, \bc)$ holds for $\br, k, \bn$.

  (\underline{Case $R = \sigma_{\$i \circ
      \$j}(q_1)$}). \underline{Only if part}.  Suppose that
  $\phi_{q,\idb}(\br, k, \bn)$ holds for some
  $\br \in \bbR^n, k \in \bbR,\bn \in \bbR^m$. Then
  $\exists k_1. \phi_{q,\idb}(\br, k_1, \bn)$ holds, and, by I.H.,
  there exists $k_1 \ge 0$ such that $\#(\bn, q_1(v_\br\idb)) =
  k_1$. If $k > 0$, then $k = k_1$ and $\bn$ satisfies the selection
  condition $(\$i \circ \$j)$, in turn proving that
  $\#(\bn, q( v_\br \idb)) = k$. If $k = 0$, then either $k_1 = 0$ or
  $\bn$ does not satisfy the selection condition $(\$i \circ \$j)$, in
  turn proving that $\#(\bn, q, v_\br) = 0$.
  \underline{If part}. Suppose that $\#(\bn, q, v_\br(\idb)) = k$, for
  some $\br \in \bbR^n, k \in \bbR,\bn \in \bbR^m$. If $k > 0$ then
  $\#(\bn, q_1, v_\br(\idb)) = k$, and $\bn$ satisfies the selection
  condition $(\$i \circ \$j)$. If this is the case, there exists
  $k_1 \ge 0$ such that $\#(\bn, q_1 (v_\br \idb)) = k_1$, in turn, by
  I.H., proving that $\phi_{q,\idb}(\br, k, \bn)$ holds. Similarly, if
  $k = 0$ then either $\#(\bn, q_1, v_\br(\idb)) = 0$ or
  $\#(\bn, q_1, v_\br(\idb)) > 0$ and $\bn$ does not satisfy the
  selection condition $(\$i \circ \$j)$. Then, by I.H., there exists
  $k_1$ such that $\phi_{q_1,\idb}(\br, k_1, \bn)$ and either
  $k_1 = 0$ or $\bn$ does not satisfy $(i \circ j)$, in turn proving
  that $\phi_{q,\idb}(\br, k, \bn)$ holds.

  % (\underline{Case $R = \varepsilon(q_1)$}).  \underline{Only if part}.  Suppose
  % that $\phi_{q,D}(\ba, k, \bc)$ holds for some
  % $\bc \in \bbR^n, k \in \bbR, \ba \in \bbR^m$. By I.H., there exists
  % $k_1 \ge 0$ such that $\#(\bc, q_1, v_\ba(\idb)) = k_1$. If $k = 1$,
  % then $k_1 > 0$, and we can conclude that
  % $\#(\bc, q, v_\ba(\idb)) = 1$ from the definition of duplicate
  % elimination. If $k = 0$, then $k_1 = 0$ as well, proving
  % $\#(\bc, q, v_\ba(\idb)) = 0$.
  % %
  % \underline{If part}. Suppose that $\#(\bc, q, v_\ba(\idb)) = k$, for some
  % $\bc \in \bbR^n, k \in \bbR, \ba \in \bbR^m$. From the definition of
  % duplicate elimination, $k \in \{0, 1\}$. If $k = 1$, then
  % $\#(\bc, q, v_\ba(\idb)) > 0$, and, by I.H., we can conclude that
  % $\phi_{q,D}(\ba, k, \bc)$ holds. If $k = 0$, then
  % $\#(\bc, q, v_\ba(\idb)) = 0$, and, once again by I.H., we can
  % conclude that $\phi_{q,D}(\ba, k, \bc)$ holds.

  (\underline{Case $R = q_1 \times q_2$}).  \underline{Only if part}.  Suppose
  that $\phi_{q,\idb}(\br, k, \bn)$ holds for some
  $\br \in \bbR^n, k \in \bbR,\bn \in \bbR^m$. By I.H., there exists
  $k_1,k_2 \ge 0$ and tuples $\bn_1, \bn_2$ such that
  $\#(\bn_1, q_1, v_\br(\idb)) = k_1$ and
  $\#(\bn_2, q_1, v_\br(\idb)) = k_1$, and $(\bn_1, \bn_2) =
  \bn$. Moreover, $k_1 \cdot k_2 = k$, proving
  $\#(\bn, q, v_\br(\idb)) = k$.
  \underline{If part}. Suppose that $\#(\bn, q, v_\br(\idb)) = k$, for some
  $\br \in \bbR^n, k \in \bbR, \bn\in \bbR^m$. Then, there exist
  tuples $\bn_1, \bn_2$ such that $\#(\bn_1, q_1, v_\br(\idb)) = k_1$
  and $\#(\bn_2, q_1, v_\br(\idb)) = k_1$, $(\bn_1, \bn_2) = \bn$, and
  $k_1 \cdot k_2 = k$. By I.H then, we can conclude that
  $\phi_{q,\idb}(\br, k, \bn)$ holds.

  (\underline{Case $R = q_1 \cup q_2$}).  \underline{Only if part}.  Suppose that
  $\phi_{q,\idb}(\br, k, \bn)$ holds for some
  $\br \in \bbR^n, k \in \bbR, \bn \in \bbR^m$. By I.H., there exists
  $k_1,k_2 \ge 0$ such that $\#(\bn, q_1, v_\br(\idb)) = k_1$,
  $\#(\bn, q_1, v_\br(\idb)) = k_1$, and $k_1 + k_2 = k$. We can
  conclude that  $\#(\bn, q, v_\br(\idb)) = k$.
  \underline{If part}. Suppose that $\#(\bn, q, v_\br(\idb)) = k$, for some
  $\br \in \bbR^n, k \in \bbR, \bn \in \bbR^m$. Then,
  $\#(\bn, q_1, v_\br(\idb)) = k_1$ and
  $\#(\bn, q_2, v_\br(\idb)) = k_2$, with $k_1 + k_2 = k$. By
  I.H then, we can conclude that $\phi_{q,\idb}(\br, k, \bn)$ holds.

  (\underline{Case $R = q_1 \setminus q_2$}).  \underline{Only if part}.  Suppose
  that $\phi_{q,\idb}(\br, k, \bn)$ holds for some
  $\br \in \bbR^n, k \in \bbR, \bn \in \bbR^m$. By I.H., there exists
  $k_1,k_2 \ge 0$ such that $\#(\bn, q_1, v_\br(\idb)) = k_1$,
  $\#(\bn, q_1, v_\br(\idb)) = k_1$, and $k = k_1 - k_2$. We can
  conclude that $\#(\bn, q, v_\br(\idb)) = k$.
  \underline{If part}. Suppose that $\#(\bn, q, v_\br(\idb)) = k$, for some
  $\br \in \bbR^n, k \in \bbR, \bn \in \bbR^m$. Then,
  $\#(\bn, q_1, v_\br(\idb)) = k_1$ and
  $\#(\bn, q_2, v_\br(\idb)) = k_2$, with $k =k_1 - k_2$. By
  I.H then, we can conclude that $\phi_{q,\idb}(\br, k, \bn)$ holds.
  
  (\underline{Case $q = \App_f(q_1)$}).  \underline{Only if part}.
  Suppose that $\phi_{q,\idb}(\br, k, \bn, r')$ holds for some
  $\br \in \bbR^n, k,r \in \bbR, \bn \in \bbR^m$. By I.H., there
  exists $k_1$ such that $\#(\bn, q_1, v_\br(\idb)) = k_1$. If
  $k > 0$, then $k = k_1$ and $r' = f(\bn)$. If $k = 0$, then $k = 0$
  or $r' \neq f(\bn)$. Thus, we can conclude that
  $\#(\bn, q, v_\br(\idb)) = k$.
  \underline{If part}. Suppose that
  $\#((\bn, r'), q, v_\br(\idb)) = k$, for some
  $\br \in \bbR^n, r,k \in \bbR, \bn \in \bbR^m$. If $k > 0$, then,
  $\#(\bn, q_1, v_\br(\idb)) = k$ and $r' = f(\bn)$. If $k = 0$, then
  either $\#(\bn, q_1, v_\br(\idb)) = 0$ $r' \neq f(\bn)$.  By I.H
  then, we can conclude that $\phi_{q,\idb}(\br, k, \bn, r')$ holds.
 
  (\underline{Case $R = \pi_{\$1, \ldots, \$l} q_1 $}).  To prove this
  part of the claim, we first need to prove a preliminary
  result. Suppose that $q_1$ is $m_1$-ary. Let
  $\{\bc_1, \ldots, \bc_p\} $ be a set of distinct tuples in
  $\bbR^{m_1 - l }$, and let $k_\pi, k_1, \ldots, k_p \ge 0$,
  $\br \in \bbR^n$, and $\bn \in \bbR^l$.
  \begin{claim}
    \label{cl:group-formula}
    The formula
    $\phi^\pi_{(q_1, l)}(\br, \bn, k_\pi, k_1, \bc_1, \ldots, k_p,
    \bc_p)$ holds if and only if each of the  following holds:
    \begin{itemize}
    \item $\#( (\bn, \bc_i), q_1(v_\br(\idb))) = k_i$, for each
      $i = 1, \ldots, p$;
    \item $\sum_{i=1}^p k_i = k_\pi$; and
    \item There exists no $\bc_z$ s.t. $\bc_z \neq \bc_i$, for
      each $i = 1, \ldots, p$, and
      $\#( (\bn, \bc_z), q_1(v_\br(\idb)) ) > 0$.
    \end{itemize}
  \end{claim}
  \begin{proof}
    We prove Claim~\ref{cl:group-formula} in two parts. \underline{If
      part}. Let $\{\bc_1, \ldots, \bc_p\}$ be a set of distinct
    tuples in $\bbR^{m_1 - l }$, and suppose that
    $\#((\bn, \bc_i), q_1(v_\br(\idb))) = k_i$, for
    $i = 1, \ldots, p$, $\sum_{i=1}^p k_i = k_\pi$ and there exists
    no $\bc_z$ such that $\bc_z \neq \bc_i$, for each
    $i = 1, \ldots, p$, and
    $\#( (\bn, \bc_z), q_1(v_\br(\idb)) ) > 0$. By definition,
    $\{\bc_1, \ldots, \bc_p\}$ is a set of distinct tuples, and
    therefore they satisfy
    $\bigwedge_{i \neq j \in [1, p]} \neg (\bc_i = \bc_j)$. Moreover, by I.H.,
    $\phi_{q_1, \idb}(\br, k_i, \bn, \bc_i)$ holds for each
    $i = 1, \ldots, p$. To conclude the proof, we observe that, since
    there exists no $\bc_z$ such that
    $\#( (\bn, \bc_z), q_1(v_\br(\idb)) ) > 0 $ and
    $\bc_z \not\in \{\bc_1, \ldots, \bc_p\}$, the formula
    $\neg (\exists k_z, \bz. (k_z > 0) \wedge \bigwedge_{i}^{s} \neg
    (\bz = \bc_i) \wedge \phi_{q_1, \idb}(\bx, k_z, \by, \bz))$ is
    satisfied.
    \underline{Only if part}. Suppose that
    $\phi^\pi_{(q_1, l)}(\br, \bn, k_\pi, k_1, \bc_1, \ldots, k_p,
    \bc_p)$ holds for a set of distinct tuples
    $\{\bc_1, \ldots, \bc_p\}$ in $\bbR^{m_1 - l }$. By I.H.,
    $\#( (\bn, \bc_i), q_1(v_\br(\idb)) ) = k_i $, for each
    $i = 1, \ldots, p$, and therefore
    $\sum_{i=1}^p \#( (\bn, \bc_i), q_1(v_\br(\idb)) ) =
    k_\pi$. Moreover, since
    $\neg (\exists k_z, \bz. (k_z > 0) \wedge \bigwedge_{i}^{s} \neg
    (\bz = \bc_i) \wedge \phi_{q_1, \idb}(\bx, k_z, \by, \bz))$ holds,
    again by I.H., we have that there exists no $\bc_z$ such that
    $\bc_z \neq \bc_i$, for each $i = 1, \ldots, p$, and
    $\#( (\bn, \bc_z), q_1(v_\br(\idb)) ) > 0$.
  \end{proof}

  With Claim~\ref{cl:group-formula} in place, we are now ready to
  prove our result. Once again, we discuss the two sides of the claim
  separately. \underline{Only if part}.  Suppose that
  $\phi_{q,\idb}(\br, k, \bn)$ holds for some
  $\br \in \bbR^n, k \in \bbR, \bn \in \bbR^m$. Then, from
  Claim~\ref{cl:group-formula}, we have that,
  $\sum_{i=1}^p \#( (\bn, \bc_i), q_1(v_\br(\idb)) ) = k_\pi$, for
  some set of tuples $\{\bc_1, \ldots, \bc_p\}$, and there exists no
  $\bc_z$ such that $\bc_z \neq \bc_i$, for each $i = 1, \ldots, p$,
  and $\#( (\bn, \bc_z), q_1(v_\br(\idb)) ) > 0$ and $k_\pi =
  k$. We can conclude that $\#( \bn, q(v_\br(\idb)) ) = k$. 
  \underline{If part}. Suppose that $\#(\bn, q(v_\br(\idb))) = k$, for
  some $\br \in \bbR^n, k \in \bbR, \bn \in \bbR^m$. Then,
  $\sum_{\be \in E} \#( (\bn, \be), q_1(v_\br(\idb)) ) = k$, where $E$
  is the set of all the tuples of arity $ar(q_1)-l$. From
  Lemma~\ref{lm:answeradombound}, we know that there are at most
  $s = b^{ar(q_1)-l}$ distinct tuples $\be$ such that
  $\#( (\bn, \be), q_1(v_\br(\idb)) ) > 0$. Suppose that
  $\{\be_1, \ldots, \be_p\}$, with $p\le s$, are those tuples, then
  $\sum_{\be \in E} \#( (\bn, \be), q_1(v_\br(\idb)) ) = \sum_{i=
    1}^p\#( (\bn, \be_i), q_1(v_\br(\idb)) ) = k$. From
  Claim~\ref{cl:group-formula}, then, the following formula holds:
  $\phi^\pi_{(q_1,l)}(\br, \bn, k, k_1, \be_1, \ldots, k_p, \be_p, 0,
  \bc_1, \ldots, 0, \bc_{s-p})$, where $\bc_1, \ldots, \bc_{s-p}$ are
  pairwise distinct from $\be_1, \ldots, \be_p$. We can conclude that
  $\phi_{q,\idb}(\br, k, \bn)$ holds.

  (\underline{Case $R = \sum^{\$j}_{\$1, \ldots, \$l}q_1$}). 
  \underline{Only if part}.  Suppose that
  $\phi_{q,\idb}(\br, k, \bn, r')$ holds for some
  $\br \in \bbR^n, k \in \bbR, \bn \in \bbR^m$.  Then, from
  Claim~\ref{cl:group-formula}, there exists a set of tuples
  $E = \{\be_1, \ldots, \be_p\}$ such that
  $\sum_{i=1}^p \#( (\bn, \be_i), q_1(v_\br(\idb)) ) = k_\pi$, and
  there exists no $\be_z$ such that $\be_z \not\in E$ and
  $\#( (\bn, \bc_z), q_1(v_\br(\idb)) ) > 0$. Let $B$ be the bag such
  that $\#(\be, B) = \# ((\bn, \be), q_1(v_\br(\idb)))$, for each
  $\be \in E$.  If $k = 1$, then $|B| = k_\pi > 0$ and $r'$ is
  equal to the sum of the $(j-l)$-th component of the tuples in $B$, and
  we can conclude that $\#(\bn, r', q(v_\br(\idb)) ) = k$. If $k = 0$,
  thene either $|B| = k_\pi = 0$ or $r'$ is not equal to the sum of
  the $(j-l)$-th component of the tuples in $B$. Again, we can conclude
  that $\#(\bn, r', q(v_\br(\idb)) ) = k = 0$.
  \underline{If part}. Suppose that
  $\#((\bn,r'), q(v_\br(\idb))) = k$, for some
  $\br \in \bbR^n, k \in \bbR, \bn \in \bbR^m$.  Let $B$ be the bag
  such that, for each $ar(q_1)-l$-ary tuple $\be$,
  $\#(\be, B) = \# ((\bn, \be), q_1(v_\br(\idb)))$. From
  Lemma~\ref{lm:answeradombound}, we know that there are at most
  $s = b^{n-m}$ distinct tuples in $B$. Let
  $E = \{e_1, \ldots, e_p\}$, with $p \le s$, be the set of distinct
  tuples occurring in $B$. From Claim~\ref{cl:group-formula}, we can
  conclude that there exist $k_\pi, k_1, \ldots, k_l \ge 0$ and tuples
  $\bc_1, \ldots, \bc_{s-p}$ in $\bbR^{ar(q_1)-l}$ such that the
  following holds:
  $\phi^\pi_{(q_1,l)}(\br, \bn, k, k_1, \be_1, \ldots, k_p, \be_p, 0,
  \bc_1, \ldots, 0, \bc_{s-p})$. We now observe that, by definition,
  $k \in \{0, 1\}$.  If $k = 0$, then either $|B| = k_\pi = 0$ or $r'$
  is not equal to the sum of the $(j-l)$-th component of the tuples in
  $B$.  We can conclude that $\phi_{q,\idb}(\br, k, \bn,c')$ holds. If
  $k = 1$, then $|B| = k_\pi >0$ and $r'$ is equal to the sum of the
  $(j-l)$-th component of the tuples in $B$.  Once again, we can
  conclude that $\phi_{q,\idb}(\br, k, \bn,c')$ holds.
\end{proof}

Next, we prove that also the condition defined by a generalized answer
tuple can be encoded into a formula in $\fore$. Recall that $\ba$ is a
generalized answer tuple.

\begin{claim}
  \label{cl:ans-tup-fo} There exists a formula
  $\psi_\ba(\bx, \by) \in \fore$ such that $\psi_\ba(\br, \bn)$ holds
  if and only if $\bn$ is consistent with $v_\br(\ba)$, for each
  $\br \in \bbR^n$, and $\bn \in \bbR^m$.
\end{claim}
\begin{proof}
  Let $\ba = ( a_1, \ldots, a_m)$. We analize the case where each
  $a_i$ defines an open interval, the other cases follow
  straightforwardly.  For each $a_i = (l_i, u_i)$, let
  $\psi_i(\br, \bn)$ be the formula
  $y_i > l_i^{vars} \wedge y_i < u_i^{vars}$, where $l_i^{vars}$ and
  $u_i^{vars}$ are the terms obtained from $l_i$ and $u_i$,
  respectively, by replacing, for each $i = 1, \ldots, n$, each
  occurrence of $\bot_i$ with $x_i$. We define
  $\psi_\ba(\bx, \by) = \bigwedge_{i=1}^n \psi_i(\bx, \by)$. By
  definition, $\psi_i(\br, \bn)$ is true for $\br \in \bbR^n$ and
  $\bn \in \bbR^m$ if and only if $r_i \le v_\br(u_i)$ and
  $r_i \ge v_\br(l_i)$.  In turn, the latter implies that
  $\psi(\br, \bn)$ is true if and only if $\bn$ is consistent with
  $v_\br(\ba)$.
\end{proof}

We are finally ready to prove Claim~\ref{cl:fo-definability}. Assume
that, $|adom(q(v (\idb)))| \le b$, for each $v$ for $\nulls(\idb)$
(Lemma~\ref{lm:answeradombound}), and let $s = b^{m}$. Let
$\phi_p(\bx)$ be the following formula:

\[ \exists \by_1, k_2 \ldots \by_s, k_s. \bigwedge_{i \neq j \in [1,
    s]} \neg (\by_i = \by_j) \wedge \bigwedge_{i=1}^s
  (\phi_{q,\idb}(\bx, k_i, \by_i) \wedge \psi_\ba(\bx, \by_i)) \wedge
  (\sum_{i=1}^s k_i = p) \wedge \neg ( \exists \bz,
  k_z. \phi_{q,\idb}(\bx, k_z, \bz) \wedge \psi_\ba(\bx, \bz) \wedge
  \bigwedge_{i=1}^s \neg (\bz = \by_i))\]

where $\phi_{q,\idb}(\bx, z, \by)$ is the formula defined in
Claim~\ref{cl:formula-constants}, and the $\psi_\ba(\bx, \by)$ is the
formula defined in Claim~\ref{cl:ans-tup-fo}.

To conclude the prove, we observe the following.  The formula
$\phi_p(\br)$ is true for some $\br \in \bbR^n$ if and only if there
exist distinct $\bn_1, \ldots, \bn_s \in \bbR^m$ such that
$\sum_{i=1}^s \#(\br, q(v_\br(\idb))) = k$
(Claim~\ref{cl:formula-constants}), each $\bn_i$ is consistent with
$v_\bn(\ba)$ (Claim~\ref{cl:ans-tup-fo}), and there exists no
$\bc \neq \bn_i$, for each $i = 1, \ldots s$, such that
$\#(\br, q(v_\br(\idb))) = k$ (Lemma~\ref{lm:answeradombound}).  We
can conclude that $O_p$ is definable in $\fore$ using $\phi_p(\bx)$,
for each $p \in \bbN$.

\subsection{Proof of Proposition~\ref{pr:transc}}
Assume a database schema $\{R_{/1}\}$ and $q = \sigma_{\$1 <
  1}(R)$. For every $k > 0$, $R^{\idb_k}$ consists of $k$ occurrences
of the tuple $(\bot)$, where $\bot$ is defined by the exponential
distribution with parameter $\lambda = 1$, i.e., the density function
of the distribution is $e^{-x}$.

Case $\otldef{q, =, k}$.  The cumulative distribution function of the
distribution that defines $\bot$ is $F(x) = 1 - e^{-x}$. In turn, the
latter implies that $P(\bot < 1) = \otldef{q, =, k}(\idb, ()) = 1 - e^{-1}$. We
now observe that $e^{-1}$ is a trancendental number, and every
non-constant polynomial function of one variable yields a
transcendental value if the value of the argument is
transcendental. We can conclude that $1 - e^{-1}$ is transcendental.

Case $\otldef{q, >, k}$. Let $R_> = R^{\idb_k} \cup \{(\bot)\}$ and
$\idb_>$ be the instance of $\cS$ where $R^{\idb_>} = R_>$. Clearly,
$P(\bot < 1) = \otldef{q, >, k}(\idb_>, ()) = 1 - e^{-x}$. For what we
stated above, we can conclude that $\otldef{q, >, k}(\idb_>, ())$ is
transcendental.

Case $\otldef{q, <, k}$. First, we observe
$\otldef{q, <, k}(\idb_>, ()) = 1 - \otldef{q, >, k}(\idb_>, ()) -
\otldef{q, =, k}(\idb_>, ())$.  Moreover,
$\otldef{q, =, k}(\idb_>, ()) = 0$, since $q(v(\idb_>))$ contains either
$k+1$ tuples or none at all. We can conclude
$\otldef{q, <, k}(\idb_>, ()) = e^{-x}$, which is transcendental.

\subsection{Proof of Theorem~\ref{th:shp-hardness}}

%\repeatresult{theorem}{\ref{th:shp-hardness}}{\thmBQAsharpp}

To prove the claim, for each $\circ \in \{<,=,>\}$ and $k\ge0$, we
show a reduction from $\#DNF$ i.e., the problem of counting the number
of satisfying assignments for a given 3-DNF formula.

We first recall some standard notions. Assume a countable set of
symbols $Vars$ that we call propositional variables.  A literal over
$Vars$ (simply, literal) is either a propositional variable
$x \in Vars$ or its negation $\neg x$, and an $n$-clause is an
$n$-tuple of literals.  We call $n$-formula a finite set of
$n$-clauses. Given a formula $\phi$, we use $Vars(\phi)$ to denote the
set of variables used in $\phi$. An assignment for a $Vars(\phi)$ is a
mapping $\alpha : Vars(\phi) \to \{0,1\}$.

A 3-DNF formula $\phi$ is a $3$-formula. We say that an assignment
$\alpha$ satisfies a 3-DNF formula if there exists at least one clause
$(\ell_1, \ell_2, \ell_3) \in \phi$ such that, for each
$i \in \{1, 2, 3\}$, $\alpha(\ell_i) = 1$, if $\ell_i$ is a
propositional variable, and $\alpha(\ell_i) = 0$, if $\ell_i$ is the
negation of a propositional variable. The $\# DNF$ problem asks, given
a 3-DNF $\phi$, for the number of assignments for $Vars(\phi)$ that
satisfy $\phi$. This problem is known to be $\# P$-complete.

We remind the reader that a problem $A$ is $\# P$-hard if there exists
a polynomial-time Turing reduction from every problem in $\# P$ to
$A$. A polynomial-time Turing reduction from a function problem $A$ to
a function problem $B$ is a Turing machine with an oracle for $B$ that
solves $A$ in polynomial time.  We proceed to show one such reduction
from $\# DNF$ to $\otlprob$, for each $k\ge0$ and
$\circ \in \{<, =, >\}$.
   
To this end, we start by presenting an encoding of 3-DNF formulae that
is common among all reductions.  Let $\phi$ be a 3-DNF formula with
$Vars(\phi) = \{x_1, \ldots, x_n\}$ and clauses $\{c_1, \ldots,
c_m\}$. We assume literals $\ell_{1,j}, \ell_{2,j}, \ell_{3,j}$, for
each clause $j = 1, \ldots, m$. Finally, we will use $\# \phi$ to
denote the number of satisfying assignments for $\phi$.
   We proceed to define the elements of our encoding.

   \begin{itemize}
   \item We fix a database schema $\cS$ containing relation schemas
     $R_{/7}$;

   \item We associate the null symbol $\bot_i$ to the variable
     $x_i \in Vars(\phi)$, for each $i = 1, \ldots, n$. Each of these
     nulls is defined by the uniform distribution with deinsity
     function $F(x) = 0.5$, for $x \in (0, 2)$, $F(x) = 0$, otherwise.
   \item To each literal $l$, we associate the pair
     $t(\ell) = (1, \bot_i)$, if $\ell = x_i$, and
     $t(\ell) = (\bot_i, 1)$, if $\ell = \neg x_i$;
   \item To each clause
     $c_j = \{\ell_{1,j}, \ell_{2,j}, \ell_{3,j}\}$, we associate the
     tuple $t(c_j)$ defined as
     $(t(\ell_{1,j}), t(\ell_{2,j}), t(\ell_{3,j}), j)$, and define
     the instance $R_\phi$ of $R_{/7}$ as the bag $t(c_j)$, for each
     $j = 1, \ldots, m$.
   \end{itemize}

   Clearly, we can obtain $R_\phi$ from $\phi$ in time polynomial in
   the size of $\phi$. Let $R_w$ be the instance of $R_{/7}$
   containing $k$ occurrences of the tuple $(3,4,3,4,3,4,m+1)$, and
   let $\idb_k$ be the incomplete database for $\cS$ with
   $R^{\idb_k} = R_\phi \cup R_w$. Finally, we define the Boolean
   query $q$ as follows:
   \[q = \pi_\emptyset (\sigma_{(\$1 < \$2)}(\sigma_{(\$3 <
       \$4)}(\sigma_{(\$5 < \$6)} (R)))) \]

   From now on, we analyze the three cases of $\circ \in \{<, =, >\}$
   separately.

   \underline{Case $\otldef{q, >, k}$} We present a reduction from
   $\# DNF$ to $\otldef{q,>, k}$, for every $k \ge 0$. 
   
   To show our reduction, we first establish the following claim:

   \begin{claim}
     \label{cl:sh-p-hard-multi}
     $\# \phi = \otldef{q,>, k}(\idb_k, ()) \cdot 2^{|Vars(\phi)|}$
   \end{claim}
   We proceed to prove a relevant property of the query $q$.  A
   $(1,2)$-partition $f$ for $\nulls(\idb_k)$ is a function that
   assigns either the interval $(0,1)$ or the interval $(1, 2)$ to
   each $\bot \in \nulls(\idb_k)$, i.e., $f(\bot) = (0, 1)$ or
   $f(\bot) = (1, 2)$, for each $\bot \in \nulls(\idb_k)$. We say that
   a valuation $v$ is within a $(1,2)$-partition $f$ if
   $v(\bot) \in f(\bot)$, for each $\bot \in \nulls(\idb_k)$. The
   query $q$ is invariant within $(1,2)$-partitions, as the following
   claim shows.

   \begin{claim}
     \label{cl:invariant-partitions}
     Let $f$ be a $(1,2)$-partition for $\null(\idb_k)$. Then, for
     each $v,v'$ within $f$ we have
     $\#( (), q(v(\idb_k))) =\#( (), q(v' (\idb_k)))$
   \end{claim}

   \begin{proof}
     The claim is a straightforward consequence of  the fact that
     the selection conditions in $q$ selects tuples that satisfy each
     of the following conditions: $(\$1 < \$2)$, $(\$3 < \$4)$,
     $(\$5 < \$6)$.  Therefore, the elements of $R_w$ are selected
     independently from the valuation. Consider now the tuples in
     $R_\phi$. By definition, each $\bc \in R_\phi$ is of the form
     $(t(\ell_{1,j}), t(\ell_{2,j}), t(\ell_{3,j}), j)$, where each
     $t(\ell_{i,j})$ is a pair of the form $(\bot, 1)$ or $(1,
     \bot)$. It is easy to see that a pair $(1, \bot)$ satisfies the
     selection conditions in $q$ for each $v$ such that $v(\bot)$ is in
     the interval $(1, 2)$, and it does not satisfy the conditions for
     each $v$ such that $v(\bot)$ is in the interval $(0,
     1)$. Similarly, a pair $(\bot, 1)$ satisfies the selection
     conditions in $q$ for each $v$ such that $v(\bot)$ is in the
     interval $(0, 1)$, and it does not satisfy the conditions for
     each $v$ such that $v(\bot)$ is in the interval $(1, 2)$. The
     claim follows straightforwardly.
 \end{proof}

 A direct consequence of Claim~\ref{cl:invariant-partitions} is that
 either $\#( (), q(v(\idb_k)))>k$, for every valuation $v$ within a
 $(1,2)$-partition $f$, or $\#( (), q(v(\idb_k)))=k$, for every
 valuation $v$ within $f$. This is becasue, as we said, the $k$
 occurrences of the tuple $(3,4,3,4,3,4, m+1)$ satisfy the selection
 conditions in $q$ independently from the valuation considered.
 We proceed to show that $\otldef{q, >, k}((), \idb_k)$ can be defined
 in terms $(1,2)$-partitions. Let $\# V$ be the total number of
 $(1, 2)$-partitions for $\nulls(\idb_k)$ and let $\# q$ be the number
 of such partitions $f$ where $\#( (), q(v(\idb_k)))>k$, for each
 valuation $v$ within $f$.

 \begin{claim}
   $\otldef{q, >, k}(\idb_k, ()) = \frac{\# q}{\# V}$
 \end{claim}
 \begin{proof}
   Since every tuple in $R^{\idb_k}$ is uniquely identified by its
   last attribute, every valuation for $\nulls(\idb_\phi)$ defines a
   distinct database. By definition then, every $(1,2)$-partition $f$
   defines a distinct set of databases
   $D_f = \{ v(\idb_\phi)~\mid~ v \text{ is within } f\}$, and
   $D_f \cap D_{f'} = \emptyset$, for each $f \neq f'$. 

   We now observe that, from Claim~\ref{cl:invariant-partitions}, we
   know that $\#( (), q(v(\idb_k)))>k$ if and only if $v$ is within a
   $(1,2)$-partition $f$ s.t., for at least one tuple $\bc \in R_\phi$
   and one $v'$ within the partition, $v'(\bc)$ satisfies the
   selection conditions in $q$. Let $F$ be the set of such
   $(1,2)$-partitions, then $\otldef{q, >, k} = \sum_{f \in F} P(D_f)$
   since each $D_f$ is a measurable set in the PDB defined by
   $\idb_k$. To conclude the proof, we observe that
   $P(D_f) = P(D_{f'}) =\frac{1}{2^{n}}$, for each $f,f' \in F$, and
   $\# V = 2^n$, where $n = |\nulls(\idb_k)|=|Vars(\phi)|$ is the
   number of null values in $\idb_k$. This is because, by definition,
   the set $V_f$ of valuations within a given $(1, 2)$-partition $f$
   has probability $\prob_V(V_f) = \frac{1}{2^{n}}$, where $\prob_V$
   is the probability function defined in the valuation space of
   $\idb_k$.  We can conclude that
   $\otldef{q, >, k}(\idb_k, ()) = \frac{\# q}{\# V}$.
\end{proof}

We now show that $\#q$ indeed represents $\# \phi$.

\begin{claim}
  $\#q = \# \phi$
\end{claim}
\begin{proof}
  To prove the claim we build a bijection $g$ from the set of
  $(1,2)$-partitions for $\nulls(\idb_k)$ to the set of assignments
  for $\phi$ such that $g(f)$ satisfies $\phi$ if and only if
  $\#( (), q(v(\idb_k)))>k$, for each $v$ within $f$.  Let $f$ be a
  $(1, 2)$-partitions, we define the assignment $\mu_f$ for
  $Vars(\phi)$ such that $\mu_f(x_i) = 1$ if $f(\bot_i) \in (1, 2)$,
  and $\mu_f(x_i) = 0$ if $f(\bot_i) \in (0, 1)$. Clearly, the mapping
  $g$ s.t. $g(f) = \mu_f$ is a bijection. We proceed to prove that
  $g(f)$ satisfies $\phi$ if and only if $\#( (), q(v(\idb_k)))>k$, for
  each $v$ within $f$. % In what follows, we will use $\theta$ to denote
  % the selection condition
  % $(\$1 < \$2 \wedge \$3 < \$4 \wedge \$5 < \$6)$ of $q$.
  Clearly, $\#( (), q(v(\idb_k)))>k$ if there exists at least one tuple
  $\bc$ in $R^{\idb_k}$ such that all the following conditions hold
  for $v(\bc)$: $(\$1 < \$2)$, $(\$3 < \$4)$, and $(\$5 < \$6)$.

  (if part) Suppose that, for each $v$ within $f$, we have
  $\#( (), q(v(\idb_k)))>k$. Then, there exists a tuple $\bc$ in
  $R^{\idb_k}$ such that the aforementioned conditions hold for
  $v(\bc)$. By construction, there exists a clause
  $k = \{\ell_1, \ell_2, \ell_3\}$ in $\phi$ such that $t(k) =
  \ba$. Consider the case of $\ell_1 = 1$. If $\ell_1 = x_p$, then
  $t(\ell_1) = ( 1 , \bot_p)$ and, since $v(\bc)$ satisfies
  $(\$1 < \$2)$, we can conclude that $1 < v(\bot_p)$, for each $v$
  within $f$. By definition, then, $\mu_f(x_p) = 1$. Similarly,
  $\ell_1 = \neg x_p$, then $t(\ell_1) = ( \bot_p, 1)$ and, since
  $v(\bc)$ satisfies $(\$1 < \$2)$, we can conclude that
  $v(\bot_p) < 1$, for each $v$ within $f$. By definition, then,
  $\mu_f(x_p) = 0$. The case of $\ell_2$ and $\ell_3$ are identical,
  and therefore we can conclude that $\mu_f$ satisfies the clause $k$
  and therefore $\phi$.

  (only if part) Assume an assignment $\mu_f$ that satisfies
  $\phi$. Then, $\mu_f$ satisfies one of its clauses
  $k = \{\ell_1, \ell_2, \ell_3\}$. By construction, there exists a
  tuple $\bc = t(k)$ in $R^{\idb_k}$.  Consider the case of
  $\ell_1$. If $\ell_1 = x_p$, then $t(\ell_j) = ( 1 , \bot_p)$ and,
  since $\mu_f$ satisfies $k$, we can conclude that $\mu_f(x_p) = 1$,
  and therefore we have $1 < v(\bot_p)$, for each $v$ within
  $f$. Similarly, if $\ell_1 = \neg x_p$, then
  $t(\ell_j) = ( \bot_p, 1)$ and, since $\mu_f$ satisfies $k$, we can
  conclude that $\mu_f(x_p) = 0$, and therefore we have
  $v(\bot_p) < 1$, for each $v$ within $f$. The case of $\ell_2$ and
  $\ell_3$ are identical, and therefore we can conclude that
  $\#( (), q(v(\idb_k)))>k$, for each $v$ within $f$.
 \end{proof}

 To use Claim~\ref{cl:sh-p-hard-multi} in a polynomial-time reduction,
 we need to prove that $\#V \cdot \otldef{q, >, k}$ can be computed in
 time polynomial w.r.t. $\phi$. To this end, we assume that the output
 of the oracle for $\otldef{q, >, k}$ is represented in binary using
 floating point notation. Observe that $\# V = 2^{|Vars(\phi)|}$, and
 therefore $\otldef{q, >, k} = \frac{\# q}{\# V}$ can be represented
 with a number of bits that is logaritmic with
 $\# q = \# \phi \le 2^{|Vars(\phi)|}$. Moreover,
 $\#V \cdot \otldef{q, >, k}$ can be computed by shifting
 $\otldef{q, >, k}$ left $|Vars(\phi)|$ bits which, in turn, can be
 done in time polynomial in $|\phi|$.

 \underline{Case $\otldef{q, <, k}$} We simply observe that
 $\otldef{q, <, k}(\idb, \ba) = 1 - \otldef{q, >, k}(\idb, \ba)$, for
 every $q \in \sqlra$ and $k \ge 0$.  With notation as above, then,
 $ \# \phi = \# V \cdot ( 1 - \otldef{q, <, k}(\idb_k, \ba) ) = \# V -
 \# V \cdot \otldef{q, <, k}(\idb_k, \ba) ) )$. Since
 $\otldef{q, >, k}(\idb_k, \ba) $ can be represented with polyomially
 many bits w.r.t. $\phi$, also $\otldef{q, <, k}(\idb_k, \ba) $ can.
 We now observe that $\# V = 2^{Vars(\phi)}$, and therfeore its binary
 representation can be computed in time polynomial in $|\phi|$. We can
 conclude that $\# V - \# V \cdot \otldef{q, <, k}(\idb_k, \ba) ) ) $
 can be computed in time polinomial in $|\phi|$, and the claim
 follows.

 \underline{Case $\otldef{q, =, k}$} For this case, we observe that
 $\otldef{q, <, k}(\idb, \ba) = \sum_{i=0}^{k-1}\otldef{q, =, i}(\idb,
 \ba) $, since $\{\out_{q, =, \idb}(\idb, i)~\mid~ k>0\}$ are pairwise
 disjoint. Since $k$ is not part of the input, with an oracle for
 $\otldef{q, =, k}$ we can solve $\otldef{q, <, k}$ summing up the
 results of a constant number of oracle calls. With notation as above,
 we have

 \[\# \phi = \# V - \# V \cdot \otldef{q, <, k}(\idb_k, \ba) ) = \# V
   - \# V \cdot (\sum_{i=1}^{k-1} \otldef{q, =, k}(\idb_k, i))\]

 We now observe that
 $\otldef{q, =, k}(\idb_k, i) = k \cdot \frac{1}{2^{|Vars(\phi)|}}$,
 for some integer $k \in [0, 2^{|Vars(\phi)|}]$ due to
 Claim~\ref{cl:invariant-partitions}. In turn, the latter implies that
 $\otldef{q, =, k}(\idb_k, i)$ can be represented with a number of
 bits that is polynomial in $|\phi|$. We can conclude that
 $\sum_{i=1}^{k-1} \otldef{q, =, k}(\idb_k, i)$ can be computed in
 time polynomial w.r.t. $|\phi|$, and the claim follows.

\subsection{Proof of Theorem~\ref{th:thresholdnphardnew}}

To prove the claim, for each $\circ \in \{<,=,>\}$ and $k\ge0$, we
show a reduction from $SAT$ for $3CNF$ formulae, i.e., the problem of
checking satisfiability of propositional formulae in conjunctive
normal form with $3$ literals each clause.

We first recall some standard notions. Assume a countable set of
symbols $Vars$ that we call propositional variables.  A literal over
$Vars$ (simply, literal) is either a propositional variable
$x \in Vars$ or its negation $\neg x$, and an $n$-clause is an
$n$-tuple of literals.  We call $n$-formula a finite set of
$n$-clauses. Given a formula $\phi$, we use $Vars(\phi)$ to denote the
set of variables used in $\phi$. An assignment for a $Vars(\phi)$ is a
mapping $\alpha : Vars(\phi) \to \{0,1\}$.

A 3-CNF formula $\phi$ is a $3$-formula. We say that an assignment
$\alpha$ satisfies a 3-CNF formula if for every clause
$(\ell_1, \ell_2, \ell_3) \in \phi$ there exists at least one
$i \in \{1, 2, 3\}$, such that $\alpha(\ell_i) = 1$, if $\ell_i$ is a
propositional variable, and $\alpha(\ell_i) = 0$, if $\ell_i$ is the
negation of a propositional variable. The $3SAT$ problem asks whether
a given 3-CNF formula $\phi$ is satisfied by at least one of the
assignments for $Vars(\phi)$. This problem is well known to be
NP-complete.

We remind the reader that a decision problem $A$ is NP-hard if there
exists a Karp reduction from every problem in $NP$ to $A$. A Karp
reduction from a decision problem $A$ to a decision problem $B$ is a
polynomial-time algorithm that takes as input an instance $a$ of $A$
and returns an instance $b_a$ of $B$ such that $b_a$ is a positive
instance of $B$ if and only if $a$ is a positive instance of $A$. We
proceed to show one such reduction from $3SAT$ to $\totlprob$, for
each $k\ge0$ and $\circ \in \{<, =, >\}$.
   
To this end, we start by defining a schema
$\cS = \{R_{/7}, W_{/1} \}$. Given a 3-CNF formula
$\phi = \{c_1, \ldots, c_m\}$ with
$Vars(\phi) = \{x_1, \ldots, x_n\}$, we use $R_\phi$ to denote the
instance of $R_{/7}$ that encodes $\phi$ introduced in the proof of
Theorem~\ref{th:shp-hardness}. Observe that, in the two proofs, we use
the same representation for formulae and therefore we can adapt such
construction directly. Moreover, for every $k \ge 0$, we use $W_k$ and
$\idb_k$ to denote the bag containing $k$ occurrences of the tuple
$(0)$, and the instance of $\cS$ that consists of $R_\phi$ and $W_k$,
respectively. Finally, we define the following query:

\[q_{sat} = \pi_{\$2} (COUNT_\emptyset(\varepsilon(\pi_{\$8}
  (\sigma_{(\$1 < \$2)}(R) \cup \sigma_{(\$3 < \$4)}(R) \cup
  \sigma_{(\$5 < \$6)} (R))))) \]

where $\varepsilon$ is a shorthand for duplicate elimination, i.e.,
$\pi_{\$1} (\sum_{\$1}^{\$2} ( Apply_{1}(\ldots)))$.  We proceed to
prove the following property of $q_{sat}$.

\begin{claim}
  \label{cl:invariant-partitions-np}
  Let $f$ be a $(1,2)$-partition for $\nulls(\idb_k)$. Then, for
  each $v,v'$ within $f$ we have
  $q_{sat}(v(\idb_k)) = q_{sat}(v' \idb_k)$
\end{claim}

\begin{proof}
  First, we observe that the expression
  $q_{sel} = \sigma_{(\$1 < \$2)}(R) \cup \sigma_{(\$3 < \$4)}(R) \cup
  \sigma_{(\$5 < \$6)} (R)$ returns the tuples of $v(\idb_k)$ that
  satisfy at least one of the selection conditions. These conditions
  are preserved inside $(1,2)$-partitions (see proof of
  Claim~\ref{cl:invariant-partitions}), and therefore we can conclude
  that $q_{sel}(v(\idb_k)) = q_{sel}(v'(\idb_k))$. The claim follows
  straightforwardly.
 \end{proof}

 Next, we show that the result of $q_{sat}(v(\idb_k))$ can be used to
 check whether $\phi$ is satisfiable. To this end, given a
 $(1,2)$-partition $f$ for $\nulls(\idb_k)$, we use $\mu_f$ to denote
 the assignment for $Vars(\phi)$ such that $\mu_f(x_i) = 1$, if
 $f(\bot_i) \in (1, 2)$, and $\mu_f(x_i) = 0$ if
 $f(\bot_i) \in (0, 1)$. It is easy to see, for every assignment
 $\alpha$ for $Vars(\phi)$ there exists a $(1,2)$-partition for
 $\nulls(\idb_k)$ such that $\alpha = \mu_f$. Next, we prove

 \begin{claim}
   \label{cl:np-hard-qcount}
   if $\# ((m), q_{sat}(v(\idb_k))) = 1)$, for each $v$ within $f$,
   then $\mu_f$ satisfies $\phi$, $\# ((m), q_{sat}(v(\idb_k))) = 0)$,
   otherwise.
 \end{claim}
 \begin{proof}

   From the definition of $q_{sat}$, it is easy to see that
   $\#( (m), q_{sat}(v(\idb_k))=1$ if at least $m$ distinct tuples in
   $R^{\idb_k}$ satisfy one of the conditions $ (\$1 < \$2)$,
   $\$3 < \$4$, or $\$5 < \$6$, $\# ((m), q_{sat}(v(\idb_k))) = 0)$,
   otherwise. We prove the two sides of the claim separately.

   (if part) Suppose that, for each $v$ within $f$, we have
   $\# ((m), q_{sat}(v(\idb_k))) = 1$. Then, there exist $\bc_i$, for
   each $i = 1, \ldots, m$, in $R^{\idb_k}$ that satisfies one of the
   aforementioned conditions. Assume $\bc_i$, for some
   $i = 1, \ldots m$, such that $v(\bc_i)$ satisfy $(\$1 < \$2)$ (the
   other cases are identical). By construction, there exists a clause
   $k_i = \{\ell_1, \ell_2, \ell_3\}$ in $\phi$ such that
   $t(k_i) = \bc_i$. If $\ell_1$ is the variable $x_p$, then
   $t(\ell_1) = ( 1 , \bot_p)$. Since $v(\bc_i)$ satisfies
   $(\$1 < \$2)$, we can conclude that $1 < v(\bot_p)$, for each $v$
   within $f$. Thus, by definition, we have $\mu_f(x_p) = 1$ which, in
   turn, proves that $\mu_f$ satisfies $k_i$. Similar considerations
   allow us to conclude that, if $\ell_1$ is the negated literal
   $\neg x_p$, then $\mu_f(x_p) = 0$, thus proving that $\mu_f$
   satisfies $k_i$. Therefore, we can conclude that $\mu_f$ satisfies
   at least $m$ clauses of $\phi$, and, since $\phi$ consists of $m$
   clauses, it follows that $\mu_f$ satisfies $\phi$.

   (only if part) Assume an assignment $\mu_f$ that satisfies
   $\phi$. Then, $\mu_f$ satisfies all the clauses
   $k_i = \{\ell_1, \ell_2, \ell_3\}$ of $\phi$, for each
   $i = 1, \ldots, m$. By construction, there exists a tuple
   $\bc_i = t(k_i)$ in $R^{\idb_k}$, for each such $i$. Thus, since
   $\mu_f$ satisfies $k_i$, it satisfies at least one of its literals
   $\ell_i$. Consider the case of $\ell_1$ (the other cases are
   identical).  If $\ell_1$ is the variable $x_p$, then
   $t(\ell_1) = ( 1 , \bot_p)$ and, since $\mu_f$ satisfies $k$, we
   have $\mu_f(x_p) = 1$ and $1 < v(\bot_p)$, for each $v$ within
   $f$. Then $v(\bc_i)$ satisfies $(\$ 1 < \$2)$.  Similar
   considerations allow us to conclude that, if $\ell_i = \neg x_p$,
   then $t(\ell_1) = ( \bot_p, 1)$ and $ v(\bot_p) < 1$, and
   $v(\bc_i)$ satisfies $(\$ 1 < \$2)$.  We can conclude that
   $\# ((m), q_{sat}(v(\idb_k))) = 1)$.
 \end{proof}
 
 Let $q_{count}$ be the query that returns the number of tuples in
 $R$, then, from the definition of $\idb_k$, $q_{count}(v(\idb_k))$
 returns the total number of clauses in $\phi$. From now on, we
 analyze the three cases of $\circ \in \{<, =, >\}$ separately.

\underline{Case $\totlunifdef{q, >, k, 0}$} Let $q$ be the following
$\sqlra$ query:

\[q = \pi_{\emptyset}(\sigma_{\$1 = \$2}(q_{sat} \times q_{count}) \cup
  W) \]

Using Claim~\ref{cl:np-hard-qcount}, we can conclude that
$\#((), q(v(\idb_k))) > k$ if and only if $v$ is within a
$(1, 2)$-partition of $Vars(\idb_k)$ such that $\mu_f$ satisfies
$\phi$. Observe now that the set $V_f$ of valuations within $f$ has
probability $\prob_V(V_f) = 0.5^{|Vars(\phi)|} > 0$. Therefore, we can
conclude that $\totlunifdef{q_>, >, k}((), \idb_k) > 0$ if and only if
$\phi$ is satisfiable. The claim follows from $\idb_k$ being
computable in time polynomial in $|\phi|$.

\underline{Case $\totlunifdef{q, <, k, 0}$} Let $q_<$ be the following
$\sqlra$ query:

\[q_< = \pi_{\emptyset}(W) \setminus \pi_{\emptyset}(\sigma_{\$1 =
      \$2}(R \times q_{count})) \]
  Let
  $q' = \pi_{\emptyset}(\sigma_{\$1 = \$2}(q_{sat} \times
  q_{count}))$. Using Claim~\ref{cl:np-hard-qcount}, we can conclude
  that $\#((), q'(v(\idb_k))) = 1$ if $v$ is within a
  $(1, 2)$-partition of $Vars(\idb_k)$ such that $\mu_f$ satisfies
  $\phi$, $\#((), q'(v(\idb_k))) = 0$, otherwise. The claim follows
  from an argument analogous to the one used for the case
  $\totlunifdef{q, >, k}$.

  \underline{Case $\totlunifdef{q, =, k, 0}$} Let $\idb_{k-1}$ be the
  instance of $\cS$ such that $R^{\idb_{k-1}} = R_\phi$, and
  $W^{\idb_{k-1}} = W_{k-1}$. Using Claim~\ref{cl:np-hard-qcount}, we
  can conclude that $\#((), q_>(v(\idb_{k-1}))) = k$ if $v$ is within
  a $(1, 2)$-partition of $Vars(\idb_k)$ such that $\mu_f$ satisfies
  $\phi$, $\#((), q_>(v(\idb_{k-1}))) = k-1$, otherwise. The claim
  follows from an analogous argument to the one used for the case
  $\totlunifdef{q, >, k}$.

\subsection{Proof of Corollary~\ref{cr:no-fpras}}

If there exists an FPRAS for $\otlprob$, one can show that
$NP \subseteq BPP$, where $BPP$ is the class of of decision problems
that can be solved with bounded-error probabilistic polynomial time
algorithm (see, e.g., \cite{complexity-book}). In turn, it is known
that the latter implies that $RP = NP$ (\cite{jerrum-fpras}).

\section{Proofs of Section~\ref{sec:apx}}

\subsection{Proof of Lemma~\ref{lm:db-sampler}}
In what follows, we assume that
$\nulls(\idb) = \{\bot_1, \ldots, \bot_n\}$ and that $\bot_i$ is
defined by $\ndis_i = (\bbR, \bsalg, P_i)$. Moreover, we use
$\cV_\idb = (\pdom_V, \salg_V, \prob_V)$ to denote the valuation space
of $\idb$.

Let now $X$ be the random variable describing the behavior of
$\vsalg$, i.e., $P(X \in \sigma)$ is equal to the probability that the
output of $\vsalg$ falls in the set $\sigma$, for each
$\sigma \in \salg_V$. Let $V(\sigma_1, \ldots, \sigma_n)$ be the set
of valuations such that $v(\bot_i) = \sigma_i$, for each
$i = 1, \ldots, n$. By assumption, each $\nsampalg{i}$ is an
independent random variable, and therefore,
$P(X \in V(\sigma_1, \ldots, \sigma_n)) = \Pi_i P_i(\sigma_i)$, for
each $n$-tuple of sets $\sigma_1, \ldots, \sigma_n \in
\bsalg$. Therefore, $P$ agrees with $\prob_V$ on the generators $\cG$
of $\salg_V$. Moreover, $\cG$ is a $\pi$-system
(Claim~\ref{cl:vals-algebra}), and therefore, from
Proposition~\ref{pr:unique-pie}, it follows that
$P(X \in \sigma) = \prob_V(\sigma)$, for each $\sigma \in \salg_V$.

\subsection{Proof of Lemma~\ref{lm:direct-apx-err}}
Let $\cV = (\pdom_V, \salg_V, \prob_V)$ be the valuation space of
$\idb$. From Lemma~\ref{lm:db-sampler}, we know that, for every
$\sigma_V \in \salg_V$, we have:
\begin{equation}
  P(\vsalg \in \sigma_V) = P_V(\sigma_V) \label{eq:dir-apx-samp}
\end{equation}

From Theorem~\ref{th:measurability}, we know that the set of
valuations $Q = \{v~\mid~ \#(v(\ba), q(v\idb)) \ge n \}$ is measurable
in $\cV$. Applying Equation~\ref{eq:dir-apx-samp} we obtain

\begin{equation}
P(\vsalg \in Q) = P_V(Q) \label{eq:dir-apx-q}
\end{equation}

Let $S_j$, for each $j = 1, \ldots, \gamma$, be the value that is
added to $\mathtt{count}$ in the $j$-th iteration of the
algorithm. Each such $S_j$ is an independent and identically
distributed Bernoulli random variable with
$P(S_j = 1) = P(\vsalg \in Q)$, and therefore the expected value
$E[S_j]$ of $S_j$ is equal to $P(\vsalg \in Q)$. Applying
Equation~\ref{eq:dir-apx-q}, we obtain $E[S_j] = P(Q)$.

We proceed to show that $\mathtt{count}$ can be used to estimate the
value of $E[S_j]$. To this end, let $S$ be a random variable
identically distributed to $S_j$, for each $j = 1, \ldots,
\gamma$. First, we recall the well-known bound on the sum of
independent random variables due to Hoeffding. Let $X_1, \ldots X_n$
be $n$ independent and identically distributed Bernoulli random
variables, and let $X = \sum_{i=1}^n X_n$. Then,
$P(|X - n E[X_n]| \ge t) \le 2e^{-2 t^2n}$. Solving
$ 2e^{-2 t^2n} \le 0.25$ for $n$ we obtain $n \ge t^{-2}$, and
therefore, $P(|X - n E[X_n]| \le t) \ge 0.75 $, for $n \ge
t^{-2}$. Observe now that, after $\gamma$ iterations of the algorithm,
$\mathtt{count}$ is equivalent to $\sum_{i=1}^\gamma S_i$, and
therefore, $P(|\mathtt{count} - n E[S]| \le \err) \ge 0.75 $. To
conclude the proof, we simply recall that $E[S_j] = P(Q)$, for each
$j = 1, \ldots, \gamma$, in turn proving that
$P(|\frac{\mathtt{count}}{n} - P(Q)| \le \err) \ge 0.75$.

\subsection{Proof of Theorem~\ref{th:afpras}}
The fact that $\daalg$ provides the desired approximation is proved in
Lemma~\ref{lm:direct-apx-err}. What is left to show is that the
algorithm runs in time polynomial w.r.t. the input instance and
$\err^{-1}$. Under the assumption that $\idb$ is E.S. , $\vsalg$ runs
in time polynomial in the size of $\idb$. Similarly,
$\mathtt{k} = \lceil \err^{-2} \rceil$ implies that the algorithm runs
polynomially many tests w.r.t. $\err^{-1}$. Finally, we observe that,
given $v$, we can compute $v(\idb)$ and evaluate $q(v(\idb))$ in time
polynomial w.r.t. $\idb$. The claim follows straightforwardly.

\subsection{Proof of Theorem~\ref{th:tuple-query-apx}}

We prove that $\mathtt{apx}_\varepsilon$ has the desired
properties. Assume $V = \{v_1, \ldots, v_\gamma\}$ to be the set of
$\gamma$ outputs of $\vsalg$. We assume that $\idb$ contains the
relation $\mathtt{Rand}^\idb$ such that, for each
$\bot \in \nulls(\idb)$, $\mathtt{Rand}^\idb$ contains the tuple
$(\bot, v_1(\bot), \ldots, v_\gamma(\bot))$. In practice,
$\mathtt{Rand}^\idb$ could also be embedded inside an $\sqlra$ query
without altering the input database but, to ease the presentation, we
assume to have it in $\idb$.

Using $\mathtt{Rand}^\idb$, we proceed to define $\tilde q_v$.  Let
$R \in \schemdb$ be a relation symbol of airty $m$, and let $R_{i,j}$
be the query
\[
  \err(\sigma_{\$1 = \$2}(\sigma_{Const(\$1)}(\pi_{j}(R) \times
  \pi_{j}(R))) \cup \pi_{1,i}(\mathtt{Rand}))\]
where $\varepsilon$ is a shorthand for duplicate elimination, i.e.,
$\pi_{\$1} (\sum_{\$1}^{\$2} ( Apply_{1}(\ldots)))$, and $Const$ is a
predicate that is true over a symbol $c$ if $c \not\in\nulls(\idb)$,
false otherwise.
Intuitively, $R_{i,j}$ defines the $j$-th attribute of $v_i(R^\idb)$,
and the query $\tilde q_V$ will use $R_{i,j}$ to construct
$v_i(R^\idb)$. Let $\theta_j$ be the selection condition
$ {\$j = \$(m+ 2j-1)}$, for $j = 1, \ldots, m$, we define
$\sigma_\theta$ as the query
$\sigma_{\theta_1}(\ldots(\sigma_{\theta_m}(\ldots)))$. The query
$R_i$ is defined as follows:
\[
  \pi_{\$m+2, \$m+4, \ldots, \$m+2m} (\sigma_{\theta}(R \times
  R_{i, 1} \times R_{i, 2} \times \ldots \times R_{i, m}))
\]

We proceed to prove that $R_i(\idb) = v_i(R^\idb)$. To this end,
observe that $R_{i,j}(\idb)$ returns a bag of tuples of arity $2$ of
one of the following two forms: $(n, n)$, where $n$ is a constant
occurring in the $j$-th attribute of $R^\idb$; or $(\bot_i, n)$, where
$n$ is the $i$-th attribute of the tuple $(\bot_i, s_1, \ldots, s_b)$
in $\mathtt{Rand}^\idb$.
Consider now the subquery
$\sigma_{\theta}(R \times R_{i, 1} \times R_{i, 2} \times \ldots
\times R_{i, m})$ of $R_i$. Since $\theta$ is equivalent to the
condition $\bigwedge_{j = 1}^m {\$j = \$(m+ 2j)}$, this query returns
a bag of tuples of arity $3m$ of the form
$\bc = (c_1, \ldots, c_m, c_1, n_1, \ldots, c_m, n_m)$ where the first
$m$ components define a tuple in $R^\idb$, and each pair $(a_j, n_j)$,
for $j = 1, \ldots, m$ is the result of $R_{i,j}(\idb)$. We can
conclude that $v_i(\ba) = (n_1, \ldots, n_m)$. The claim follows from
the fact that $R_i$ returns the projection $(n_1, \ldots, n_m)$ of
$\bc$. 

From now on, we assume $\tilde q_{v_i}$ to be the query obtained from
$q$ by replacing each occurrence of each $R \in \schemdb$ with
$R_i$. It is easy to see that $\tilde q_{v_i} = q(v_i(\idb))$. Thus,
the output of $q_{v_i}$ is a set of tuples of the form $(\bt, k')$
s.t. $\#(\bt, q(v_i(\idb))) = k'$.

Next, we analyze $q_{v, \ba}$. Let $\ba = (a_1, \ldots, a_n)$, and
assume that the left and right endpoints of each $a_i$ are $l_i$ and
$r_i$, respectively. The condition $\att{i} \in v(a_i)$ can be
expressed as $\att{i} \circ_r v(r_i) \wedge \att{i} \circ_l v(l_i)$,
where $\circ_r$ is $<$ if $a_i$ is right-open, $\le$, otherwise, and
where $\circ_l$ is $>$ if $a_i$ is left-open, $\ge$, otherwise. Such
selection condition can be expressed via suitable queries in
$\sqlra$. Then, $q_{v_i, \ba}(\idb)$ returns a set of tuples of the
form $(\bt, k')$ s.t. $\#(\bt, q(v_i(\idb))) = k'$ and $\bt$ is
consistent with $v(\ba)$. Thereofre,
$\sum_{\emptyset}^{\att{1}} q_{v_i,\bar a}$ returns the number of
tuples in $q_{v_i, \ba}(\idb)$ that are consistent with $v_i(\ba)$,
and we can conclude that $\#((), q_{v,\bar a, \circ k}(\idb)) = 1$, if
that number is $\circ k$, $\#((), q_{v,\bar a, \circ k}(\idb)) = 0$,
otherwise.

To prove the claim, we finally analyze the result of
$\mathtt{apx}_\err$. From the definition of the set of valuations
$V = \{v_1, \ldots, v_\gamma\}$,for each $v_j \in V$ and each
$\bot_i \in \nulls(\idb)$ defined by $(\bbR, \bsalg, \prob_i)$, we
have $P(v_j(\bot_i) \in \sigma) = \prob_i(\sigma)$, for each
$\sigma \in \bsalg$. This is due to Lemma~\ref{lm:db-sampler}.

The result of $\mathtt{apx}_\err(\idb))$, then, is equal to a random
variable $\sum_{i = 1}^\gamma X_i$, where each $X_i$ is an independent
and identically distributed Bernoulli random variable. Observe now
that, using the arguments used in the proof of
Lemma~\ref{lm:db-sampler}, we can conclude that
$|P(X_i = 1) - \otlprob(\idb, \ba)| \le \err$. The claim is now a
consequence of the Hoeffding bound (see proof of
Lemma~\ref{lm:direct-apx-err}).

\subsection{Proof of Proposition~\ref{pr:q-size}}
$\mathtt{apx}_\err$ contains $\gamma = \lceil \err^{-2} \rceil$
disjuncts. Each such disjunct can be constructed in time polynomial
w.r.t. $\idb$ and $q$ applying the construction in the proof of
Theorem~\ref{th:tuple-query-apx}. The claim follows straightforwardly.

\subsection{Proof of Theorem~\ref{th:list-query-apx}}

Due to the definition of $\mathtt{compute}_\err$, every tuple
$(\bc, b, p)$ in $eval(\mathtt{apx}^\err_{q}, \idb)$ is such that
$\bc$ occurs $p \cdot \gamma$ times in
$\bigcup_{j=1}^\gamma (q_{v_j}, \idb)$. Therefore, $p$ is the random
variables $\frac{1}{\gamma} \sum_{i =1}^\gamma X_i$, where each $X_i$
is an independent and identically distributed Bernoulli random
variable such that
$|P(X_i = 1) - \otldef{q, =, n}(\idb, \bc)| \le \err$
(Lemma~\ref{lm:db-sampler}). The claim is now a consequence of the
Hoeffding bound (see proof of Lemma~\ref{lm:direct-apx-err}).

\section{Proofs of Section~\ref{sec:finiterepoutput}}

\subsection{Proof of Lemma~\ref{lem:condworldasprob} }

\repeatresult{lemma}{\ref{lem:condworldasprob}}{\lemcondworldasprob}
 \begin{proof}
 
 Let us recall the definition of the probability interpretation of $\condworld$.
 
 The \emph{probabilistic interpretation $(\dom_{\mathcal C},\salg_{\mathcal C}, P_{\mathcal C})$ of a conditional world $\mathcal C \df \{ (\adb_1,c_1), \ldots, (\adb_m,c_m)\}$} is defined by:
\begin{align*}
    \dom_{\mathcal{C}}
    &\df \bigcup_{i=1}^m \{ v(\adb_i)\,\mid\, 
    \dom(v) = \Nulls(\condworld) \}
    \\
     \salg_{\mathcal C} &\df \{ A\,\mid\,\chi^{-1}(A)\in \salg_V\}\\
    %&\df \sigma \{A \, \mid \, A\cap D_i \in \Sigma_i \text{ for every }1\le i\le m \}\\
    P_{\mathcal C}(A) &\df  
    %P_V(\{v\,\mid \, v(\adb_1)\cap A\ne \emptyset , c_1 = \true\}) 
P_V(\chi^{-1}(A))
\end{align*}
where $(\dom_V, \salg_V, P_V)$ is the valuation space of %$\bigcup_{i=1}^k \Nulls(\adb_i)\cup \bigcup_{i=1}^k \Nulls(c_i)$.
$\Nulls(\condworld)$, 
and $\chi : \dom_V \rightarrow \dom_{\condworld}$ is defined by
\[\chi(v) \df \left\{\begin{matrix}
 v(\adb_1)& \text{if } v(c_1) = \true \\ 
 \vdots& \vdots \\ 
 v(\adb_m)& \text{otherwise }% v(c_k) = \true \\ 
\end{matrix}\right.\]

We first show that $\salg_{\condworld}$ is indeed a $\sigma$-algebra. 
First, notice that $\chi^{-1}(\emptyset ) = \emptyset$ and since $\emptyset \in \salg_V$ the it is also the case that $\emptyset \in \salg_{\condworld}$.
Second, 
we wish to show that $\chi^{-1}(\dom_{\condworld}) = \dom_V$. By definition $\chi^{-1}(\dom_{\condworld}) \subseteq \dom_V$, so it suffices to show that the other direction also holds. Let $v\in V$. Due to Definition~\ref{def:conddisj}, there is a unique $i$ for which $v(c_i) = \true$. Thus, $v\in \chi^{-1}(v(A_i))$.
By defintion $\chi^{-1}(\dom_{\condworld}) = \chi^{-1}\left( 
\bigcup_{i=1}^m \{ v(\adb_i)\,\mid\, 
    \dom(v) = \Nulls(\condworld) \}
\right).
$, which allows us to conclude that $\chi^{-1}(\dom_{\condworld}) =\dom_V$.
Third, let $A,B \in \salg_{\condworld}$. It holds that $\chi^{-1}(A\cup B) =\chi^{-1}(A) \cup \chi^{-1}(B)  $. Since $\chi^{-1}(A), \chi^{-1}(B) \in \salg_V$ and $\salg_V $ is a $\sigma$-algebra, we obtain that $A\cup B \in \salg_{\condworld}$.
Last, we can show that $\salg_{\condworld}$ is closed under intersection similarly, which enables us to conclude that $\salg_{\condworld}$ is a $\sigma$-algebra.

We now show that $P_{\condworld}$ is a probability. 
First, 
$P_{\condworld}(\dom_{\condworld}) = P_V(\chi^{-1}(\dom_{\condworld}))$. Since $\chi^{-1}(\dom_{\condworld}) = \dom_V$, and since $P_V$ is a probability,  we can conclude that $P_{\condworld}(\dom_{\condworld}) =1$.
Second, $P_{\condworld}(\emptyset) = P_V(\chi^{-1}(\emptyset))$. By defintion, $P_V(\chi^{-1}(\emptyset)) = \emptyset$, and hence $P_{\condworld}(\emptyset) = 0$.
Last, assume that $A, B$ are disjoint sets of $\dom_{\condworld}$. It holds, by definition,  that $P_{\condworld}(A\cup B) = P_V(\chi^{-1}(A\cup B))$. Since $\chi^{-1}(A\cup B) = \chi^{-1}(A) \cup \chi^{-1}(B)$, we can conclude that $P_{\condworld}(A\cup B) = P_V(\chi^{-1}(A) \cup \chi^{-1}(B))$. Notice that $\chi^{-1}(A) \cap \chi^{-1}(B) = \emptyset$, and since $P_V$ is a probability, we can conclude that  $P_{\condworld}(A\cup B) = P_V(\chi^{-1}(A) )+P_V( \chi^{-1}(B)) $. This implies, by definition, that  $P_{\condworld}(A\cup B) = P_{\condworld}(A) +P_{\condworld}(B)$.
 \end{proof}

\subsection{Proof of Lemma~\ref{lem:conddisj} }
\repeatresult{lemma}{\ref{lem:conddisj}}{\lemconddisj}

\begin{proof}
Let us first recall the definition of conditional worlds:

\defcondworld

Let us denote $\condworld \df \{(\adb_1,c_1),\ldots, (\adb_m, c_m)\}$.
We prove the claim by induction on the structure of $q$, we distinguish between three cases:

If $q$ is Unary and different than $\sigma_{\att{i} < \att{j}}$, then since in the result of applying $q$ there are exactly the same conditions, the claim holds.

If $q$ is Binary, i.e., $q \in \{\times, \cup, \setminus \}$, then using the notation
$\condworld \df \{ (\adb_1,c_1),\ldots,(\adb_m,c_m)\}$ and $\condworld' \df \{ (\adb'_1,c'_1),\ldots,(\adb'_{m'},c'_{m'})\}$, it suffices to show that 
\begin{itemize}
    \item[(1)]
$\set{v \mid v(c_i \cup c'_j) = \true} \cap  \set{v \mid v(c_{i'} \cup c'_{j'}) = \true} = \emptyset $ whenever $i \ne i'$ or $j\ne j'$ 
\item[(2)]
$P_V\left(\set{v \mid \vee_{1\le i\le m, 1\le j\le m'}v(c_i \cup c'_j )= \true}\right) = 1$ 
\end{itemize}

For (1), be definition of $v(c\cup c') = \true$ we have \[ \set{v \mid v(c_i \cup c'_j) = \true \wedge v(c_{i'} \cup c'_{j'}) = \true 
}  = 
 \set{v \mid v(c_i) = \true \wedge  v(c'_j) = \true \wedge  v(c_{i'}) = \true  \wedge  v(c'_{j'}) = \true 
}.\]
Since $ \set{v \mid v(c_i) = \true} \cap \set{ v\mid   v(c_{i'}) = \true } = \emptyset $, we can conclude the desired claim.  

For (2), notice that by definition, 
$\{v\mid v(c_i) = \true \} = 
\{ v\mid v(c_i) = \true \wedge v(c'_1) = \true \}
\cup \cdots\cup 
\{ v\mid v(c_i) = \true \wedge v(c'_{m'}) = \true \}
$ for every $i$.
Again, by definition of $v(c) = \true$, we have 
$\{v\mid v(c_i) = \true \} = 
\{ v\mid v(c_i\cup c'_1) = \true \}
\cup \cdots\cup 
\{ v\mid v(c_i\cup c'_{m'}) = \true \}$ for every $i$.
Since $\condworld$ is a conditional world, we have 
$P_V\left(\cup_{i=1}^m\left(\{ v\mid v(c_i\cup c'_1) = \true \}
\cup \cdots\cup 
\{ v\mid v(c_i\cup c'_{m'}) = \true \}\right)\right) = 1.
$

Finally, if $q$ is of the form $\sigma_{\att{i}<\att{j}}$ then it we need to show:
\begin{itemize}
\item[(1)]
$\set{v \mid v(c_i \cup A') = \true \wedge  v(c_{i'} \cup A'') = \true 
}= \emptyset $ whenever $i \ne i'$ or $A'\ne A''$
\item[(2)]
$P_V\left(\set{v \mid \vee_{1\le \ell \le m, A'\subseteq C(i,j)} v(c_{\ell} \cup A' )= \true}\right) = 1$
\end{itemize}

For (1), note that
each $A'\ne A''$ has at least one $f$ such that $f\in A'$ and $-f \in A''$. Thus, the intersection $A'\cap A'' = \emptyset$ whenever $A'\ne A''$.
Using a similar proof to the  claim (1) in the Binary case we obtain the desired result. 
\OMIT{
$\set{v \mid \vee_{1\le \ell \le m, A'\subseteq C(i,j)} v(c_{\ell} \cup A' )= \true} \supseteq
\set{v \mid \vee_{1\le \ell \le m} v(c_{\ell}  )= \true} 
$ since $A'$ can be the empty set.
Since $P_V$ is a probability function we have 
 $P_V\left(\set{v \mid \vee_{1\le \ell \le m, A'\subseteq C(i,j)} v(c_{\ell} \cup A' )= \true} \right)\ge
P_V\left(\set{v \mid \vee_{1\le \ell \le m} v(c_{\ell}  )= \true}\right) 
$. It holds that $P_V\left(\set{v \mid \vee_{1\le \ell \le m} v(c_{\ell}  )= \true}\right) =1$, and therefore we can conclude the desired claim. 
}

For (2), the claim is shown similarly to claim (2) for Binary queries.
\end{proof}

 \subsection{Proof of Theorem~\ref{thm:condandoutput}}
 \repeatresult{theorem}{\ref{thm:condandoutput}}{\thmcondandoutput}

 We prove a stronger claim:
 \begin{lemma}
 For every conditional world $\condworld$ 
 and query $q$, if for every $\bot_i$ in $\Nulls(\condworld)$ it holds that $P_i(\bot_i = c) = 0 $ for every constant $c\in\mathbb{R}$ then the answer space of $q$ over the probabilistic interpretation $(\dom_{\condworld}, \salg_{\condworld},P_{\condworld})$ of $\condworld$ is a trivial extension of the probabilistic interpretation of $q(\condworld)$.
 \end{lemma}
 
 \begin{proof}
The proof of this lemma is by induction on the structure of $q$.
We denote the answer space of $q$ over $(\dom_{\condworld}, \salg_{\condworld},P_{\condworld})$ by $(\dom_{q,\condworld}, \salg_{q,\condworld},P_{q,\condworld})$, 
and the probabilistic interpretation of $q(\condworld)$
by $(\dom_{q(\condworld)}, \salg_{q(\condworld)},P_{q(\condworld)})$.
We also denote  $\condworld \df \set{(\adb_1,c_1),\ldots, (\adb_m, c_m) }$. (Notice that, by definition, it holds that $q(\condworld) = 
 \set{ (q(\adb_1),c_1),\ldots, (q(\adb_m), c_m)}$. 
)

 \subsubsection*{$q$ is Unary and  different from $\sigma_{\att{i} < \att{j}}$}
 
 We denote, 

 \begin{tikzcd}
(\dom_V, \salg_V,P_V) 
\arrow[rd, "\tilde{\chi}"] \arrow[r, "\chi"] & (\dom_{\condworld},\salg_{\condworld},P_{\condworld})
\arrow[r, "q"] & (\dom_{q,\condworld},\salg_{q,\condworld},P_{q,\condworld}) \\
& (\dom_{q(\condworld)},\salg_{q(\condworld)},P_{q(\condworld)})&
\end{tikzcd}
where $\chi(v) \df \left\{\begin{matrix}
 v(\adb_1)& \text{if } v(c_1) = \true \\ 
 \vdots& \vdots \\ 
 v(\adb_m)& \text{otherwise }
 \end{matrix}\right.$.

Our goal is to show that $(\dom_{q(\condworld)},\salg_{q(\condworld)},P_{q(\condworld)})$ and $(\dom_{q,\condworld},\salg_{q,\condworld},P_{q,\condworld})$ are similar. Note that $$\tilde\chi(v) \df \left\{\begin{matrix}
 v(q(\adb_1))& \text{if } v(c_1) = \true \\ 
 \vdots& \vdots \\ 
 v(q(\adb_m))& \text{otherwise }
 \end{matrix}\right.$$

 By definition, it holds that the domain of $\dom_{q(\condworld)}$ is similar to that of $\dom_{q,\condworld}$. 
 Due to definition, we denote
  
 We now show that the $\sigma$-algebras are similar. 
 Let $A\in \salg_{q(\condworld)}$. We show that \[v\in \tilde{\chi}^{-1}(A) \text{ if and only if  }v\in 
 \chi^{-1}\left(q^{-1}(A)\right)\]
 Without loss of generality, we assume that $v(c_i) = \true$. Thus, $v\in \tilde{\chi}^{-1}(A)$ iff there is $a\in A $ such that $\tilde\chi(v) = a$. Since we have the assumption on $v$, the former holds if and only if there is $a\in A $ such that $a\in v\left(q(\adb_i)\right)$. By definition, $v$ and $q$ are interchangeable and thus we have that  the former holds if and only if there is $a\in A $ such that $a\in q\left(v(\adb_i)\right) = q\left(\chi(v)\right) $. By definition, this holds iff $v\in \chi^{-1}(q^{-1}(A))$, which completes the claim, and allows us to conclude that $A\in \salg_{q,\condworld}$.
 We can show similarly that $A\in \salg_{q,\condworld}$ implies $A\in \salg_{q(\condworld)}$.
 Notice that this implies that for each set $A$ $\sigma$-algebra, it holds that $ \tilde{\chi}^{-1}(A) =
 \chi^{-1}\left(q^{-1}(A)\right)$ and therefore we can conclude that $P_{q(\condworld)}(A) = P_{q,\condworld}(A)$.

\OMIT{
 Note that $\dom_{q(\condworld)}$ can be written as the disjoint union \[\{v(q(\adb_1)) \mid v(c_1) = \true \} \cup \cdots \cup\{v(q(\adb_m)) \mid v(c_1) = \false \wedge \cdots \wedge v(c_{m-1}) = \false \},\] and thus
 we can write $A$ also as the disjoint union 
 \[(A\cap \{v(q(\adb_1)) \mid v(c_1) = \true \} )\cup \cdots \cup(A \cap \{v(q(\adb_m)) \mid v(c_1) = \false \wedge \cdots \wedge v(c_{m-1}) = \false \} )\]
 Let us denote $B_i \df A\cap \{v(q(\adb_1)) \mid v(c_1) = \true \}$.
 Since $\tilde{\chi}^{-1}(A)$ is in $\salg_V$ and since by its definition $\{v(q(\adb_i)) \mid v(c_i) = \true \}$ is in $\salg_V$, we can conclude that each $B_i$ is in $\salg_V$.
 
 On the other hand, due to the definition of $q$, we have that $q$ and $v$ are interchangeable and thus $B_i=A\cap \{q(v(\adb_i)) \mid v(c_i) = \true \}$ and thus $A = 
 (A\cap \{q(v(\adb_1)) \mid v(c_1) = \true \} )\cup \cdots \cup(A \cap \{q(v(\adb_m)) \mid v(c_m) = \true \} )$.
 Each $A\cap \{q(v(\adb_i)) \mid v(c_i) = \true \}$ is measurable since each $B_i$ is measurable, which completes the proof.
 }

 \subsubsection*{$q$ is Binary}
 In this case we are dealing with the operators $\{\times, \cup, \setminus \}$. For convenience, for every $\circ\in \{\times, \cup, \setminus \}$ we write $q^{\circ}(A, B)$ to denote $A\,\circ\,B$. Since the treatment for all of the operators is similar, in what follows,  we just use the notation $q(A,B)$.
 
Let $\condworld = \{ (\adb_1,c_1),\ldots, (\adb_m,c_m) \}$ and $\condworld' = \{  (\adb'_1,c'_1),\ldots, (\adb'_m,c'_n) \}$. We have

\begin{tikzcd}
(\dom_V, \salg_V,P_V) 
 \arrow[rd, "{\chi'}"] \arrow[r, "\chi"] & (\dom_{\condworld},\salg_{\condworld},P_{\condworld})
 \\
& (\dom_{\condworld'},\salg_{\condworld'},P_{\condworld'})
\end{tikzcd}
where $\chi(v) \df \left\{\begin{matrix}
 v(\adb_1)& \text{if } v(c_1) = \true \\ 
 \vdots& \vdots \\ 
 v(\adb_m)& \text{otherwise }
 \end{matrix}\right.$ and $\chi(v) \df \left\{\begin{matrix}
 v(\adb'_1)& \text{if } v(c'_1) = \true \\ 
 \vdots& \vdots \\ 
 v(\adb'_n)& \text{otherwise }
 \end{matrix}\right.$

\begin{tikzcd}
(\dom_{\condworld},\salg_{\condworld},P_{\condworld}),(\dom_{\condworld'},\salg_{\condworld'},P_{\condworld'})
\arrow[r, "q"] 
& (\dom_{q,\condworld,\condworld'},\salg_{q,\condworld,\condworld'},P_{q,\condworld,\condworld'})
\end{tikzcd}
where $q$ is defined in Definition~\ref{def:semcondworld}.

\begin{tikzcd}
(\dom_V, \salg_V,P_V) 
\arrow[r, "\tilde \chi"] 
& (\dom_{q(\condworld,\condworld')},\salg_{q(\condworld,\condworld')},P_{q(\condworld,\condworld')})
\end{tikzcd}
where $\tilde \chi (v) \df v(q(\adb_i,\adb'_j))$ whenever $v(c_i) = \true , v(c'_j)$.

By definition, it holds that the domains are similar. 

Very similarly to Unary $q$, we show
 We now show that the $\sigma$-algebras are similar. 
 Let $A\in \salg_{q(\condworld,\condworld')}$. We show that \[v\in \tilde{\chi}^{-1}(A) \text{ if and only if  }v\in (\chi^{-1}(B_1),\chi^{'-1}(B_2))\text{ where  }q^{-1}(A) = (B_1,B_2) \]
 Without loss of generality, we assume that $v(c_i) = \true$ and $v(c'_j) = \true$ (which implies that there are no other $i',j'$ for which $v(c_i) = \true$ and $v(c'_j) = \true$). 
 Thus, $v\in \tilde{\chi}^{-1}(A)$ iff there is $a\in A $ such that $\tilde\chi(v) = a$. Since we have the assumption on $v$, the former holds if and only if there is $a\in A $ such that $a\in v\left(q(\adb_i,\adb'_j)\right)$. 
 By definition, $v$ and $q$ are interchangeable and thus we have that  the former holds if and only if there is $a\in A $ such that $a\in q\left(v(\adb_i), v(\adb_j)\right) = q\left(\chi(v),\chi'(v)\right) $. 
 By definition, this holds iff $v\in (\chi^{-1}(B_1),\chi^{'-1}(B_2))$ where $q^{-1}(A) = (B_1,B_2)$ which completes the claim, and allows us to conclude that $A\in \salg_{q,\condworld}$.
 We can show similarly that $A\in \salg_{q,\condworld}$ implies $A\in \salg_{q(\condworld)}$.
 Notice that this implies that for each set $A$ $\sigma$-algebra, it holds that $ \tilde{\chi}^{-1}(A) =
 (\chi^{-1}(B_1),\chi^{'-1}(B_2))$ where $q^{-1}(A) = (B_1,B_2)$ and therefore we can conclude that $P_{q(\condworld)}(A) = P_{q,\condworld}(A)$.
 
 \subsubsection*{$q$ is of the form $\sigma_{\att{i}<\att{j}}$:}
 In this case, we use a proof that is very similar to the Unary case but we have to take care of the valuations that do not satisfy any of the conditions. Luckily, their measure is zero so they do not affect dramatically the situation. We now formalize the proof. 
 
  We denote, 

 \begin{tikzcd}
(\dom_V, \salg_V,P_V) 
\arrow[rd, "\tilde{\chi}"] \arrow[r, "\chi"] & (\dom_{\condworld},\salg_{\condworld},P_{\condworld})
\arrow[r, "q"] & (\dom_{q,\condworld},\salg_{q,\condworld},P_{q,\condworld}) \\
& (\dom_{q(\condworld)},\salg_{q(\condworld)},P_{q(\condworld)})&
\end{tikzcd}
where $\chi(v) \df \left\{\begin{matrix}
 v(\adb_1)& \text{if } v(c_1) = \true \\ 
 \vdots& \vdots \\ 
 v(\adb_m)& \text{otherwise }
 \end{matrix}\right.$.

Our goal is to show that $(\dom_{q(\condworld)},\salg_{q(\condworld)},P_{q(\condworld)})$ and $(\dom_{q,\condworld},\salg_{q,\condworld},P_{q,\condworld})$ are similar. 

Recall that $
\semcondworld{\sigma_{\att{i}<\att{j}}} \df 
\{(\sigma_{B}(\adb), c \cup c_B \,\mid\,
(\adb,c)\in \condworld,\, 
B\subseteq \condworld(i,j) \} $ and therefore
$\tilde\chi(v) \df \sigma_B(\adb)$ whenever $v(c) = \true, v(c_B) = \true$.

Note that, by definition, $\dom_{q,\condworld} \supseteq  \dom_{q(\condworld)}$. Moreover, the difference $\dom_{q,\condworld} \setminus  \dom_{q(\condworld)}$ consists of selections of the form  $\sigma_\theta(\adb)$ such that there is $c$ where $(\adb,c)\in\condworld$, and $\theta$ can be written as $f = 0 \wedge \theta'$ for $f\in \condworld(i,j)$. We call the union of all possible such sets $\mathcal Z$.

 Recall that we assumed that each null in $\Nulls(\idb)$ has a singleton-unlikely distribution. Thus, for each set  $Z\subset \mathcal{Z}$, it holds that $P_V\left(\chi^{-1} \left(q^{-1} (Z)\right)\right) = 0$.
 Therefore, each set $A\in \salg_{q,\condworld}$ can be written as a union of $((A \setminus \mathcal Z )\cup  (A \cap \mathcal Z ))$.
 Using similar arguments as before, we can show that $A \setminus \mathcal Z \in \salg_{q(\condworld)}$. On the other hand, if $A\in \salg_{q(\condworld)}$ then for every $B\subseteq \mathcal Z$ it holds that $A\cup B\in \salg_{q,\condworld}$. 
 
 This allows us to conclude that indeed $\salg_{q,\condworld} \supseteq \salg_{q(\condworld)}$, and that the probabilities behave as stated.
  \end{proof}
  
  The proof of the theorem is a direct consequence of the above lemma based on the fact that the probabilistic interpretation of $\{(\idb,-1) \}$ equals the PDB of $\idb$.
\end{document}